\documentclass[11pt]{article}

\usepackage[margin=1in]{geometry}
\usepackage{amsmath, amssymb, amsthm}
\usepackage{setspace}
\usepackage{natbib}
\usepackage{booktabs}
\usepackage{rotating}
\usepackage{caption}
\usepackage{subcaption}
\usepackage{graphicx}
\usepackage{algorithm}
\usepackage{algpseudocode}

\newtheorem{assumption}{Assumption}
\newtheorem{proposition}{Proposition}
\newtheorem{theorem}{Theorem}
\newtheorem{corollary}{Corollary}
\newtheorem{lemma}{Lemma}
\theoremstyle{definition}

\newcommand{\E}{\mathbb{E}}

\newcommand{\att}{\ensuremath{\mathrm{ATT}}}
\newcommand{\atu}{\ensuremath{\mathrm{ATU}}}
\newcommand{\ind}{\mathrm{I}}

\usepackage{hyperref}
\hypersetup{colorlinks=true, linkcolor=blue, citecolor=blue, urlcolor=blue}

\newcommand{\SfiveAsevenWald}{4.4}
\newcommand{\SfiveAsevenP}{0.11}

\newcommand{\SfiveAninePA}{0.92}

\newcommand{\SfiveAninePB}{0.62}
\newcommand{\SfiveAnineWaldJoint}{0.4}
\newcommand{\SfiveAninePJoint}{0.81}

\newcommand{\SfiveAsevenPBC}{0.055}

\newcommand{\SfiveAtenWald}{2.0}
\newcommand{\SfiveAtenP}{0.360}

\newcommand{\SfiveAtenPBC}{0.452}

\newcommand{\SfiveSwAboveBelow}{35}
\newcommand{\SfiveSwBelowAbove}{102}

\newcommand{\SfiveSwTotal}{137}
\newcommand{\SfiveSwMed}{258}
\newcommand{\SfiveNonSwMed}{853}
\newcommand{\SfiveMcKeep}{0.365}
\newcommand{\SfiveMcOmit}{0.045}
\newcommand{\SfiveNmuni}{1{,}122}
\newcommand{\SfiveAsevenJumpZZ}{-0.22}
\newcommand{\SfiveAsevenJumpZZse}{0.12}

\newcommand{\SfiveAsevenJumpZF}{-0.17}
\newcommand{\SfiveAsevenJumpZFse}{0.14}
\newcommand{\SfiveAnineDAbelow}{25.8}
\newcommand{\SfiveAnineDAbelowSE}{60.6}
\newcommand{\SfiveAnineDAabove}{19.0}
\newcommand{\SfiveAnineDAaboveSE}{25.5}
\newcommand{\SfiveAnineDBbelow}{76.0}
\newcommand{\SfiveAnineDBbelowSE}{74.7}
\newcommand{\SfiveAnineDBabove}{37.7}
\newcommand{\SfiveAnineDBaboveSE}{20.0}
\newcommand{\SfiveAtenPabove}{0.359}
\newcommand{\SfiveAtenPbelow}{0.273}

\newcommand{\SfiveAeightJump}{-0.37}
\newcommand{\SfiveAeightJumpSE}{0.11}
\newcommand{\SfiveAeightJumpP}{0.001}
\newcommand{\SfiveAeightJumpBC}{-0.36}
\newcommand{\SfiveAeightJumpBCse}{0.14}
\newcommand{\SfiveAeightJumpPBC}{0.011}
\newcommand{\SfiveAsevenPZZ}{0.053}
\newcommand{\SfiveAsevenPZF}{0.221}

\newcommand{\SsixDA}{56.0}
\newcommand{\SsixDB}{81.8}

\newcommand{\SsixDF}{-61.8}

\newcommand{\SsixATUConstCctF}{-130.7}\newcommand{\SsixATUseConstCctF}{43.1}
\newcommand{\SsixhConstCctFF}{390}
\newcommand{\SsixhConstCctFA}{445}
\newcommand{\SsixhConstCctFB}{461}
\newcommand{\SsixATUConstJointF}{-130.3}\newcommand{\SsixATUseConstJointF}{40.5}\newcommand{\SsixATUbcConstJointF}{-144.9}\newcommand{\SsixATUbcseConstJointF}{46.5}
\newcommand{\SsixhConstJointFF}{460}

\newcommand{\SsixATULinCctF}{-247.1}\newcommand{\SsixATUseLinCctF}{83.9}

\newcommand{\SsixLinAlphaF}{185.3}
\newcommand{\SsixConstAlpha}{68.9}
\newcommand{\SsixDslope}{25.9}

\newcommand{\SsixAvgConstCct}{-121.9}\newcommand{\SsixAvgseConstCct}{36.3}

\newcommand{\SsixAvgLinCct}{-199.5}\newcommand{\SsixAvgseLinCct}{60.3}

\newcommand{\SsixMVhOnew}{0.88}\newcommand{\SsixMVhOnewbc}{0.85}\newcommand{\SsixMVhOneATU}{-120.3}\newcommand{\SsixMVhOneATUse}{40.2}\newcommand{\SsixMVhOneATUbc}{-134.8}\newcommand{\SsixMVhOneATUbcse}{46.2}
\newcommand{\SsixItHone}{459.8}
\newcommand{\SsixItWtwo}{0.876}
\newcommand{\SsixItWtwoB}{0.124}
\newcommand{\SsixItHthree}{480.8}
\newcommand{\SsixItWfour}{0.896}
\newcommand{\SsixItWfourB}{0.104}
\newcommand{\SsixItWlim}{0.897}

\newcommand{\SsixItHlim}{482.7}
 
\title{Correcting Invalid Regression Discontinuity Designs \\ Using Multiple Time-Period Data}
\author{Dor Leventer \and Daniel Nevo\thanks{Leventer: UCL (University College London), d.leventer@ucl.ac.uk. Nevo: Department of Statistics and Operations Research Tel-Aviv University, danielnevo@tauex.tau.ac.il. We thank Itay Saporta-Eksten, Analia Schlosser, Yoav Goldstein, Kirill Borusyak, Xavier Jaravel and seminar participants at EUROCIM 2024, Tel Aviv University and Hebrew University for their helpful comments and suggestions. This work was supported by a grant from the Tel Aviv University Center for AI and Data Science (TAD). DL gratefully acknowledges financial support from The Israel Pollak Fellowship Program for Excellence.}}
\date{\today}

\begin{document}
\maketitle

\begin{abstract}
The difference-in-discontinuities (diff-in-disc) design identifies the effect of a treatment of interest when a confounding treatment is assigned at the same cutoff, by assuming the confounding discontinuity is constant over time. Existing identification theory covers only two periods with a running variable fixed over time. Allowing for multiple periods, we derive diff-in-disc identification when the confounding discontinuity has a known time trend. When the running variable changes over time, even in the two period case, units can move across the cutoff, so the mix of units at the cutoff can change within a period and across periods. We discuss assumptions that rule out these composition effects, so diff-in-disc identifies the treatment effect, and derive validation tests in datasets where units are observed over time. Our proposed estimation approach runs a separate RD in each period and aggregate the estimates with weights that difference out the confounding discontinuity under the assumed trend. We derive its variance under different sampling schemes, in addition to an MSE-optimal bandwidth common to all periods, and variance-minimizing weights.
\end{abstract}

\section{Introduction}

Regression discontinuity (RD) is a widely used design for causal inference that exploits an observable discontinuity in assignment of a treatment of interest at a cutoff \citep{imbens2008regression, lee2010regression}. A common identification assumption in RD is that the mean potential outcomes are continuous at the cutoff \citep{hahn2001identification}. This assumption may fail when a confounding treatment is assigned by the same running variable and cutoff. 
To overcome such issues, researchers implement a difference-in-discontinuities (diff-in-disc) design, which differences the discontinuities between a time period where the treatment of interest switches at the cutoff (RD period) and a time period where the treatment of interest is uniform across units (comparison period), while the confounding treatment switches at the cutoff in both periods \citep{grembi2016fiscal}. The diff-in-disc design is implemented in a variety of applied contexts \citep[e.g.,][]{debenedetto2019term,corekcioglu2021unveiling,aslim2022public,marattin2022revenue,alonso2026drinking}. However, the theoretical foundations are limited to two time periods, a running variable that is fixed over time, and a confounding discontinuity that is constant over time, and the common estimation approach is to pool all periods in a single regression without implementing contemporary RD estimation methods.

To set ideas, we first introduce the two-period diff-in-disc design. We consider as example \citet{grembi2016fiscal}, where the units are Italian municipalities, the treatment and outcome of interest are fiscal rules and outcomes, respectively, and the confounding treatment is higher salaries for the mayor and executives, and both change discontinuously at the 5,000-resident cutoff. The RD assignment of the fiscal rules started in 2001, so in 2000, the comparison period, the rules were uniform across municipalities, while salaries changed at 5,000 residents in both years. In the RD period, 2001, the observed outcome discontinuity at the cutoff is the sum of two terms. The first is the treatment effect, the local-to-the-cutoff difference in potential outcomes when the treatment of interest switches and the confounding treatment is held fixed. In the example, this is the effect of the fiscal rules on municipalities at 5,000 residents, holding salaries fixed. The second is the confounding discontinuity, the discontinuity in potential outcomes at the cutoff when the confounding treatment switches and the treatment of interest is held fixed. In the example, this is the effect of the higher salaries, holding the fiscal rules fixed. In the comparison period, the observed discontinuity equals the confounding discontinuity of that period, as all municipalities faced the same fiscal rules in 2000. Hence, if the confounding discontinuity is constant over time, the difference between the two discontinuities, the diff-in-disc, identifies the treatment effect in the RD period.

In this paper we develop both identification and estimation theory for diff-in-disc designs with sharp RD assignments. Regarding identification, we extend the theory to more than two time periods and to a running variable that changes over time. A time-varying running variable introduces composition effects, and we provide assumptions under which the diff-in-disc recovers the treatment effect in such settings, and derive tests for these assumptions. For estimation, we develop an approach that estimates an RD separately in each period and aggregates the estimates according to the identification assumption, and derive its variance, optimal bandwidth, and variance-minimizing weights. We use the application of \citet{grembi2016fiscal} as a running example throughout.

As a first extension, we consider a diff-in-disc setting with several comparison periods and several RD periods. 
Allowing potential outcomes to depend on treatments in other periods introduces effects similar to anticipation as discussed in the DID literature, which we rule out with an additional local-to-the-cutoff no-anticipation-type assumption for identification. 
Furthermore, we generalize the constant confounding discontinuity assumption in two ways: imposing the assumption across several periods, for example all comparison periods, and allowing time trends. 
Returning to the \citet{grembi2016fiscal} example, it has two comparison periods, 1999 and 2000, and four RD periods, 2001--2004. One can assume that the confounding discontinuity, the effect of the higher salaries on fiscal outcomes holding fiscal rules fixed, is not constant over time but rather has a time trend, such as a linear time trend. Under the assumption that the confounding discontinuity has a constant or a linear time trend, we show that the treatment effect of interest in a single RD period is identified by the observed discontinuity in that period minus a weighted sum of the observed discontinuities in the comparison periods. Finally, we discuss how across several RD periods the average effect is also identified, similar to the overall treatment effect in DID \citep{callaway2021difference}. In the example, this is the average of the treatment effects in 2001--2004.

Next, we extend the diff-in-disc setting by allowing the running variable to change over time. In the example, the number of residents in a municipality changes over time, e.g., due to fertility or migration. When the running variable changes over time, units can move across the cutoff between time periods, which introduces composition effects. Consider the case of a single comparison period (2000 in the example) and RD period (2001 in the example). The share of units that are above the cutoff in 2001 can change discontinuously across the cutoff in 2000, and vice versa. Such discontinuities, in the composition of units by treatment in other periods, can arise due to various mechanisms. For intuition, we give two examples. First, the confounding treatment may have an effect on the running variable. In the example, higher salaries may improve city outcomes and in turn attract more residents, so municipalities just above 5,000 in 2000 are more likely to stay above it in 2001. Second, the gains from being treated may cause units to try to modify their running variables \citep{hsu2024dynamic}, for example, municipalities just below 5,000 in 2000 may act to cross it by 2001 to obtain the higher salaries. Such behaviors induce sorting of units at the cutoff \citep{lee2010regression}, which, if these units differ in outcomes, may bias the RD even when there is no confounding discontinuity. Returning to the diff-in-disc, the design compares the confounding discontinuity across the two time periods. This introduces between-period composition effects. For example, a change in the composition of units just above the cutoff between 2000 and 2001 can cause the confounding discontinuity to change over time, if the confounding discontinuity differs between units that are above and below the cutoff in the other period. 

We  derive the assumptions under which  diff-in-disc identifies the treatment effect when the running variable changes over time, focusing on the two-period case. For within-period composition effects, the new assumption is that, within each period (e.g., 2000), there is no discontinuity at the cutoff in the share of units that are above the cutoff in the other period (e.g., 2001). This rules out the two mechanisms described above. For between-period composition effects, it suffices that either (i) among units just above the cutoff, the share that is above it in the other period is the same in both periods, or (ii) in the comparison period, the confounding discontinuity does not vary with the side of the cutoff in the other period. Together with a constant-confounding-discontinuity-type assumption, these imply that the diff-in-disc identifies the treatment effect by the same difference in discontinuities as with a running variable that is fixed over time. In panel data, where each unit's side of the cutoff is observed in both periods, the composition assumptions can be tested with standard RD tools. To test the within-period assumption, we run in each period an RD on that period's running variable, with an indicator for being above the cutoff in the other period as the outcome, and test for a discontinuity at the cutoff. In the example, the test in 2000 uses as the outcome whether the municipality has at least 5,000 residents in 2001. A density test such as that of \citet{mccrary2008manipulation} cannot replace this test, because it asks how many units are at the cutoff, while our test asks who those units are.
To test the between-period assumption, we compare two shares: the share of municipalities just above 5,000 residents in 2000 that are also above 5,000 in 2001, and the share of municipalities just above 5,000 in 2001 that were also above it in 2000. If the composition at the cutoff is stable over time, the two shares are equal. 
We test the equality of these shares by creating a stacked dataset, where the cutoff is normalized to zero, municipalities above the cutoff in 2000 enter at minus their distance from the cutoff, municipalities above the cutoff in 2001 at plus their distance, and the outcome is an indicator for being above the cutoff in the other year. We show that a test for a discontinuity in this stacked data tests whether the shares differ at the cutoff across time.

Having established identification, we turn to estimation. Most applications estimate the diff-in-disc with a single regression that pools all periods \citep[e.g.,][]{grembi2016fiscal,avdic2018modern,baltrunaite2019let,bertrand2021improving,bennedsen2022firms}, which does not accommodate contemporary RD methods such as bias correction or MSE-optimal bandwidths. A procedure that does accommodate these methods is to run an RD on outcomes differenced between the RD and comparison periods \citep[e.g.,][]{bagues2021can}, as formalized by \citet{picchetti2024difference}. However, this approach is limited to panel data with a running variable that is fixed over time and a single comparison period. We instead estimate the discontinuity in each period separately and combine the estimates with the weights implied by the identification assumption, which accommodates several periods, time trends in the confounding discontinuity, both repeated cross sections and panels, and contemporary RD methods. For the resulting conventional and bias-corrected estimators, we then show how the variance depends on the sampling scheme. In a repeated cross section the period estimates are independent, while in a panel they are correlated. In a panel with a time-varying running variable, however, the between-period covariances vanish asymptotically, because units move out of the shrinking bandwidth, so to first order the variance is that of a repeated cross section.

Two further choices matter for estimation. First, choosing the bandwidth period by period does not minimize the MSE of the diff-in-disc estimator, which aggregates several periods. We instead derive a single bandwidth, common to all periods, that minimizes this MSE, and which depends on the sampling scheme through the variance. Second, the identification assumption admits many comparison-period weights. Under a constant trend, for example, any weights that sum to one identify the effect, but they give different estimators. We derive the weights that minimize the variance, which in a repeated cross section under a constant trend are inverse-variance weights. We show that, with a common bandwidth, estimating these weights does not affect estimation and inference to first order, so inference can proceed as if the weights were not estimated. The reason is that all admissible weights give the same combination of the true comparison-period discontinuities, so an error in the weights multiplies only the estimation error of the discontinuities. Finally, we extend these tools to the average effect across several RD periods.

We illustrate the theory with the application of \citet{grembi2016fiscal}. The application shows how the estimate changes when moving from a constant to a linear trend, how to choose the bandwidth and the comparison-period weights, and how the validation tests guide which identification assumption to rely on. Composition stability is rejected, as \SfiveSwBelowAbove{} of the \SfiveNmuni{} municipalities moved from below the cutoff to above it over time, which can naturally happen if municipalities grow on average. Since these switchers were initially close to the cutoff, dropping them, a tempting response, creates a drop in the density of the running variable at the cutoff, resembling a donut-hole RD \citep{barreca2011saving,bajari2011regression}. All estimates and tests are produced using the open source \href{https://dorleventer.github.io/didisc/}{\texttt{didisc}} R package.

This paper contributes to the literature on identification in RD designs with a confounding treatment. \citet{grembi2016fiscal} formalize the diff-in-disc design for two periods with a constant confounding discontinuity, and \citet{picchetti2024difference} and \citet{butts2023geographic} consider a single RD on differenced outcomes, which holds the running variable fixed. \citet{picchetti2024difference} also propose a test of a constant confounding discontinuity that uses multiple pre-treatment periods. We contribute to this literature by extending the design to multiple time periods, and by deriving additional assumptions that allow identification in the presence of the composition effects that arise when the running variable varies over time. The diff-in-disc design has also been extended to fuzzy RD \citep{galindo2018fuzzy,millanquijano2020fuzzy,larsen2023fuzzy}. We focus on sharp RD assignments of both the treatment of interest and the confounding treatment. \citet{cellini2010value} study a setting where units can choose treatment over several periods with a sharp RD assignment, and \citet{hsu2024dynamic} use a conditional independence assumption to overcome the selection channel. Our setting differs in that diff-in-disc designs have both a confounding RD treatment and an RD treatment of interest, rather than a single RD treatment, and our identification differs in that we use assumptions that directly zero the composition effects.

We also contribute to the literature on estimation and inference in RD designs, which has mostly focused on a single cross section \citep[e.g.,][]{imbens2012optimal,calonico2014robust}. For a review see \citet{cattaneo2022regression}. We extend this analysis to estimators that aggregate discontinuities over multiple periods, deriving their variance under each sampling scheme, an MSE-optimal common bandwidth, and variance-minimizing weights.

The paper is also related to the general difference-in-differences (DID) literature. For a recent review see \citet{roth2023s}. Our aggregation of period-specific effects into an overall effect builds on \citet{callaway2021difference}, and our no-anticipation-type assumption is a local-to-the-cutoff variant of the general no-anticipation assumption used in the literature \citep[e.g.,][]{abbring2007econometric}. Also, the constant time-trend assumption on the confounding discontinuity is equivalent to a local-to-the-cutoff parallel trends assumption, and using more than two periods is similar to moving from parallel trends that hold for one pre-treatment period \citep[e.g.,][]{callaway2021difference} to parallel trends over all pre-treatment periods \citep[e.g.,][]{borusyak2021revisiting}. In DID with repeated cross sections, sampling is usually assumed stationary, which rules out changes in composition \citep{abadie2005semiparametric,santanna2020doubly}. Our setting differs in that even in panel data a time-varying running variable allows for composition effects at the cutoff, and we analyze the assumptions under which these effects are zero. 

The paper proceeds as follows. Section~\ref{sec:setup} introduces the two-period diff-in-disc design. Section~\ref{sec:id-paths} extends identification to multiple periods, and Section~\ref{sec:tvrv} to running variables that vary over time. Section~\ref{sec:estimation} develops estimation and inference. The above sections focus on the effect on the treated, and Section~\ref{sec:atu} discusses identification and estimation of the effect on the untreated. Section~\ref{sec:application} presents the empirical illustration, and Section~\ref{sec:conclusion} concludes. Proofs are in Appendix~\ref{app:proofs}, and technical derivations regarding estimation are in Appendix~\ref{app:est}.

\section{Setup and two-period identification}\label{sec:setup}

As a starting point to the analysis, this section introduces the setup and identification underlying the difference-in-discontinuities (diff-in-disc) of \citet{grembi2016fiscal}.

\subsection{Setup}\label{sec:setup_setup}

We start from a simplified setup. Section \ref{sec:id-paths} extends the setting to multiple periods, in which the potential outcome in a period may also depend on the treatments in other periods, and Section~\ref{sec:tvrv} studies the setting where the running variable may vary over time within units. 
Assume $n$ units are observed over two time periods $t\in\{t_0,t_\mathrm{RD}\}$. 
For unit $i$ and time $t$, let $R_{i,t} \in \mathbb{R}$ be the running variable and $c$ a known, time-invariant cutoff. 
By time-invariant running variable we mean that for each unit, $R$ is constant over time, i.e., $R_{i,t}=R_i$.
Let $W_{i,t}$ be the treatment of interest and $V_{i,t}$ a confounding treatment, namely a treatment allocated according to the same running variable and cutoff but that the effect of the confounding treatment on the outcome is not the inquiry of the study.
Assume $V_{i,t} = \ind\{R_{i,t} \ge c\}$ where $\ind\{\cdot\}$ is the indicator function. That is, $V$ follows a sharp RD
treatment assignment in every period.
We further assume that the $W$ treatment is absent for all units in period $t_0$ and that $W$ follows a sharp RD assignment in period $t_\mathrm{RD}$. This assumption can be further written as $W_{i,t} = V_{i,t}\ind\{t = t_\mathrm{RD}\}$. Going forward, we term $t_0$ the \textit{comparison period} and $t_\mathrm{RD}$ the \textit{RD period}. 

Let $Y_{i,t}(w,v)$ be the potential outcome of unit $i$ at time $t$ for $w,v\in\{0,1\}$, with observed outcomes $Y_{i,t}=Y_{i,t}(W_{i,t},V_{i,t})$.
For each $(w,v)\in\{0,1\}^2$ and period $t$, let
$\mu_{t,(w,v)}(r) = \E[Y_{t}(w,v)\mid R_{t}=r]$ be the conditional mean potential outcomes and denote by $\mu_{t,(w,v),(+)} = \lim_{r\downarrow c}\mu_{t,(w,v)}(r)$ and $\mu_{t,(w,v),(-)} = \lim_{r\uparrow c}\mu_{t,(w,v)}(r)$ their one-sided limits above and below the cutoff, respectively. Similarly, let $m_t(r) = \E[Y_t\mid R_t=r]$ be the conditional mean observed outcomes 
with analogous one-sided limits at the cutoff $m_{t,(+)} = \lim_{r\downarrow c} m_t(r)$ and $m_{t,(-)} = \lim_{r\uparrow c} m_t(r)$. Finally, denote the observed outcome discontinuity by $D_t=m_{t,(+)}-m_{t,(-)}$.

The target causal estimand is the average treatment effect of the treatment $W$ keeping the confounding treatment fixed at $V=1$, local to the cutoff $c$, in period $t$: $$\att(t,c) = \lim_{r\downarrow c}\mu_{t,(1,1)}(r) - \lim_{r\downarrow c}\mu_{t,(0,1)}(r).$$ 
We develop identification for the ATT throughout. Identification and estimation of an ATU estimand, 
$$\atu(t,c) = \lim_{r\uparrow c}\mu_{t,(1,0)}(r) - \lim_{r\uparrow c}\mu_{t,(0,0)}(r),$$
when the confounding treatment is fixed at $V=0$ in the comparison period is discussed in Section~\ref{sec:atu}. For brevity we suppress the dependence on $c$ and write $\att(t)\equiv\att(t,c)$ and $\atu(t)\equiv\atu(t,c)$ throughout.\footnote{We term the estimand the local ATT because it is defined above the cutoff, where in our setup units are treated following the sharp RD treatment assignment.}

\subsection{The problem: a violation of continuity}\label{sec:problem}

To fix ideas, we first show how violations of continuity affect RD identification in a single cross section, where $t=t_\mathrm{RD}$ and hence $W_{i,t}=V_{i,t}=\ind\{R_{i,t} \ge c\}$.\footnote{An alternative identification framework is local randomization (see \citet{cattaneo2022regression}).}
Since $V=0$ below $c$ and $V=1$ above, the conventional continuity assumption \citep{hahn2001identification} can be formalized as follows.

\begin{assumption}[Continuity]\label{ass:cont}
In period $t$, for each $w\in\{0,1\}$, $\ \mu_{t,(w,0),(-)} = \mu_{t,(w,1),(+)}$.
\end{assumption}

\noindent 
To discuss the impact of violations of Assumption \ref{ass:cont}, we relax the continuity assumption as follows.

\begin{assumption}\label{ass:relax_cont}
For all periods and $w,v\in\{0,1\}$ the limits $\mu_{t,(w,v),(+)}$ and $\mu_{t,(w,v),(-)}$ exist.
\end{assumption}

\noindent
Assumption~\ref{ass:relax_cont} is a technical assumption, and is expected to hold whenever the potential outcomes and the functions $\mu_{t,(w,v)}(r)$ are well-defined around the cutoff.  Let $\alpha_{t,w} = \mu_{t,(w,1),(+)} - \mu_{t,(w,0),(-)}$ for $w\in\{0,1\}$ be the potential outcome discontinuities at the cutoff. The following lemma shows that the observed discontinuity can be decomposed into a causal effect plus the confounding discontinuity. 

\begin{lemma}\label{prop:bias}
Consider a single cross section where both $W$ and $V$ follow a sharp RD treatment assignment. If Assumption \ref{ass:relax_cont} holds, $D_t = \att(t) + \alpha_{t,0}$.
\end{lemma}
\noindent All proofs of the results in the main text are provided in Appendix~\ref{app:proofs}.
Lemma~\ref{prop:bias} subsumes the known cross-sectional RD identification result under continuity. That is, under Assumption~\ref{ass:cont}, $D_t = \att(t)=\atu(t)$, as $\alpha_{t,0}=\alpha_{t,1}=0$. 

\subsection{Identification}\label{sec:id-two}

Going back to the two time periods scenario, in the comparison period $t_0$, the observed discontinuity equals the confounding discontinuity, i.e., $D_{t_0}=\alpha_{t_0,0}$, as below the cutoff $V_{i,t_0}=0$, and above the cutoff $V_{i,t_0}=1$ and $W_{i,t_0}=0$ for all units. Lemma~\ref{prop:bias} implies that identifying $\att(t_{\mathrm{RD}})$ is possible if $\alpha_{t_{\mathrm{RD}},0}$ can be identified. The following assumption allows identification of $\alpha_{t_{\mathrm{RD}},0}$ using the comparison period. 

\begin{assumption}[Constant confounding discontinuity]\label{ass:stable}
$\alpha_{t_0,0} = \alpha_{t_{\mathrm{RD}},0}$.
\end{assumption}

\noindent
Assumption~\ref{ass:stable} is the local-to-the-cutoff analog of the parallel-trends assumption in difference-in-differences.\footnote{To see this, note that Assumption \ref{ass:stable} equivalently can be written as $$\mu_{t_\mathrm{RD},(0,0),(-)} - \mu_{t_0,(0,0),(-)}=\mu_{t_\mathrm{RD},(0,1),(+)} - \mu_{t_0,(0,1),(+)},$$ which fixes the treated of interest $W$ to the untreated state, and compares expected potential outcome time trends just below the cutoff where $V=0$, to expected potential outcome time trends just above the cutoff where $V=1$.} Hence, it is not surprising that a local-to-the-cutoff difference-in-differences achieves identification of the local ATT.

\begin{proposition}\label{prop:twoperiod}
If Assumptions~\ref{ass:relax_cont} and~\ref{ass:stable} hold, then $\att(t_{\mathrm{RD}}) = D_{t_{\mathrm{RD}}} - D_{t_0}$.
\end{proposition}

\noindent
The identification result in Proposition \ref{prop:twoperiod} is the difference-in-discontinuities of \citet{grembi2016fiscal} (Proposition 1). Namely, the local ATT at the RD period is identified by subtracting the observed outcome discontinuity in the comparison period from the observed outcome discontinuity in the RD period.

\section{Identification with multiple periods}\label{sec:id-paths}

This section extends the setup of Section~\ref{sec:setup} in two ways: it allows for more than two periods, and it allows the potential outcome in a period to depend on the treatments in previous periods. The identification result is first presented for a general formula tying $\alpha_{t,0}$ across periods, with the constant and linear functional forms following as special cases.

\subsection{Setup}\label{sec:setup_multiple}

We begin by extending the setup in Section \ref{sec:setup} to multiple time periods. Let $t=1,...,P$ index $P$ time periods. As before, the confounding treatment is allocated according to the same sharp RD mechanism at all time periods, i.e., $V_{i,t} = \ind\{R_{i,t} \ge c\}$ for all $i$ and $t$. We now assume there is a non-empty set of comparison periods, denoted by $\mathcal{T}_0$, and  a non-empty set of sharp RD periods denoted by $\mathcal{T}_{\mathrm{RD}}$, such that $W_{i,t}=V_{i,t}\ind\{t \in \mathcal{T}_\mathrm{RD}\}$. We do not assume an ordering of time periods, that is, $\mathcal{T}_0$ can include periods that are before or after time periods in $\mathcal{T}_{\mathrm{RD}}$.\footnote{Some applications use comparison periods that are only before RD periods \citep{grembi2016fiscal,lalive2014parental,carneiro2015flying}, while others also incorporate comparison periods that come after RD periods \citep{avdic2018modern,bertrand2021improving}.} Denote by $\mathbf{W}_i=(W_{i,1},...,W_{i,P})$ and $\mathbf{V}_i=(V_{i,1},...V_{i,P})$  unit $i$'s treatment history across all periods.

Building on the literature that analyzes treatments with multiple time periods \citep{robins1986new,heckman2005structural}, let $Y_{i,t}(\mathbf{w},\mathbf{v})$ be the potential outcome for treatment vectors $\mathbf{w},\mathbf{v}\in\{0,1\}^P$. As we still assume the running variable is time invariant, there are only two possible observed treatment histories: all-above with $R_i\ge c$ (so $\mathbf{V}_i=\mathbf{1})$ and all-below with $R_i<c$ (so $\mathbf{V}_i=\mathbf{0})$, where $\mathbf{0}$ and $\mathbf{1}$ are $P$-length vectors of zeros and ones, respectively.
We can therefore focus on only four potential outcomes: $Y_{i,t}(\mathbf{0},\mathbf{0})$,  $Y_{i,t}(\mathbf{1}_w,\mathbf{0})$, $Y_{i,t}(\mathbf{0},\mathbf{1})$, and $Y_{i,t}(\mathbf{1}_w,\mathbf{1})$ where $\mathbf{1}_w$ is a $P$-length vector equal to zero in comparison periods and one in RD periods. The observed outcome is equal to $Y_{i,t}=\ind\{R_i\geq c\}Y_{i,t}(\mathbf{1}_w,\mathbf{1}) + \ind\{R_i< c\}Y_{i,t}(\mathbf{0},\mathbf{0})$.
Below we use the same notation as in Section \ref{sec:setup}, where $\mu_{t,(\mathbf{w},\mathbf{v})}$, $\att(t)$, and $\alpha_{t,0}$ are defined by replacing $(w,v)\in\{0,1\}$ with $\mathbf{w}\in\{\mathbf0, \mathbf1_w\},\mathbf{v}\in\{\mathbf{0},\mathbf{1}\}$.

\subsection{Identification}

Indexing potential outcomes by the entire treatment vector allows the period-$t$ potential outcome to vary with the treatments in other periods. The literature terms these anticipation effects in pre-treatment periods \citep{callaway2021difference,roth2023s} and carry-over effects in post-treatment periods \citep{imai2019when,bojinov2021panel}. The following assumption rules out such effects.

\begin{assumption}[No anticipation and no carry-over]\label{ass:noant}
$\att(t_0)=0$ for all $t_0\in\mathcal{T}_0$.
\end{assumption}

\noindent
Assumption~\ref{ass:noant} can be interpreted as a local-to-the-cutoff no-anticipation assumption in comparison time periods that come before all time periods in $\mathcal{T}_{\mathrm{RD}}$, and a local no-carry-over-effects assumption in comparison periods which are after a period in $\mathcal{T}_{\mathrm{RD}}$. If it holds, then $D_{t_0}=\alpha_{t_0,0}$ for $t_0\in\mathcal{T}_0$ (Lemma~\ref{prop:bias}). As in Section \ref{sec:setup}, this suggests that identification can be achieved by assuming how $\alpha_{t_\mathrm{RD},0}$ is related to $\alpha_{t_0,0}$ in comparison periods. The following assumption generalizes Assumption \ref{ass:stable} to multiple periods and to general trends. 

\begin{assumption}[Confounding-discontinuity trend]\label{ass:trend}
For each $t_{\mathrm{RD}}\in\mathcal{T}_{\mathrm{RD}}$, $\alpha_{t_{\mathrm{RD}},0}=g_{0,t_{\mathrm{RD}}}\big(\{\alpha_{t_0,0}:t_0\in\mathcal{T}_0\}\big)$ for a known function $g_{0,t_{\mathrm{RD}}}$.
\end{assumption}

\noindent
Assumption \ref{ass:trend} extends Assumption \ref{ass:stable} in two ways. First, it ties the confounding discontinuity in the RD period to more than one comparison period. In the DID literature this is the move from parallel trends relative to a single pre-treatment period \citep[e.g.,][]{callaway2021difference} to parallel trends over all pre-treatment periods \citep[e.g.,][]{borusyak2021revisiting}. Similarly here, Assumption \ref{ass:trend} gives the relationship between $t_\mathrm{RD}$ and all comparison periods in $\mathcal T_0$, unlike Assumption \ref{ass:stable} that uses only one comparison period. Second, leaving the form of $g_{0,t_{\mathrm{RD}}}$ free allows for various time trends, such as a linear trend. As the number of periods in $\mathcal{T}_0$ grows, whether the form of $g_{0,t_{\mathrm{RD}}}$ is consistent with the comparison periods can be tested. One can also view the non-constant forms as sensitivity analysis, similar to the discussion in \citet{rambachan2023more} for DID. 

The following result generalizes Proposition~\ref{prop:twoperiod} to the setup of Section~\ref{sec:setup_multiple}. 

\begin{theorem}\label{thm:id}
Under the setup in Section \ref{sec:setup_multiple}, if Assumptions~\ref{ass:relax_cont}, \ref{ass:noant} and~\ref{ass:trend} hold, for any $t_{\mathrm{RD}}\in\mathcal{T}_{\mathrm{RD}}$, $\att(t_{\mathrm{RD}}) = D_{t_{\mathrm{RD}}} - g_{0,t_{\mathrm{RD}}}\big(\{D_{t_0}:t_0\in\mathcal{T}_0\}\big)$.
\end{theorem}
The intuition behind Theorem \ref{thm:id} is as follows. In time periods $t_0\in\mathcal{T}_0$, observed outcomes from both sides of the cutoff are equal to the potential outcomes under no $W$ treatment, and hence $\alpha_{t_0,0}$ is identifiable. Under Assumption \ref{ass:trend},  the bias parameter $\alpha_{t_\mathrm{RD},0}$ is obtained by applying $g_{0,t_\mathrm{RD}}$ to the $\alpha_{t_0,0}$'s in time periods $t_0\in\mathcal{T}_0$. By Lemma \ref{prop:bias}, the $\att(t_\mathrm{RD})$ is then identified as the difference between  $D_{t_\mathrm{RD}}$ and the confounding discontinuity $\alpha_{t_\mathrm{RD},0}$.

Next, we discuss two special cases of $g_{0,t_{\mathrm{RD}}}$. We consider functional forms of $g_{0,t_{\mathrm{RD}}}$ that can be written as a weighted sum of the comparison-period confounding discontinuities, i.e., $g_{0,t_{\mathrm{RD}}}\big(\{\alpha_{t_0,0}:t_0\in\mathcal{T}_0\}\big)=\sum_{t_0\in\mathcal{T}_0} \kappa_{t_0}\,\alpha_{t_0,0}$. The following result shows the conditions on the weights $\kappa_t$ for the constant and linear in time cases of Assumption \ref{ass:trend}.

\begin{corollary}\label{cor:trend}
Assume the setup in Section \ref{sec:setup_multiple} and Assumptions~\ref{ass:relax_cont} and~\ref{ass:noant}. If $\alpha_{t,0}$ is constant in $t$ then $\alpha_{t_\mathrm{RD},0}=\sum_{t_0\in\mathcal{T}_0} \kappa_{t_0}\,\alpha_{t_0,0}$ under the constraints that $\mathcal{T}_0$ includes at least one period and $\sum_{t_0\in\mathcal{T}_0} \kappa_{t_0} = 1$. If $\alpha_{t,0}$ is linear in $t$ then $\alpha_{t_\mathrm{RD},0}=\sum_{t_0\in\mathcal{T}_0} \kappa_{t_0}\,\alpha_{t_0,0}$ under the constraints that $\mathcal{T}_0$ includes at least two periods,  $\sum_{t_0\in\mathcal{T}_0} \kappa_{t_0} = 1$, and $\sum_{t_0\in\mathcal{T}_0} \kappa_{t_0}\, t_0 = t_{\mathrm{RD}}$. In both cases, $$
\att(t_{\mathrm{RD}}) = D_{t_{\mathrm{RD}}} - \sum_{t_0\in\mathcal{T}_0} \kappa_{t_0}\, D_{t_0}.$$
\end{corollary}

\noindent
In the constant case one might also impose a non-negative weights restriction, so $g_{0,t_{\mathrm{RD}}}$ becomes a weighted average. In the linear case, when $t_{\mathrm{RD}}$ lies outside the range of the comparison periods, the additional restriction forces some weights to be negative. For example, with $\mathcal{T}_0=\{1,2\}$ and $t_{\mathrm{RD}}=3$ the only admissible weights are $\kappa_1=-1$ and $\kappa_2=2$. With more than two comparison periods multiple sets of weights are admissible, one of which is the one implied by the ordinary least squares over $\mathcal{T}_0$. Building on Corollary~\ref{cor:trend}, Section \ref{sec:estimation} discusses estimators of the form $\widehat{D}_{t_{\mathrm{RD}}} - \sum_{t_0\in\mathcal{T}_0} \kappa_{t_0}\, \widehat{D}_{t_0}$. For identification the numerical values of the weights do not matter as long as they satisfy the constraints. However, for estimation this is not the case, and we discuss weights that aim to maximize precision in Section \ref{sec:estimation}.

When $\mathcal{T}_\mathrm{RD}$ includes more than one time period, Theorem~\ref{thm:id} identifies $\att(t_{\mathrm{RD}})$ for each RD period separately, and hence also any weighted average of these effects. Denote by $\phi_{t_{\mathrm{RD}}}$ a choice of weights by the researcher that sum to one, i.e., $\sum_{t_\mathrm{RD}\in\mathcal T_\mathrm{RD}}\phi_{t_{\mathrm{RD}}}=1$. For example, one can set these to be equal weights, i.e., $\phi_{t_{\mathrm{RD}}}=1/|\mathcal T_{\mathrm{RD}}|$. Let $\att^{\phi}=\sum_{t_{\mathrm{RD}}\in\mathcal T_{\mathrm{RD}}}\phi_{t_{\mathrm{RD}}}\att(t_{\mathrm{RD}})$ be the weighted average of $t_\mathrm{RD}$-period ATTs. For each $t_{\mathrm{RD}}\in\mathcal T_{\mathrm{RD}}$, denote by $\{\kappa_{t_0}(t_{\mathrm{RD}})\}_{t_0\in\mathcal T_0}$ a set of  comparison period weights that satisfy the constraints of Corollary~\ref{cor:trend} for RD period $t_{\mathrm{RD}}$.

\begin{corollary}\label{cor:agg}
Assume the setup in Section~\ref{sec:setup_multiple}, Assumptions~\ref{ass:relax_cont} and~\ref{ass:noant}, and that $\alpha_{t,0}$ is constant or linear in $t$, and for each $t_{\mathrm{RD}}\in\mathcal T_{\mathrm{RD}}$ let the weights $\{\kappa_{t_0}(t_{\mathrm{RD}})\}_{t_0\in\mathcal T_0}$ satisfy the corresponding constraints of Corollary~\ref{cor:trend} for RD period $t_{\mathrm{RD}}$. Let $\kappa^{\phi}_{t_0}=\sum_{t_{\mathrm{RD}}\in\mathcal T_{\mathrm{RD}}}\phi_{t_{\mathrm{RD}}}\,\kappa_{t_0}(t_{\mathrm{RD}})$. Then,
\[
  \att^{\phi}=\sum_{t_{\mathrm{RD}}\in\mathcal T_{\mathrm{RD}}}\phi_{t_{\mathrm{RD}}}\,D_{t_{\mathrm{RD}}}-\sum_{t_0\in\mathcal T_0}\kappa^{\phi}_{t_0}\,D_{t_0}.
\]
\end{corollary}

\noindent
Corollary \ref{cor:agg} shows that the aggregate treatment effect, which weights each period's effect by $\phi_{t_{\mathrm{RD}}}$, can be written as a weighted sum of each period's discontinuity in $\mathcal T_0$ and $\mathcal T_\mathrm{RD}$. This suggests a simple estimation procedure: estimate each per-period discontinuity once, and aggregate according to weights $\kappa$ and $\phi$, discussed in Section \ref{sec:est-agg}.

\section{Identification with time-varying running variables}\label{sec:tvrv}

This section allows the running variable to change over time within units. In \citet{grembi2016fiscal}, the running variable is population-size, which can change over time due to positive fertility or migration. Similarly, in another example, \citet{bennedsen2022firms} analyze firms and the running variable is the number of employees, which can change over time due to hiring or separation.  When the running variable may change over time, a unit may lie above the threshold $c$ in one period and below it in another. In this section, we show how this phenomenon introduces composition effects, present new identification assumptions that allow using the conventional diff-in-disc formula, and develop testing procedures to test whether these assumptions hold in the data. As the multiple-period setup complicates the notation, we focus on the case of two time periods. 

\subsection{Setup}\label{sec:tvrv-setup}

To discuss the composition effects that arise due to variation in $R$ over time, we need to introduce new notation. Two main objects that underlie the analysis are the shares of units at the cutoff in one period by their side of the cutoff in the other period, and mean outcomes in one period conditional on the side of the cutoff in the other period. Below we present these objects, and relate them to the ATT, the confounding discontinuity and the observed discontinuity of the previous sections. The next subsection formalizes the composition effects. 

Assume two periods, a comparison period $t_0$ and an RD period $t_{\mathrm{RD}}$ as in Section \ref{sec:setup}, with treatment vectors $\mathbf W_i$, $\mathbf V_i$ and potential outcomes $Y_{i,t}(\mathbf w,\mathbf v)$ as in Section~\ref{sec:id-paths}. The running variable may change between the two periods, so a unit can be on different sides of $c$ in $t_0$ and in $t_{\mathrm{RD}}$. Denote a unit's pair of sides by
\[
  \mathbf V_i=(V_{i,t_0},V_{i,t_{\mathrm{RD}}})\in\{(0,0),(0,1),(1,0),(1,1)\}.
\]
With a time-constant running variable, as in Section~\ref{sec:id-paths}, only the pairs $(0,0)$ and $(1,1)$ can occur.

For a given  $R_{i,t_0}$,  a unit's side in $t_0$, $V_{i,t_0}$, is fixed, but not its side in $t_\mathrm{RD}$. Likewise, for a given $R_{i,t_\mathrm{RD}}$,  the side in $t_\mathrm{RD}$ is fixed but not the side in $t_0$. For $v\in\{0,1\}$, let $\pi_{t,(+)}(v)$ be the share of units just above the cutoff in period $t$ that are on side $v$ in the other period. Formally, 
\[
  \pi_{t_0,(+)}(v)=\lim_{r\downarrow c}\Pr(V_{i,t_\mathrm{RD}}=v\mid R_{i,t_0}=r),\qquad
  \pi_{t_\mathrm{RD},(+)}(v)=\lim_{r\downarrow c}\Pr(V_{i,t_0}=v\mid R_{i,t_\mathrm{RD}}=r),
\]
and let $\pi_{t,(-)}(v)$ be the same share just below the cutoff with $r\uparrow c$. For example, $\pi_{t_0,(+)}(1)$ is the share of units just above the cutoff in $t_0$ that stay above it in $t_\mathrm{RD}$, and $\pi_{t_0,(-)}(1)$ is the share of units just below the cutoff in $t_0$ that move above it in $t_\mathrm{RD}$. We call $\pi_{t,(+)}(v)$ and $\pi_{t,(-)}(v)$ the side shares.

We now turn to defining the mean potential outcomes conditional on the side in the other period.
Recall that $Y_{i,t}(\mathbf w,\mathbf v)$ is the potential outcome when $(W_{i,t_0},W_{i,t_\mathrm{RD}})$ is set to $\mathbf w$ and $(V_{i,t_0},V_{i,t_\mathrm{RD}})$ is set to $\mathbf v$. Each mean below sets the treatments in the current period $t$ to $(w_t,v_t)$, while keeping the treatments in the other period unchanged according to the conditioned side. That is,  for $v\in\{0,1\}$, let
\[
\begin{aligned}
  \mu_{t_0,(w_{t_0},v_{t_0}),(+)}(v)&=\lim_{r\downarrow c}\mathbb{E}\bigl[Y_{i,t_0}\bigl((w_{t_0},v),(v_{t_0},v)\bigr)\mid R_{i,t_0}=r,\ V_{i,t_\mathrm{RD}}=v\bigr],\\
  \mu_{t_\mathrm{RD},(w_{t_\mathrm{RD}},v_{t_\mathrm{RD}}),(+)}(v)&=\lim_{r\downarrow c}\mathbb{E}\bigl[Y_{i,t_\mathrm{RD}}\bigl((0,w_{t_\mathrm{RD}}),(v,v_{t_\mathrm{RD}})\bigr)\mid R_{i,t_\mathrm{RD}}=r,\ V_{i,t_0}=v\bigr],
\end{aligned}
\]
be the conditional mean potential outcomes just above the cutoff and let $\mu_{t,(w_t,v_t),(-)}(v)$ be the analogous means just below the cutoff, defined by replacing the $r\downarrow c$ in $\mu_{t,(w_{t},v_{t}),(+)}(v)$ with $r\uparrow c$. In period $t_0$, the argument $v$ is the unit's side in $t_\mathrm{RD}$, and in period $t_\mathrm{RD}$ it is the unit's side in $t_0$.

In the setting of time-varying running variables the ATT of Sections~\ref{sec:setup}--\ref{sec:id-paths} is equal to a weighted average of the conditional ATTs, 
\[
\att(t_\mathrm{RD})=\sum_{v\in\{0,1\}}\pi_{t_\mathrm{RD},(+)}(v)\,\att(t_\mathrm{RD}\mid v),
\]
where $\att(t_\mathrm{RD}\mid v)$ is the conditional ATT in the RD period, defined as
\[
  \att(t_\mathrm{RD}\mid v)=\mu_{t_\mathrm{RD},(1,1),(+)}(v)-\mu_{t_\mathrm{RD},(0,1),(+)}(v).
\]
Similarly, the confounding discontinuity under the setting in Sections~\ref{sec:setup} and~\ref{sec:id-paths} can be written as a difference of two weighted averages over $v$, 
\[
\alpha_{t,0}=\sum_{v\in\{0,1\}}\pi_{t,(+)}(v)\,\mu_{t,(0,1),(+)}(v)-\sum_{v\in\{0,1\}}\pi_{t,(-)}(v)\,\mu_{t,(0,0),(-)}(v),
\] 
where, for $v\in\{0,1\}$, the conditional confounding discontinuities are
\[
  \alpha_{t_0,0}(v)=\mu_{t_0,(0,1),(+)}(v)-\mu_{t_0,(0,0),(-)}(v),\qquad
  \alpha_{t_\mathrm{RD},0}(v)=\mu_{t_\mathrm{RD},(0,1),(+)}(v)-\mu_{t_\mathrm{RD},(0,0),(-)}(v),
\]
the discontinuities at the period-$t$ cutoff in the potential outcome without $W$, among units on side $v$ in the other period. 
In this section we instead focus on the average of the conditional confounding discontinuities weighted by the side shares above the cutoff:
\[
  \bar\alpha_{t,0}=\sum_{v\in\{0,1\}}\pi_{t,(+)}(v)\,\alpha_{t,0}(v).
\]
The two, $\alpha_{t,0}$ and $\bar\alpha_{t,0}$, differ by a within-period composition term, introduced in the next section. 

In a similar fashion, for $v\in\{0,1\}$, the conditional observed outcome means just above the cutoff are
\[
\begin{aligned}
  m_{t_0,(+)}(v)&=\lim_{r\downarrow c}\mathbb{E}\bigl[Y_{i,t_0}\mid R_{i,t_0}=r,\ V_{i,t_\mathrm{RD}}=v\bigr],\\
  m_{t_\mathrm{RD},(+)}(v)&=\lim_{r\downarrow c}\mathbb{E}\bigl[Y_{i,t_\mathrm{RD}}\mid R_{i,t_\mathrm{RD}}=r,\ V_{i,t_0}=v\bigr],
\end{aligned}
\]
and $m_{t,(-)}(v)$ are the analogous means just below the cutoff, by taking $r\uparrow c$. The observed outcome is the potential outcome at the unit's observed treatments, i.e., $Y_{i,t}=Y_{i,t}\bigl((0,W_{i,t_\mathrm{RD}}),(V_{i,t_0},V_{i,t_\mathrm{RD}})\bigr)$ where $W_{i,t_\mathrm{RD}}=V_{i,t_\mathrm{RD}}$. This implies that $m_{t,(+)}(v)=\mu_{t,(\ind\{t=t_\mathrm{RD}\},1),(+)}(v)$ and $m_{t,(-)}(v)=\mu_{t,(0,0),(-)}(v)$. The conditional observed discontinuity is denoted by $D_t(v)=m_{t,(+)}(v)-m_{t,(-)}(v)$, and the unconditional observed discontinuity is 
\[
D_t=\sum_{v\in\{0,1\}}\pi_{t,(+)}(v)\,m_{t,(+)}(v)-\sum_{v\in\{0,1\}}\pi_{t,(-)}(v)\,m_{t,(-)}(v).
\] 

Finally, the following regularity condition generalizes Assumption \ref{ass:relax_cont} to the time-varying running variable setup, and assumes existence of the above limits. 
\begin{assumption}\label{ass:tvrv-limits}
For each period $t$ and all $w_t,v_t,v\in\{0,1\}$, the limits $\mu_{t,(w_t,v_t),(+)}(v)$, $\mu_{t,(w_t,v_t),(-)}(v)$, $\pi_{t,(+)}(v)$ and $\pi_{t,(-)}(v)$ exist.
\end{assumption}

\subsection{The problem: composition effects}\label{sec:tvrv-problem}

A time-varying running variable introduces two composition effects, one within a period and one between periods. The following lemma generalizes the single-cross-section result in Lemma~\ref{prop:bias}. It shows that the observed outcome discontinuity includes not only $\bar\alpha_{t,0}$, but also a composition effect.

\begin{lemma}\label{prop:sdisc}
Under the setup in Section~\ref{sec:tvrv-setup} and Assumption~\ref{ass:tvrv-limits}, $D_{t_0}=\bar\alpha_{t_0,0}+S_{t_0}$ and $D_{t_\mathrm{RD}}=\att(t_\mathrm{RD})+\bar\alpha_{t_\mathrm{RD},0}+S_{t_\mathrm{RD}}$, where
\[
  S_t=\bigl[\pi_{t,(+)}(1)-\pi_{t,(-)}(1)\bigr]\bigl[\mu_{t,(0,0),(-)}(1)-\mu_{t,(0,0),(-)}(0)\bigr].
\]
\end{lemma}

\noindent
The lemma decomposes the observed discontinuity in each period into a causal effect $\att(t_\mathrm{RD})$ (in the RD period only), a confounding discontinuity term $\bar\alpha_{t,0}$, and a composition term $S_t$. The first factor of $S_t$ is the discontinuity at the period-$t$ cutoff in the share of units that are above the cutoff in the other period. The second is the difference in the mean untreated outcome just below the cutoff between units above and units below the cutoff in the other period. Hence $S_t=0$ if the side shares are continuous at the cutoff, or if the mean untreated outcome just below the cutoff does not depend on the side in the other period. Additionally, as shown in the proof of the above lemma, conditional on $v$ there is no composition term, i.e., $D_{t_0}(v)=\alpha_{t_0,0}(v)$ and $D_{t_\mathrm{RD}}(v)=\att(t_\mathrm{RD}\mid v)+\alpha_{t_\mathrm{RD},0}(v)$.\footnote{An assumption similar to Assumption~\ref{ass:noant}, regarding anticipation or carry-over, is not needed conditional on $v$. Units on the same side in the other period have the same treatments in that period on both sides of the cutoff, and hence there is no different treatment path to create anticipation or carry over effects. The treatments in the other period instead enter $\alpha_{t,0}(v)$; see the discussion after Assumption~\ref{ass:trend-cell}.}
Lemma \ref{prop:sdisc} is also related to identification in a conventional RD. Even if $\bar\alpha_{t_\mathrm{RD},0}=0$, the presence of $S_{t_\mathrm{RD}}$ poses a problem for the identification of the ATT, which can be seen as a form of sorting on the running variable, as discussed in \citet{lee2010regression}.

To help intuition, we discuss two possible mechanisms that might cause $S_t$ to be nonzero, focusing on the side shares, illustrated using the application of \citet{grembi2016fiscal} as an example. In $t_0$, one reason is that the confounding treatment affects how the running variable evolves. In the example, this can happen since municipalities just above 5,000 residents in 2000 ($t_0$) pay higher salaries to the mayor and executives ($V$), which in turn might affect outcomes that attract more residents, making these municipalities more likely to be above the cutoff in 2001 ($t_\mathrm{RD}$) as well. A second reason is that units act to obtain the confounding treatment in the next period \citep[the selection channel in][]{hsu2024dynamic}. Returning to the example, municipalities just below 5,000 inhabitants in 2000 could try to raise their number of residents to be above 5,000 in 2001 to obtain the higher salaries, while municipalities just above 5,000 make less effort to stay above it. 

Next, we turn to the between-period composition effect. The following lemma shows that the change in $\bar\alpha_{t,0}$ between the two time periods can be decomposed into two components: one that holds $\alpha_{t,0}(v)$ at their levels in time period $t_0$ and allows the shares $\pi_{t,(+)}(1)$ to vary over time, which we term a composition term, and a second term that holds $\pi_{t,(+)}(v)$ at their values in time period $t_\mathrm{RD}$ and varies $\alpha_{t,0}(v)$ over time, which we term a trend term.

\begin{lemma}\label{prop:cdisc}
Under the setup in Section~\ref{sec:tvrv-setup} and Assumption~\ref{ass:tvrv-limits}, $\bar\alpha_{t_{\mathrm{RD}},0}-\bar\alpha_{t_0,0}=C_{t_0,t_{\mathrm{RD}}}+T_{t_0,t_{\mathrm{RD}}}$, where
\[
\begin{aligned}
  C_{t_0,t_\mathrm{RD}}&=\bigl[\pi_{t_\mathrm{RD},(+)}(1)-\pi_{t_0,(+)}(1)\bigr]\bigl[\alpha_{t_0,0}(1)-\alpha_{t_0,0}(0)\bigr],\\
  T_{t_0,t_\mathrm{RD}}&=\pi_{t_{\mathrm{RD}},(+)}(0)\bigl[\alpha_{t_{\mathrm{RD}},0}(0)-\alpha_{t_{0},0}(0)\bigr]+\pi_{t_{\mathrm{RD}},(+)}(1)\bigl[\alpha_{t_{\mathrm{RD}},0}(1)-\alpha_{t_{0},0}(1)\bigr].
\end{aligned}
\]
\end{lemma}

\noindent
The first term in Lemma \ref{prop:cdisc}, $C_{t_0,t_\mathrm{RD}}$ is the change in the share $\pi_{t,(+)}(1)$ between the two periods, and the second is the difference in the period-$t_0$ confounding discontinuity between units above and units below the cutoff in $t_\mathrm{RD}$. Hence $C_{t_0,t_\mathrm{RD}}=0$ if the share does not change, or if the period-$t_0$ confounding discontinuity does not depend on $v$. The trend term $T_{t_0,t_\mathrm{RD}}$ is the change between the two periods in the conditional confounding discontinuities, weighted by the side shares just above the cutoff in $t_\mathrm{RD}$. Hence $T_{t_0,t_\mathrm{RD}}=0$ if, for each $v$, the conditional confounding discontinuity is the same in both periods. Lemma \ref{prop:cdisc} implies that even if $T_{t_0,t_\mathrm{RD}}=0$, the confounding discontinuity is not constant across periods when $C_{t_0,t_\mathrm{RD}}\neq0$.

Next, we discuss what can make the first factor of $C_{t_0,t_\mathrm{RD}}$ nonzero. Since $\pi_{t,(+)}(0)=1-\pi_{t,(+)}(1)$, the first factor equals $\pi_{t_0,(+)}(0)-\pi_{t_\mathrm{RD},(+)}(0)$: the share of units just above the cutoff in $t_0$ that are below it in $t_\mathrm{RD}$, minus the share of units just above the cutoff in $t_\mathrm{RD}$ that were below it in $t_0$. One reason for this difference to be nonzero is when units near the cutoff grow or shrink between the two periods. Returning to the \citet{grembi2016fiscal} example, if municipalities with around 5,000 residents grow between the two censuses, those just above 5,000 in 2000 tend to remain above it in 2001, while those just above 5,000 in 2001 tend to have crossed from below. Another potential reason is that near the cutoff there are more units below it than above it, or the reverse, so that even if each unit is as likely to move up as down, more units cross in one direction. In the example, this happens if near the cutoff there are more municipalities with fewer than 5,000 residents. Neither reason creates a jump in the side shares at the cutoff, so the between-period composition effect can arise even when $S_{t_\mathrm{RD}}=0$. Finally, the behavioral mechanisms discussed above for $S_t$ can also make the first factor of $C_{t_0,t_\mathrm{RD}}$ nonzero. For example, if municipalities above 5,000 in 2000 act to keep their population above it, fewer of them leave the above side.

Lemmas \ref{prop:sdisc}--\ref{prop:cdisc} present the potential composition effects. They are also useful, in that together they provide a roadmap to the identification assumptions under which the diff-in-disc formula identifies the causal estimand in the presence of composition effects. Consider the difference $D_{t_\mathrm{RD}}-D_{t_0}$, substitute each $D_t$ with its decomposition in Lemma \ref{prop:sdisc}, and then substitute the difference in $\bar\alpha_{t,0}$ from Lemma \ref{prop:cdisc}. This results in 
\begin{equation}\label{eq:anatomy}
D_{t_{\mathrm{RD}}} - D_{t_0} = \att(t_{\mathrm{RD}}) +S_{t_{\mathrm{RD}}}-S_{t_0} + C_{t_0,t_{\mathrm{RD}}}+T_{t_0,t_{\mathrm{RD}}}.
\end{equation}
Building on \eqref{eq:anatomy}, we next show how the difference in observed discontinuities, $D_{t_\mathrm{RD}}-D_{t_0}$, identifies the ATT, under assumptions that result in zero trend and composition terms.

\subsection{Identification}\label{sec:tvrv-id}

We begin with the trend term $T_{t_0,t_\mathrm{RD}}$. For it to be zero we require a stronger, conditional version of the constant-confounding assumption of Section \ref{sec:setup}.

\begin{assumption}[Constant conditional confounding discontinuity]\label{ass:trend-cell}
For $v\in\{0,1\}$, $\alpha_{t_\mathrm{RD},0}(v)=\alpha_{t_0,0}(v)$.
\end{assumption}

\noindent
Assumption~\ref{ass:trend-cell} is the conditional analog of Assumption~\ref{ass:stable}: for each $v$, the conditional confounding discontinuity is stable across the two periods, which implies $T_{t_0,t_\mathrm{RD}}=0$. For each value of $v$ (each side), the assumption equates the discontinuity of units that stay on one side of the cutoff in both periods with units that switch sides. The stayers are the same in both periods, but the switchers move in opposite directions.\footnote{For $v=0$, units with $\mathbf V_i=(0,0)$ just below the cutoff are compared with units with $\mathbf V_i=(1,0)$ just above it in $t_0$ and with $\mathbf V_i=(0,1)$ in $t_\mathrm{RD}$. For $v=1$, units with $\mathbf V_i=(1,1)$ just above the cutoff are compared with units with $\mathbf V_i=(0,1)$ just below it in $t_0$ and with $\mathbf V_i=(1,0)$ in $t_\mathrm{RD}$. For $v=1$, the two sides also hold the treatments in the other period at different values: the period-$t_0$ jump at $W_{i,t_\mathrm{RD}}=V_{i,t_\mathrm{RD}}=1$, and the period-$t_\mathrm{RD}$ jump at $W_{i,t_0}=0$ and $V_{i,t_0}=1$. Hence, for $v=1$, the assumption also restricts how these treatments affect the jump. This part of the restriction holds automatically when potential outcomes depend only on the treatments in the current period, $Y_{i,t}(\mathbf w,\mathbf v)=Y_{i,t}(w_t,v_t)$, as in Section~\ref{sec:setup}.} Next, we turn to the composition terms, beginning with an assumption that zeros $S_t$.

\begin{assumption}[Continuous side shares]\label{ass:type-cont}
For every period $t$ and $v\in\{0,1\}$, $\pi_{t,(+)}(v)=\pi_{t,(-)}(v)$.
\end{assumption}

\noindent
Assumption~\ref{ass:type-cont} rules out the selection and sorting mechanisms discussed in Section~\ref{sec:tvrv-problem}. In words, among units at the cutoff in one period, the side of the cutoff in that period is independent of the side in the other period. This holds, for example, when the running variable changes over time for reasons unrelated to the cutoff, as in a random walk.

Next, we introduce assumptions that zero the between-period composition effect, $C_{t_0,t_\mathrm{RD}}=0$. This can be achieved under either of two conditions: the composition is constant across time, or the confounding does not vary with $v$.

\begin{assumption}[Composition stability]\label{ass:comp-stable}
For $v\in\{0,1\}$, $\pi_{t_\mathrm{RD},(+)}(v)=\pi_{t_0,(+)}(v)$.
\end{assumption}

\noindent
Assumption \ref{ass:comp-stable} is the local-to-the-cutoff analog of the no-compositional-change condition that is sometimes invoked in repeated-cross-section DID \citep{abadie2005semiparametric,santanna2020doubly}. The assumption holds, for example, when for a unit just above the cutoff in period $t$, its side of the cutoff in the other period is a random coin toss, and vice-versa. It fails, for example, when units near the cutoff grow on net between the two periods, as discussed in Section~\ref{sec:tvrv-problem}.

\begin{assumption}[Homogeneous confounding]\label{ass:homog}
In the comparison period, the conditional confounding discontinuity does not vary with $v$: $\alpha_{t_0,0}(0)=\alpha_{t_0,0}(1)$.
\end{assumption}

\noindent
Assumption \ref{ass:homog} restricts the confounding treatment effect across $v$. It holds, for example, if $Y_{i,t}(\mathbf w,\mathbf v)=\mu_i+f_t(R_{i,t})+\tau w_t+\alpha_{t,0}\,v_t+\varepsilon_{i,t}$ and $\E[\mu_i+\varepsilon_{i,t_0}\mid R_{i,t_0}=r,V_{i,t_\mathrm{RD}}=v]$ is continuous at $c$ for each $v$, which implies $\alpha_{t_0,0}(v)=\alpha_{t_0,0}$ for both $v\in\{0,1\}$.

The following result shows that under these additional assumptions, identification of the local ATT is achieved via the difference in discontinuities. 

\begin{theorem}\label{thm:route-stable}
Under the setup in Section \ref{sec:tvrv-setup}, assume Assumptions \ref{ass:tvrv-limits}--\ref{ass:type-cont} hold. If in addition either Assumption \ref{ass:comp-stable} or Assumption \ref{ass:homog} holds, $\att(t_{\mathrm{RD}})=D_{t_{\mathrm{RD}}}-D_{t_0}$.
\end{theorem}

\noindent
Theorem \ref{thm:route-stable} combines \eqref{eq:anatomy} with assumptions that zero $T_{t_0,t_\mathrm{RD}}$, $S_t$ and $C_{t_0,t_\mathrm{RD}}$ to identify the ATT using the observed difference in discontinuities under a time-varying running variable. This clarifies for such settings the conditions under which the identification formulas under a time-varying (Theorem \ref{thm:route-stable}) and time-constant (Proposition \ref{prop:twoperiod}) running variable are equivalent, which motivates using the same estimator under both approaches when these conditions hold. 

\subsection{Testing the identification assumptions}\label{sec:tvrv-assess}

We now turn to propose validation tests of the identification assumptions. 
Assumptions~\ref{ass:type-cont} and~\ref{ass:comp-stable} restrict the side shares, which are observable, and hence can be empirically verified. Assumption~\ref{ass:homog} can also be empirically verified as it restricts the conditional confounding discontinuity in comparison periods, where it can be identified from the conditional observed discontinuity. However, Assumption~\ref{ass:trend-cell} restricts unobserved potential outcomes in the RD period, and hence is not directly testable, similarly to Assumptions~\ref{ass:stable} and \ref{ass:trend}. The tests are stated for balanced panel data with time-varying running variables.

\paragraph{Continuous side shares (Assumption~\ref{ass:type-cont}).} The proposed test is, in each period $t$, to run a local-linear RD of the side in the other period, $V_{i,t_\mathrm{RD}}$ in $t_0$ and $V_{i,t_0}$ in $t_\mathrm{RD}$, on $R_{i,t}$, which estimates the side-share jump, $\widehat\pi_{t,(+)}(1)-\widehat\pi_{t,(-)}(1)$. One can evaluate whether the jump is equal to zero in each period separately, or in both periods via a Wald statistic. When constructing the Wald statistic, note that the off-diagonal element in the variance-covariance matrix of the estimators is nonzero in panel data, as the same units enter both periods' regressions. Local-linear RD estimators are discussed below in Section~\ref{sec:estimation}. 

As an alternative testing procedure, one might consider testing the continuity of the running-variable density with a McCrary test \citep{mccrary2008manipulation}, which some diff-in-disc applications apply separately within each period \citep[e.g.,][]{eggers2017regression}. However, the two tests ask different questions about the units at the cutoff: the McCrary test asks how many units are on each side, while Assumption~\ref{ass:type-cont} concerns who they are. Hence one test cannot stand in for the other, as the number of units can be smooth at the cutoff while the side shares are discontinuous, and the number of units can be discontinuous while the side shares are not.

\paragraph{Composition stability (Assumption~\ref{ass:comp-stable}).}
Our proposed test runs a local-linear RD using a dataset containing a subset of the units and variables ($R^{test}$, $Y^{test}$) constructed in the following way. First, take the units above the cutoff in period $t_0$, center their running variable and flip its sign. That is, assign their running variable as $R^{test}_i=-(R_{i,t_0}-c)$ for $\{i: R_{i,t_0} > c\}$. Then, append the units above the cutoff in period $t_\mathrm{RD}$, and set $R^{test}$ as their centered running variable, i.e., $R^{test}_i=R_{i,t_\mathrm{RD}}-c$, for $\{i: R_{i,t_\mathrm{RD}} > c\}$. This way units with $R^{test}$ below zero are from $t_0$ and above zero are from $t_\mathrm{RD}$. As outcomes, for the period-$t_0$ units take $Y_i^{test}=V_{i,t_\mathrm{RD}}$, and for the period-$t_\mathrm{RD}$ units $Y_i^{test}=V_{i,t_0}$. Finally, run a local-linear RD with ($R^{test}$, $Y^{test}$) on the constructed data.  This procedure estimates $\widehat\pi_{t_\mathrm{RD},(+)}(1)-\widehat\pi_{t_0,(+)}(1)$, and the hypothesis that the discontinuity is equal to zero can be tested. Since a unit above the cutoff in both periods appears in both groups, the estimated variance of the estimated jump must account for this dependence.

\paragraph{Homogeneous confounding (Assumption~\ref{ass:homog}).}
As discussed above, in the comparison period $D_{t_0}(v)=\alpha_{t_0,0}(v)$. Hence, to validate the assumption that $\alpha_{t_0,0}(0)=\alpha_{t_0,0}(1)$, estimate both $\widehat{D}_{t_0}(0)$ within the subpopulation $\{i:V_{i,t_\mathrm{RD}}=0\}$ and $\widehat{D}_{t_0}(1)$ within the subpopulation $\{i: V_{i,t_\mathrm{RD}}=1\}$, both via local-linear RD, and test whether the two conditional discontinuities are equal.

\paragraph{Constant conditional confounding discontinuity (Assumption~\ref{ass:trend-cell}).}
Assumption~\ref{ass:trend-cell} states that $\alpha_{t_0,0}(v)=\alpha_{t_\mathrm{RD},0}(v)$. However, $\alpha_{t_\mathrm{RD},0}(v)$ is not an observable object. Hence, similar to pre-trend testing in DID, this assumption cannot be empirically validated, and a suggestive test can be conducted when multiple comparison time periods are available. For example, if $\mathcal{T}_0=\{t_1,t_2\}$, and since $D_{t}(v)=\alpha_{t,0}(v)$ for $t\in\mathcal{T}_0$, we can test $\alpha_{t_1,0}(v)=\alpha_{t_2,0}(v)$ via $D_{t_1}(v)=D_{t_2}(v)$, with $v$ the period-$t_\mathrm{RD}$ side.

The tests for Assumptions~\ref{ass:trend-cell} and \ref{ass:homog} are related. Both compare one $D_{t}(v)$ to another $D_{t^\prime}(v^\prime)$. The test for Assumption~\ref{ass:homog} keeps the period constant, $t=t^\prime$, and changes the side in the other period, $v\neq v^\prime$, while the test for Assumption~\ref{ass:trend-cell} keeps the side constant, $v=v^\prime$, and changes the period, $t\neq t^\prime$.

\section{Estimation and inference}\label{sec:estimation}

This section discusses estimation and inference of $\att(t_{\mathrm{RD}})$. Our estimation approach runs a local-linear RD in each time period separately, and then aggregates the resulting discontinuities. We can therefore incorporate contemporary RD estimation methods at each period, and then aggregate the obtained period-specific discontinuities according to the identification formula $\att(t_{\mathrm{RD}})=D_{t_{\mathrm{RD}}}-\sum_{t_0\in\mathcal T_0}\kappa_{t_0}\,D_{t_0}$ given in Corollary~\ref{cor:trend}, covering either constant or linear confounding trends assumptions. This estimation approach accommodates multiple time periods, is not limited to time-constant running variables or panel data, and allows bias correction.\footnote{The estimate-by-period-and-aggregate approach can potentially extend beyond the constant and linear cases; however, we focus on these as they are likely the most common in empirical work.} For inference, we develop the variance under three different sampling schemes: repeated cross sections, a panel with a running variable that is constant over time, and a panel with a running variable that varies over time.

Applications which implement a diff-in-disc design mostly use one regression which pools all time periods together \citep[e.g.,][]{grembi2016fiscal,avdic2018modern}. However, the pooled regression approach does not accommodate contemporary RD methods, such as bias correction or MSE-optimal bandwidths. 
A second approach, presented in \citet{picchetti2024difference}, is to run a local-linear RD on outcomes differenced between an RD period and a comparison period. The difference-RD regression removes the confounding in a single step, and has the advantage that it can implement contemporary RD methods. However, it is limited to constant confounding assumption and a single comparison period, and requires panel data with a time-constant running variable.

\subsection{Point estimation}\label{sec:est-point}

The exposition below builds on \citet{calonico2014robust}. Let $\mathbf{X}_p(r)=\bigl[1,\,(r-c),\,\dots,\,(r-c)^p\bigr]'$. Denote by $K$ the kernel, by $h$ the bandwidth, by $p\geq 0$ the polynomial order of the local RD and by $\blacktriangle\in\{+,-\}$ a side of the cutoff. Common choices are uniform or triangular kernels \citep{cattaneo2020practical}, MSE-optimal bandwidths \citep{imbens2012optimal,calonico2014robust}, and $p=1$ \citep{gelman2019why}. We return to the choice of bandwidth below. First, we fit a separate weighted least squares at each time period $t$ and on each side of the cutoff $c$, 
\begin{align}\label{eq:wls_rd}
\widehat\beta_{t,(+),p}(h)
  &= \arg\min_{\beta}
    \sum_{i:\,R_{i,t}\geq c}
      K\!\left(\frac{R_{i,t}-c}{h}\right)
      \bigl(Y_{i,t}-\beta'\mathbf{X}_p(R_{i,t})\bigr)^2, \nonumber\\
\widehat\beta_{t,(-),p}(h)
  &= \arg\min_{\beta}
    \sum_{i:\,R_{i,t}< c}
      K\!\left(\frac{R_{i,t}-c}{h}\right)
      \bigl(Y_{i,t}-\beta'\mathbf{X}_p(R_{i,t})\bigr)^2
\end{align}
Let $\widehat\beta^{(0)}_{t,(\blacktriangle),p}(h)$ denote the estimated intercept, i.e., the first element in the vector $\widehat\beta_{t,(\blacktriangle),p}(h)$. The period-$t$ discontinuity estimator is $\widehat D_t(h)=\widehat\beta^{(0)}_{t,(+),p}(h)-\widehat\beta^{(0)}_{t,(-),p}(h)$, and Corollary~\ref{cor:trend} gives the aggregated estimator
\begin{equation}\label{eq:estimation_multiple}
  \widehat{\att}(t_{\mathrm{RD}};\{h_t\})
  = \widehat D_{t_{\mathrm{RD}}}(h_{t_{\mathrm{RD}}})
    - \sum_{t_0\in\mathcal{T}_0} \kappa_{t_0}\,\widehat D_{t_0}(h_{t_0}),
\end{equation}
for any weights $\{\kappa_{t_0}\}_{t_0 \in \mathcal{T}_0}$ that are admissible under the identification constraints of Corollary \ref{cor:trend}.

\subsection{Bias correction}
\label{sec:est-bias}
As is well known, $\widehat D_t(h)$ is a biased estimator of $D_t$, with bias that shrinks with smaller $h$. A simple bias-corrected aggregate estimator can be constructed building on the cross-sectional bias-corrected estimator of \citet{calonico2014robust}. To lighten notation, from here on the main text sets $p=1$, and the general case with more technical details is provided in Appendix~\ref{app:est-bias}. Conditional on the period-$t$ running variables, we show in Appendix~\ref{app:est-bias} that the first-order bias of $\widehat D_t(h)$ is
\[
  \mathtt B_t(h)=\frac{h^2}{2}\Bigl[m^{(2)}_{t,(+)}\,B_{t,(+),p}(h)-m^{(2)}_{t,(-)}\,B_{t,(-),p}(h)\Bigr],
\]
where, recall, $m_t$ is the conditional expectation of the observed outcome, $m^{(2)}_{t,(\blacktriangle)}$ is the one-sided limit at $c$ of the second derivative of $m_t$, and $B_{t,(\blacktriangle),p}(h)$ is a function of the running variable and the kernel given explicitly in Appendix~\ref{app:est-bias}. The only unknowns in $\mathtt B_t(h)$ are $m^{(2)}_{t,(\blacktriangle)}$ above and below the cutoff. To estimate these derivatives, we can fit \eqref{eq:wls_rd} with a bandwidth $b$ instead of $h$ and a polynomial order $q\geq2$ instead of $p$. Let $\widehat\beta^{(2)}_{t,(\blacktriangle),q}(b)$, twice the coefficient on $(r-c)^2$, be the resulting estimate of $m^{(2)}_{t,(\blacktriangle)}$. We obtain the estimated bias $\widehat{\mathtt B}_t(h,b)$ by replacing $m^{(2)}_{t,(\blacktriangle)}$ with $\widehat\beta^{(2)}_{t,(\blacktriangle),q}(b)$ in $\mathtt B_t(h)$. Then, let $\widehat D^{\mathrm{BC}}_t(h,b)=\widehat D_t(h)-\widehat{\mathtt B}_t(h,b)$ be that bias-corrected estimator that subtracts the estimated bias. The aggregate bias-corrected estimator is then the conventional one minus the aggregated estimated bias,
\begin{equation}\label{eq:est_agg_bc}
\begin{split}
  \widehat{\att}^{\mathrm{BC}}(t_{\mathrm{RD}};\{h_t\},\{b_t\})
  &= \widehat D^{\mathrm{BC}}_{t_{\mathrm{RD}}}(h_{t_{\mathrm{RD}}},b_{t_{\mathrm{RD}}})
    - \sum_{t_0\in\mathcal{T}_0} \kappa_{t_0}\,\widehat D^{\mathrm{BC}}_{t_0}(h_{t_0},b_{t_0})\\
  &= \widehat{\att}(t_{\mathrm{RD}};\{h_t\})
    -\Bigl[\widehat{\mathtt B}_{t_{\mathrm{RD}}}(h_{t_{\mathrm{RD}}},b_{t_{\mathrm{RD}}})
    -\sum_{t_0\in\mathcal{T}_0} \kappa_{t_0}\,\widehat{\mathtt B}_{t_0}(h_{t_0},b_{t_0})\Bigr].
\end{split}
\end{equation}

\subsection{Inference}\label{sec:est-inf}

The estimator $\widehat{\att}(t_{\mathrm{RD}})$ is a linear combination of the per-period discontinuities, so its variance, $V(t_{\mathrm{RD}})=V\bigl(\widehat{\att}(t_{\mathrm{RD}})\bigr)$, is
\begin{align}\label{eq:agg-var}
  V(t_\mathrm{RD})= V(\widehat D_{t_{\mathrm{RD}}})
  &+ \sum_{t_0\in\mathcal{T}_0}\kappa_{t_0}^2\,V(\widehat D_{t_0})
  - 2\sum_{t_0\in\mathcal{T}_0}\kappa_{t_0}\,\mathrm{Cov}(\widehat D_{t_{\mathrm{RD}}},\widehat D_{t_0})
  \nonumber \\
  &+ \sum_{t_0\in\mathcal{T}_0}\sum_{\substack{s_0\in\mathcal T_0\\s_0\ne t_0}}\kappa_{t_0} \kappa_{s_0}\,\mathrm{Cov}(\widehat D_{t_0},\widehat D_{s_0}),
\end{align}
where variances and covariances are conditional on the running variables of the periods involved, and we suppress the bandwidths and the conditioning in the notation.

Within a period, the two sides of the cutoff use different units, so the two intercepts in $\widehat D_t$ are uncorrelated. Across periods, $\mathrm{Cov}(\widehat D_t,\widehat D_s)$ combines the covariances of the intercepts on the same side and on opposite sides of the cutoff, and which of these are non-zero depends on how the data are sampled. We consider three sampling schemes. In a repeated cross section (CS) each unit is observed once, so all cross-period covariances are zero. In a panel with a time-constant running variable (PC) each unit stays on one side of the cutoff, so only the same-side covariances are non-zero. In a panel with a time-varying running variable (PV) units may switch sides, so all covariances may be non-zero. We write $V^{\mathrm S}(t_{\mathrm{RD}})$ for the variance under sampling scheme $\mathrm S\in\{\mathrm{CS},\mathrm{PC},\mathrm{PV}\}$, and Appendix~\ref{app:est-var} gives the closed form of each variance. 

The next result gives the rate and limit of each variance and covariance term by sampling scheme.
The result requires introducing further notation. Let $\varepsilon_{i,t}$ be the period-$t$ regression error, and let $\sigma^2_{t,(+)}$ and $\sigma^2_{t,(-)}$ be the limits of its conditional variance $V(\varepsilon_{i,t}\mid R_{i,t}=r)$ above and below the cutoff, respectively. In a panel, for sides $\blacktriangle,\blacktriangledown\in\{+,-\}$, let $\sigma_{t,s,(\blacktriangle,\blacktriangledown)}$ be the limit of the conditional covariance $\mathrm{Cov}(\varepsilon_{i,t},\varepsilon_{i,s}\mid R_{i,t}=r,R_{i,s}=r')$ as $r$ approaches $c$ from side $\blacktriangle$ and $r'$ from side $\blacktriangledown$. Under PC, $r=r'$, and if $\blacktriangle\neq\blacktriangledown$ the limit is set to zero. Let $f_t$ be the density of $R_{i,t}$, $f$ the common density under PC, $f_{t,s}$ the joint density of $(R_{i,t},R_{i,s})$, $\eta_t$ the probability that a unit is observed in period $t$ (so $\eta_t=1$ under PC and PV), and $\bar\rho_{t,s}$ the limit of the bandwidth ratio $h_t/h_s$. Finally, $\tilde\omega_{(\blacktriangle),p}(\rho)$ is a constant that depends only on the kernel, the polynomial order $p$ and $\rho>0$, and $\tilde\omega_{(\blacktriangle),p}=\tilde\omega_{(\blacktriangle),p}(1)$, given explicitly in Appendix~\ref{app:est-asy}. Let
\begin{align*}
  \bar V_t&=\frac{\sigma^2_{t,(+)}\,\tilde\omega_{(+),p}+\sigma^2_{t,(-)}\,\tilde\omega_{(-),p}}{f_t(c)},\\
  \bar C^{\mathrm{PC}}_{t,s}&=\frac{\sigma_{t,s,(+,+)}\,\tilde\omega_{(+),p}(\bar\rho_{t,s})+\sigma_{t,s,(-,-)}\,\tilde\omega_{(-),p}(\bar\rho_{t,s})}{f(c)},\\
  \bar C^{\mathrm{PV}}_{t,s}&=\frac{f_{t,s}(c,c)}{f_t(c)\,f_s(c)}\bigl[\sigma_{t,s,(+,+)}+\sigma_{t,s,(-,-)}-\sigma_{t,s,(+,-)}-\sigma_{t,s,(-,+)}\bigr].
\end{align*}

\begin{proposition}\label{prop:cov-orders}
Under the regularity conditions of Assumption~\ref{ass:reg} in Appendix~\ref{app:est}, for periods $t$ and $s$ where $t\neq s$,
\begin{align*}
  nh_t\,V(\widehat D_t)&\to_p\bar V_t/\eta_t>0 && \text{(all schemes)},\\
  \mathrm{Cov}(\widehat D_t,\widehat D_s)&=0 && \text{(CS)},\\
  nh_s\,\mathrm{Cov}(\widehat D_t,\widehat D_s)&\to_p\bar C^{\mathrm{PC}}_{t,s} && \text{(PC)},\\
  n\,\mathrm{Cov}(\widehat D_t,\widehat D_s)&\to_p\bar C^{\mathrm{PV}}_{t,s} && \text{(PV)}.
\end{align*}
\end{proposition}

\noindent
The proof of Proposition~\ref{prop:cov-orders} is in Appendix~\ref{app:est-asy}, and the regularity conditions are stated in Assumption~\ref{ass:reg}. Proposition~\ref{prop:cov-orders} shows that under PC the cross-period covariances are of the same order as the variances, so they matter to first order. Under PV, they are of order $1/n$, smaller than the order $1/(nh_t)$ of the variances, so asymptotically a panel with a time-varying running variable behaves like a repeated cross section. The reason is that a cross-period covariance arises only from units inside the bandwidth in both periods. With a common bandwidth $h$, under PC every unit inside the bandwidth in period $t$ is also inside it in period $s$. Under PV, a unit inside the bandwidth in period $t$ has typically moved by period $s$. As the bandwidth shrinks while these movements stay the same size, it becomes less and less likely that the unit is still inside the bandwidth in period $s$. Fewer units therefore contribute to the covariance than to the variance, and the covariance is of smaller order. 

Turning to variance estimation, under CS we set all covariances to zero, under PC we set the opposite-side covariances to zero and estimate the same-side ones, and under PV we estimate all covariances, as in finite samples, when units move little relative to the bandwidth, the covariances between periods can be non-negligible. 
The variance of the bias-corrected estimator \eqref{eq:est_agg_bc} is derived in Appendix~\ref{app:est-var}, building on the robust variance of \citet{calonico2014robust}.
Inference for the conventional estimator uses $\widehat{\att}(t_{\mathrm{RD}})/\widehat V^{\mathrm S}(t_{\mathrm{RD}})^{1/2}$ and the interval $\widehat{\att}(t_{\mathrm{RD}})\pm z_{1-\alpha/2}\,\widehat V^{\mathrm S}(t_{\mathrm{RD}})^{1/2}$ \citep{hahn2001identification}, and for the bias-corrected estimator $\widehat{\att}^{\mathrm{BC}}(t_{\mathrm{RD}})/\widehat V^{\mathrm{BC,S}}(t_{\mathrm{RD}})^{1/2}$ and $\widehat{\att}^{\mathrm{BC}}(t_{\mathrm{RD}})\pm z_{1-\alpha/2}\,\widehat V^{\mathrm{BC,S}}(t_{\mathrm{RD}})^{1/2}$ \citep{calonico2014robust}, where $\widehat V^{\mathrm S}(t_{\mathrm{RD}})$ and $\widehat V^{\mathrm{BC,S}}(t_{\mathrm{RD}})$ estimate the variances under the relevant sampling scheme (Appendix~\ref{app:est-var}).

\subsection{Bandwidth choice}\label{sec:est-bw}

We consider two possible choices for ``optimal'' bandwidths when estimating $\widehat{\att}(t_{\mathrm{RD}})$: optimizing the per-period bandwidth for the per-period MSE and optimizing a single bandwidth for the joint MSE. Technical details are deferred to Appendices~\ref{app:est-asy}--\ref{app:est-bw}.

The first procedure applies the single-period asymptotic mean squared error (AMSE) optimal bandwidth of \citet{imbens2012optimal,calonico2014robust} to each $\widehat D_t$ on its own. The leading bias of $\widehat D_t$ is the first-order term of $\mathtt B_t(h_t)$ in Section~\ref{sec:est-bias}. As $h_t\to0$, $B_{t,(\blacktriangle),p}(h_t)$ converges in probability to $\tilde\nu_{(\blacktriangle),p}$, a constant that depends only on the kernel and $p$ given explicitly in Appendix~\ref{app:est-asy}, so that
\[
  \mathtt B_t(h_t)=\frac{h_t^2}{2}\,\bar{\mathtt B}_t+o_p(h_t^2),
  \qquad
  \bar{\mathtt B}_t=m^{(2)}_{t,(+)}\,\tilde\nu_{(+),p}-m^{(2)}_{t,(-)}\,\tilde\nu_{(-),p}.
\]
The leading term of the variance of $\widehat D_t$ is $\bar V_t/(n_th_t)$ by Proposition~\ref{prop:cov-orders}, where $n_t$ is the period-$t$ sample size. The AMSE is defined as the sum of the squared leading bias and the leading variance, which, using the above derivations within a single time period, gives $\mathrm{AMSE}_t(h_t)=\frac{h_t^4}{4}\bar{\mathtt B}_t^2+\frac{\bar V_t}{n_th_t}$. The bandwidth that minimizes the per-period AMSE is $h_t^{\mathrm{CCT}}=(\bar V_t/\bar{\mathtt B}_t^2)^{1/5}n_t^{-1/5}$. However, the aim of this per-period bandwidth rule is to optimize the MSE of the observed discontinuity estimator within each period, and not the MSE of the multiple-period aggregate estimator $\widehat{\att}(t_{\mathrm{RD}};\{h_t\})$. 

Therefore, we introduce a new optimal-bandwidth procedure that targets the AMSE of the aggregate $\widehat{\att}(t_{\mathrm{RD}})$ directly, using the same bandwidth $h$ in all periods. Recall that $n$ is the number of units and $\mathrm{S}\in\{\mathrm{CS},\mathrm{PC},\mathrm{PV}\}$ the sampling scheme. The leading bias and variance terms of $\widehat{\att}(t_{\mathrm{RD}})$ are $\tfrac{h^2}{2}\bar{\mathtt B}(t_{\mathrm{RD}})$ and $\bar V^{\mathrm{S}}(t_{\mathrm{RD}})/(nh)$, where
\begin{align*}
  \bar{\mathtt B}(t_{\mathrm{RD}})&=\bar{\mathtt B}_{t_{\mathrm{RD}}}-\sum_{t_0\in\mathcal T_0}\kappa_{t_0}\,\bar{\mathtt B}_{t_0},\\
  \bar V^{\mathrm{S}}(t_{\mathrm{RD}})&=\frac{\bar V_{t_{\mathrm{RD}}}}{\eta_{t_{\mathrm{RD}}}}+\sum_{t_0\in\mathcal T_0}\kappa_{t_0}^2\,\frac{\bar V_{t_0}}{\eta_{t_0}}
  +\ind\{\mathrm S=\mathrm{PC}\}\Bigl(-2\sum_{t_0\in\mathcal T_0}\kappa_{t_0}\,\bar C^{\mathrm{PC}}_{t_{\mathrm{RD}},t_0}+\sum_{t_0\in\mathcal T_0}\sum_{\substack{s_0\in\mathcal T_0\\s_0\neq t_0}}\kappa_{t_0}\kappa_{s_0}\,\bar C^{\mathrm{PC}}_{t_0,s_0}\Bigr),
\end{align*}
with $\bar C^{\mathrm{PC}}_{t,s}$ evaluated at $\bar\rho_{t,s}=1$. The variance term depends on the sampling scheme through the cross-period covariances, which enter under PC only: they are zero under CS and of lower order under PV (Proposition~\ref{prop:cov-orders}).

\begin{proposition}\label{prop:common-h}
Under the regularity conditions of Assumption~\ref{ass:reg} in Appendix~\ref{app:est}, assume that all periods use a common bandwidth $h$ and that the weights $\{\kappa_{t_0}\}_{t_0\in\mathcal T_0}$ are fixed. Then the AMSE of $\widehat{\att}(t_{\mathrm{RD}})$ is
\begin{equation}\label{eq:joint_amse}
\mathrm{AMSE}^{\mathrm{S}}(h)=\tfrac{h^4}{4}\,\bar{\mathtt B}(t_{\mathrm{RD}})^2+\frac{\bar V^{\mathrm{S}}(t_{\mathrm{RD}})}{nh},
\end{equation}
and if $\bar{\mathtt B}(t_{\mathrm{RD}})\neq0$, its unique minimizer over $h>0$ is
\[
h^{\star,\mathrm{S}}(t_{\mathrm{RD}})=\bigl(\bar V^{\mathrm{S}}(t_{\mathrm{RD}})/\bar{\mathtt B}(t_{\mathrm{RD}})^2\bigr)^{1/5}n^{-1/5}.
\]
\end{proposition}

\noindent
Appendix~\ref{app:est-bw} presents the proof. The result extends the single optimal bandwidth rule from a single cross section to the multiple time period setup. 

\subsection{Weighting choice}\label{sec:est-weights}

Any weights that satisfy the constraints of Corollary~\ref{cor:trend} identify $\att(t_{\mathrm{RD}})$, but different weights give different estimators. Denote by $\boldsymbol{\kappa}=(\kappa_{t_0})_{t_0\in\mathcal T_0}$ the weight vector. Below we analyze the weights that minimize the variance of the estimator, which we term the MV weights:
\begin{equation}\label{eq:w-mv}
  \boldsymbol{\kappa}^{\mathrm{MV}}=\arg\min_{\boldsymbol{\kappa}}\ V^{\mathrm S}(t_{\mathrm{RD}}) \quad\text{subject to the constraints of Corollary~\ref{cor:trend}}.
\end{equation}
As the closed form in general requires further notation, we defer it to Lemma~\ref{lem:w-mv} in Appendix~\ref{app:est-weights}. The bias-corrected estimator \eqref{eq:est_agg_bc} uses its own MV weights, which minimize its variance $V^{\mathrm{BC,S}}(t_{\mathrm{RD}})$ subject to the same constraints, and are given by the same closed form with the variances and covariances of the bias-corrected estimators (Appendix~\ref{app:est-var}) in place of the conventional ones.
A simple example is CS with a constant trend, where the MV weights are the inverse-variance weights,
\[
  \kappa^{\mathrm{MV}}_{t_0}=\frac{V(\widehat D_{t_0})^{-1}}{\sum_{s_0\in\mathcal T_0}V(\widehat D_{s_0})^{-1}},
\]
which gives more weight to more precise periods. In the panel case, the MV weights also give more weight to comparison periods with higher covariance with the RD period. 
In general, the MV weights depend on unknown variances and covariances. Let $\widehat{\boldsymbol{\kappa}}^{\mathrm{MV}}$ be the plug-in estimate that replaces the variances and covariances with their estimates, and $\widehat{\att}(t_{\mathrm{RD}};\boldsymbol{\kappa})$ and $V^{\mathrm S}(t_{\mathrm{RD}};\boldsymbol{\kappa})$ be the estimator and its variance with weights $\boldsymbol{\kappa}$, and
\[
  T(\boldsymbol{\kappa})=\frac{\widehat{\att}(t_{\mathrm{RD}};\boldsymbol{\kappa})-\att(t_{\mathrm{RD}})}{V^{\mathrm S}(t_{\mathrm{RD}};\boldsymbol{\kappa})^{1/2}}
\]
be the corresponding $t$-statistic.

\begin{proposition}\label{prop:w-mv}
Under the regularity conditions of Assumption~\ref{ass:reg} in Appendix~\ref{app:est}, suppose that $\att(t_{\mathrm{RD}})=D_{t_{\mathrm{RD}}}-\sum_{t_0\in\mathcal T_0}\kappa_{t_0}D_{t_0}$ for every $\boldsymbol{\kappa}$ satisfying the constraints of Corollary~\ref{cor:trend} (constant or linear case), that all periods use a common bandwidth $h$, and that the estimated variances and covariances are consistent after scaling by $nh$. Then
\[
  \frac{\widehat{\att}(t_{\mathrm{RD}};\widehat{\boldsymbol{\kappa}}^{\mathrm{MV}})-\widehat{\att}(t_{\mathrm{RD}};\boldsymbol{\kappa}^{\mathrm{MV}})}{V^{\mathrm S}(t_{\mathrm{RD}};\boldsymbol{\kappa}^{\mathrm{MV}})^{1/2}}\to_p0,
  \qquad
  \frac{V^{\mathrm S}(t_{\mathrm{RD}};\widehat{\boldsymbol{\kappa}}^{\mathrm{MV}})}{V^{\mathrm S}(t_{\mathrm{RD}};\boldsymbol{\kappa}^{\mathrm{MV}})}\to_p1,
\]
and consequently $T(\widehat{\boldsymbol{\kappa}}^{\mathrm{MV}})-T(\boldsymbol{\kappa}^{\mathrm{MV}})\to_p0$.
\end{proposition}

\noindent
The proof is given in Appendix~\ref{app:est-weights}. Proposition~\ref{prop:w-mv} shows that estimation and inference with the estimated MV weights are asymptotically equivalent to using the infeasible MV weights. The effect of estimating the weights on inference is negligible because all weights that satisfy the constraints give the same combination of the true comparison discontinuities, so an error in the weights multiplies only the estimation error of the discontinuities. 

With the common bandwidth of Section~\ref{sec:est-bw}, $h^{\star,\mathrm{S}}(t_{\mathrm{RD}})$ depends on the weights, and we choose both in three steps: estimate $h^{\star,\mathrm{S}}(t_{\mathrm{RD}})$ at initial weights (equal weights in the constant case, OLS weights in the linear case), estimate $\boldsymbol{\kappa}^{\mathrm{MV}}$ at that bandwidth, and re-estimate $h^{\star,\mathrm{S}}(t_{\mathrm{RD}})$ at these weights. A further step is not needed, as at a common bandwidth the weights $\boldsymbol{\kappa}^{\mathrm{MV}}$ do not depend on $h$ to first order (Appendix~\ref{app:est-weights}, Lemmas~\ref{lem:w-mv-h} and~\ref{lem:three-steps}).

\subsection{Aggregating over RD periods}\label{sec:est-agg}

Having finished stating the theory for estimating the ATT for a single RD period, we turn to presenting an estimator for a weighted average over all RD periods, as discussed in Corollary~\ref{cor:agg}. Recall that $\att^{\phi}$ averages the $\att(t_\mathrm{RD})$s using weights $\phi_{t_\mathrm{RD}}$. Using the single-period estimates of Section~\ref{sec:est-point}, its plug-in estimator is
\begin{equation}\label{eq:est-agg}
  \widehat{\att}^{\phi}(\{h_t\})=\sum_{t_{\mathrm{RD}}\in\mathcal T_{\mathrm{RD}}}\phi_{t_{\mathrm{RD}}}\,\widehat D_{t_{\mathrm{RD}}}(h_{t_{\mathrm{RD}}})-\sum_{t_0\in\mathcal T_0}\kappa^{\phi}_{t_0}\,\widehat D_{t_0}(h_{t_0}).
\end{equation}
Building on Section~\ref{sec:est-bias}, we can construct a bias-corrected $\widehat{\att}^{\phi,\mathrm{BC}}$ by replacing each $\widehat D_t(h_t)$ with $\widehat D^{\mathrm{BC}}_t(h_t,b_t)$. 

The variance of \eqref{eq:est-agg}, denoted $V^{\phi,\mathrm S}$ for sampling scheme $\mathrm S$, is
\begin{align*}
  V^{\phi,\mathrm S}={}&\sum_{t_{\mathrm{RD}}\in\mathcal T_{\mathrm{RD}}}\sum_{s_{\mathrm{RD}}\in\mathcal T_{\mathrm{RD}}}\phi_{t_{\mathrm{RD}}}\phi_{s_{\mathrm{RD}}}\,\mathrm{Cov}(\widehat D_{t_{\mathrm{RD}}},\widehat D_{s_{\mathrm{RD}}})
  -2\sum_{t_{\mathrm{RD}}\in\mathcal T_{\mathrm{RD}}}\sum_{t_0\in\mathcal T_0}\phi_{t_{\mathrm{RD}}}\kappa^{\phi}_{t_0}\,\mathrm{Cov}(\widehat D_{t_{\mathrm{RD}}},\widehat D_{t_0})\\
  &+\sum_{t_0\in\mathcal T_0}\sum_{s_0\in\mathcal T_0}\kappa^{\phi}_{t_0}\kappa^{\phi}_{s_0}\,\mathrm{Cov}(\widehat D_{t_0},\widehat D_{s_0}),
\end{align*}
with $\mathrm{Cov}(\widehat D_t,\widehat D_t)=V(\widehat D_t)$. As before, Proposition~\ref{prop:cov-orders} applies to these variances and covariances. They are estimated as in Section~\ref{sec:est-inf}, which also gives the variance of the bias-corrected estimator. 

The per-period CCT bandwidths do not depend on the target and are unchanged. For the common bandwidth, Proposition~\ref{prop:common-h} holds for $\widehat{\att}^{\phi}$ by the same argument, with $\bar{\mathtt B}(t_{\mathrm{RD}})$ and $\bar V^{\mathrm S}(t_{\mathrm{RD}})$ replaced by
\begin{align*}
  \bar{\mathtt B}^\phi={}&\sum_{t_{\mathrm{RD}}\in\mathcal T_{\mathrm{RD}}}\phi_{t_{\mathrm{RD}}}\,\bar{\mathtt B}_{t_{\mathrm{RD}}}-\sum_{t_0\in\mathcal T_0}\kappa^{\phi}_{t_0}\,\bar{\mathtt B}_{t_0},\\
  \bar V^{\phi,\mathrm S}={}&\sum_{t_{\mathrm{RD}}\in\mathcal T_{\mathrm{RD}}}\phi_{t_{\mathrm{RD}}}^2\,\frac{\bar V_{t_{\mathrm{RD}}}}{\eta_{t_{\mathrm{RD}}}}+\sum_{t_0\in\mathcal T_0}\bigl(\kappa^{\phi}_{t_0}\bigr)^2\,\frac{\bar V_{t_0}}{\eta_{t_0}}
  +\ind\{\mathrm S=\mathrm{PC}\}\Bigl(\sum_{t_{\mathrm{RD}}\in\mathcal T_{\mathrm{RD}}}\sum_{\substack{s_{\mathrm{RD}}\in\mathcal T_{\mathrm{RD}}\\s_{\mathrm{RD}}\neq t_{\mathrm{RD}}}}\phi_{t_{\mathrm{RD}}}\phi_{s_{\mathrm{RD}}}\,\bar C^{\mathrm{PC}}_{t_{\mathrm{RD}},s_{\mathrm{RD}}}\\
  &-2\sum_{t_{\mathrm{RD}}\in\mathcal T_{\mathrm{RD}}}\sum_{t_0\in\mathcal T_0}\phi_{t_{\mathrm{RD}}}\kappa^{\phi}_{t_0}\,\bar C^{\mathrm{PC}}_{t_{\mathrm{RD}},t_0}+\sum_{t_0\in\mathcal T_0}\sum_{\substack{s_0\in\mathcal T_0\\s_0\neq t_0}}\kappa^{\phi}_{t_0}\kappa^{\phi}_{s_0}\,\bar C^{\mathrm{PC}}_{t_0,s_0}\Bigr),
\end{align*}
with $\bar C^{\mathrm{PC}}_{t,s}$ evaluated at $\bar\rho_{t,s}=1$. If $\bar{\mathtt B}^{\phi}\neq0$, the common bandwidth that minimizes the AMSE of $\widehat{\att}^{\phi}$ is $h^{\star,\phi,\mathrm S}=\bigl(\bar V^{\phi,\mathrm S}/(\bar{\mathtt B}^{\phi})^2\bigr)^{1/5}n^{-1/5}$. As the statements here are direct extensions of the theory produced for a single $\mathrm{ATT}(t_\mathrm{RD})$, we do not provide explicit proofs in the appendix for the above statements.

\section{Identification and estimation of the ATU}\label{sec:atu}

Sections~\ref{sec:setup}--\ref{sec:estimation} target the ATT, using comparison periods in which $W$ is uniformly zero. When the comparison periods are such that $W$ is uniformly one, as in the application below, the target estimand that can be identified is the ATU. This section presents a procedure with which the ATU identification and estimation can be derived building on the same arguments for the ATT.

Recall that $\alpha_{t,1}$ is the discontinuity in potential outcomes setting $W=1$, and note that Lemma~\ref{prop:bias} can also be written as $D_t=\atu(t)+\alpha_{t,1}$.
Next, assume that in comparison time periods $t_0\in\mathcal{T}_0$ all units are treated, i.e., $W_{i,t_0}\equiv1$. Hence, units above the cutoff have $(W,V)=(1,1)$ and units below have $(W,V)=(1,0)$, so $D_{t_0}=\alpha_{t_0,1}$. Identification of the ATU following the same argument as in Theorem \ref{thm:id} then follows from connecting $\alpha_{t_\mathrm{RD},1}$ and $\alpha_{t_0,1}$. 

When targeting the ATU  the definitions in Section~\ref{sec:tvrv} should be revised. Namely, define the ATU under a time-varying running variable as an average of the conditional effects weighted by the below-cutoff side shares $\pi_{t,(-)}$, differently from the ATT which uses the above-cutoff side shares. Hence, when targeting the ATU, the composition stability assumption (Assumption~\ref{ass:comp-stable}) restricts the below-cutoff shares, $\pi_{t_\mathrm{RD},(-)}(v)=\pi_{t_0,(-)}(v)$, rather than the above-cutoff shares.

We propose a procedure which "mirrors" the design, which allows to use the theory and methods discussed so far to target the ATU. Let
\[
\tilde R_{i,t}=c-R_{i,t},\qquad \tilde W_{i,t}=1-W_{i,t},\qquad \tilde V_{i,t}=1-V_{i,t},\qquad \tilde Y_{i,t}(\tilde w,\tilde v)=Y_{i,t}(1-\tilde w,1-\tilde v),
\]
with cutoff $\tilde c=0$ and observed variable $\tilde Y_{i,t}=\tilde Y_{i,t}(\tilde W_{i,t},\tilde V_{i,t})$. The proposed procedure results in  identification and estimation as in Sections~\ref{sec:setup}--\ref{sec:estimation} by using the mirrored variables, i.e.,  $\tilde R_{i,t}$ instead of $R_{i,t}$. For example, Theorem \ref{thm:route-stable} on the mirrored variables gives $\atu(t_\mathrm{RD})=D_{t_\mathrm{RD}}-D_{t_0}$.

\section{Empirical illustration}\label{sec:application}

We revisit \citet{grembi2016fiscal}, the canonical application of the diff-in-disc design, to illustrate our theory and methods.
We first briefly introduce the application and map it onto our identification framework. We then show how to implement the proposed estimators, under both a constant and a linear time-trend assumption, compare per-period and common bandwidths, aggregate over RD periods, and construct the minimum-variance weights. Finally, we present the results of the validation tests of the composition assumptions and discuss their interpretation. All estimates are produced via the \href{https://dorleventer.github.io/didisc/}{\texttt{didisc}} \textbf{R} package.

\subsection{Setting, data and design}\label{sec:app-setting}

The units $i$ are Italian municipalities observed annually from 1999 to 2004, and the running variable $R_{i,t}$ is census population. The census population varies over time within municipalities and hence the sampling scheme is a panel with a time-varying running variable (PV sampling scheme). 
Italy's Domestic Stability Pact (DSP) imposed fiscal rules on all municipalities from 1999, and from 2001 municipalities with fewer than 5{,}000 inhabitants were exempted.
That is, the treatment of interest $W$, fiscal debt targets, equals one above the population cutoff $c=5{,}000$ and zero below it in the RD periods $t_\mathrm{RD}\in\mathcal{T}_\mathrm{RD}=\{2001,2002,2003,2004\}$, and equals one for all units in the comparison periods $t_0\in\mathcal{T}_0=\{1999,2000\}$.
However, there exists a confounding treatment $V$ at the same cutoff. 
Italian law ties mayor and executive salaries, among other things, to the same 5{,}000-inhabitant threshold. To summarize, $W_{i,t}=1$ for $t < 2001$ and all $i$, $W_{i,t}=\ind\{R_{i,t} \ge 5{,}000\}$ for $t\ge 2001$, and $V_{i,t}=\ind\{R_{i,t} \ge 5{,}000\}$ for all $i$ and $t$.

Following \citet{grembi2016fiscal} we keep units with a running variable between 3{,}500 and 7{,}000 residents. The census population is observed twice, once for the 1999--2000 period and once again in the 2001--2004 period. 
We keep the \SfiveNmuni{} municipalities observed at least once in 1999--2000 and once in
2001--2004, as these two observations reveal the unit's side of the cutoff in each group of time periods, that is, the unit's pair of sides as in Section~\ref{sec:tvrv-setup}. We limit the sample to units with an observed pair of sides so the tests discussed in
Section~\ref{sec:tvrv-assess} can be illustrated. The outcome is the per-capita fiscal gap.
The comparison periods $\mathcal{T}_0$ are such that all units are uniformly treated. They can be used to identify and estimate $\atu(t_\mathrm{RD})$ for $t_\mathrm{RD}\in\mathcal{T}_\mathrm{RD}$, applying the results of Sections \ref{sec:setup}--\ref{sec:estimation}, using the procedure discussed in Section~\ref{sec:atu}.
The target causal estimand is therefore the local effect of the fiscal rules on fiscal outcomes below the cutoff, holding mayor and executive salaries fixed at values below the cutoff.

\subsection{Estimates}\label{sec:app-estimates}

In this subsection we illustrate how to implement the estimation framework discussed in Section~\ref{sec:estimation} for the application. We begin by comparing the constant-in-time and linear-in-time discontinuity trend assumptions. We also compare optimal bandwidths: per-period CCT bandwidths and the common bandwidth, which aims for MSE optimality of the ATU estimator. Finally, we present estimates that aggregate the ATUs of multiple RD periods, and discuss the choice of weights. 

Table~\ref{tab:app-main} presents the estimates, with years in columns. Panel~A reports the observed outcome discontinuity within each period. 
Panels~B and~C report the diff-in-disc estimates $\widehat{\atu}(t_{\mathrm{RD}})$ for each RD year, using both comparison years. Standard errors (SEs) are computed under the PV sampling scheme. Both conventional and bias-corrected effect estimates are reported, where the SEs of the latter are robust to bias correction (Section~\ref{sec:est-inf}). The last column reports the estimate of average effect over the four RD years, $\widehat{\atu}^{\phi}$, with equal weights for the RD periods (Section~\ref{sec:est-agg}).  

\begin{sidewaystable}[p]
  \centering
  \caption{RD discontinuities and diff-in-disc estimates, fiscal gap per capita.}
  \label{tab:app-main}
  {\small%

\begin{tabular}{l rr rrrr r}
\toprule
 & \multicolumn{2}{c}{Comparison periods} & \multicolumn{4}{c}{RD periods} & Average \\
\cmidrule(lr){2-3}\cmidrule(lr){4-7}\cmidrule(lr){8-8}
   & 1999 & 2000 & 2001 & 2002 & 2003 & 2004 & 2001--04 \\
\midrule
\multicolumn{8}{l}{A. Single-period RD discontinuity, $\widehat D_t$} \\
  Conventional & $56.0$ $(25.7)$ & $81.8$ $(25.9)$ & $-25.9$ $(29.2)$ & $-50.6$ $(30.8)$ & $-73.7$ $(34.6)$ & $-61.8$ $(37.8)$ &  \\
  Bias-corrected & $64.8$ $(30.1)$ & $91.8$ $(29.8)$ & $-34.7$ $(34.7)$ & $-59.2$ $(37.0)$ & $-84.0$ $(41.3)$ & $-73.1$ $(45.1)$ &  \\
  Bandwidth $h_t$ & $445$ & $461$ & $595$ & $606$ & $485$ & $390$ &  \\
  Observations within $h_t$ & $332$ & $352$ & $418$ & $421$ & $345$ & $272$ &  \\
\addlinespace
\multicolumn{8}{l}{B. Diff-in-disc, $\widehat{\atu}(t_{\mathrm{RD}})$, constant confounding trend} \\
\multicolumn{8}{l}{\quad Per-period CCT} \\
  \quad\quad Conventional &  &  & $-94.8$ $(37.8)$ & $-119.5$ $(38.7)$ & $-142.5$ $(41.1)$ & $-130.7$ $(43.1)$ & $-121.9$ $(36.3)$ \\
  \quad\quad Bias-corrected &  &  & $-113.0$ $(43.9)$ & $-137.4$ $(45.1)$ & $-162.3$ $(48.3)$ & $-151.4$ $(50.6)$ & $-141.0$ $(42.4)$ \\
  \quad\quad Bandwidths $(h_{1999},h_{2000},h_{t_{\mathrm{RD}}})$ &  &  & $445,\,461,\,595$ & $445,\,461,\,606$ & $445,\,461,\,485$ & $445,\,461,\,390$ &  \\
\multicolumn{8}{l}{\quad Common $h^{\star}$} \\
  \quad\quad Conventional &  &  & $-120.5$ $(46.7)$ & $-119.7$ $(39.3)$ & $-134.7$ $(39.6)$ & $-130.3$ $(40.5)$ & $-133.8$ $(39.3)$ \\
  \quad\quad Bias-corrected &  &  & $-137.8$ $(52.7)$ & $-134.1$ $(45.3)$ & $-150.8$ $(44.8)$ & $-144.9$ $(46.5)$ & $-146.0$ $(43.5)$ \\
  \quad\quad Common bandwidth $h^{\star}$ &  &  & $389$ & $532$ & $517$ & $460$ & $435$ \\
\addlinespace
\multicolumn{8}{l}{C. Diff-in-disc, $\widehat{\atu}(t_{\mathrm{RD}})$, linear confounding trend} \\
\multicolumn{8}{l}{\quad Per-period CCT} \\
  \quad\quad Conventional &  &  & $-133.6$ $(44.9)$ & $-184.2$ $(55.2)$ & $-233.0$ $(70.1)$ & $-247.1$ $(83.9)$ & $-199.5$ $(60.3)$ \\
  \quad\quad Bias-corrected &  &  & $-153.5$ $(52.0)$ & $-205.0$ $(64.3)$ & $-256.9$ $(81.5)$ & $-273.0$ $(98.1)$ & $-222.1$ $(70.1)$ \\
  \quad\quad Bandwidths $(h_{1999},h_{2000},h_{t_{\mathrm{RD}}})$ &  &  & $445,\,461,\,595$ & $445,\,461,\,606$ & $445,\,461,\,485$ & $445,\,461,\,390$ &  \\
\multicolumn{8}{l}{\quad Common $h^{\star}$} \\
  \quad\quad Conventional &  &  & $-138.0$ $(45.6)$ & $-180.8$ $(51.2)$ & $-225.6$ $(62.0)$ & $-238.2$ $(71.7)$ & $-192.7$ $(54.0)$ \\
  \quad\quad Bias-corrected &  &  & $-161.9$ $(53.6)$ & $-208.9$ $(61.1)$ & $-263.5$ $(74.1)$ & $-280.8$ $(87.0)$ & $-225.3$ $(64.8)$ \\
  \quad\quad Common bandwidth $h^{\star}$ &  &  & $523$ & $633$ & $637$ & $643$ & $635$ \\
\bottomrule
\end{tabular}
}
  \caption*{\footnotesize \emph{Notes:} \emph{Panel A:} local linear RD estimates of the observed outcome discontinuity at $5{,}000$ inhabitants, triangular kernel, first order polynomial, CCT single-period bandwidth. \emph{Panels B--C:} $\widehat{\atu}(t_{\mathrm{RD}})$ from \eqref{eq:est_agg_bc} with comparison periods 1999 and 2000. In panel B the weights are equal on the two comparison discontinuities. In panel C the line through the 1999 and 2000 discontinuities is extrapolated to $t_{\mathrm{RD}}$. Standard errors under the panel with time-varying running variable, $V^{\mathrm{PV}}$ in \eqref{eq:var-pv}. Bias-corrected rows use the robust variance of Appendix~\ref{app:est-var}. Bandwidth rules (Section~\ref{sec:est-bw}): per-period CCT uses the Panel~A bandwidths, and common $h^{\star,\mathrm{S}}(t_{\mathrm{RD}})$ minimizes \eqref{eq:joint_amse} over one bandwidth. \emph{Column Average 2001--04:} $\widehat{\atu}^{\phi}$ from \eqref{eq:est-agg} with equal weights $\phi_{t_{\mathrm{RD}}}=1/4$ (Section~\ref{sec:est-agg}). In this column the per-period CCT rows use the bandwidths from Panel~A. The common bandwidth is $h^{\star,\phi,\mathrm S}$ from Section~\ref{sec:est-agg}.}
\end{sidewaystable}

\textbf{Trend assumption.} Under a constant trend in the confounding discontinuity with equal weights for comparison periods, using the conventional estimator at per-period CCT bandwidths the estimated confounding discontinuity is $(\SsixDA+\SsixDB)/2=\SsixConstAlpha$. The 2004 discontinuity estimate is $\SsixDF$, and hence the ATU estimate is $\SsixDF-\SsixConstAlpha=\SsixATUConstCctF$ (SE $\SsixATUseConstCctF$). 
Under a linear trend, the line through 1999 and 2000 is extrapolated to $t_{\mathrm{RD}}$. 
The two comparison discontinuities, $\SsixDA$ in 1999 and $\SsixDB$ in 2000, imply a slope of $\SsixDB-\SsixDA=\SsixDslope$ per year, so extrapolating four years ahead the confounding discontinuity in 2004 is $\SsixDB+4\times\SsixDslope=\SsixLinAlphaF$,\footnote{Equivalently $5\times\SsixDB-4\times\SsixDA$ with the weights $(-4,5)$.} and the ATU estimate is $\SsixDF-\SsixLinAlphaF=\SsixATULinCctF$ (SE $\SsixATUseLinCctF$). That is, under the linear trend assumption, the ATU estimate is larger in magnitude than the constant-trend estimate, with a standard error that grows with time. We contrast the results under the two assumptions to illustrate the sensitivity of the treatment effect estimates to the assumed time trend of the confounding discontinuity.

\textbf{Bandwidths.} Table~\ref{tab:app-main} reports results under both AMSE rules discussed in Section \ref{sec:est-bw}: per-period CCT, and one common bandwidth $h^{\star,\mathrm{S}}(t_{\mathrm{RD}})$ minimizing the AMSE of the ATU estimator. As examples, under the constant-time trend assumption, CCT within periods 1999, 2000 and 2004 sets the bandwidths to $\SsixhConstCctFA$, $\SsixhConstCctFB$ and $\SsixhConstCctFF$, respectively. Using one common bandwidth to optimize the AMSE of the ATU in 2004 gives a bandwidth of $\SsixhConstJointFF$. However, in the considered application, the estimates of the ATU do not vary substantially between the two optimal-bandwidth rules, as can be seen from comparing the estimated coefficients and standard errors.

\textbf{Aggregating RD periods.} The last column of Table~\ref{tab:app-main} reports $\widehat{\atu}^{\phi}$, the average of the ATUs over 2001--2004 with equal weights $\phi_{t_{\mathrm{RD}}}=1/4$, as discussed in Section~\ref{sec:est-agg}. At per-period CCT bandwidths the estimate is $\SsixAvgConstCct$ (SE $\SsixAvgseConstCct$) under the constant trend and $\SsixAvgLinCct$ (SE $\SsixAvgseLinCct$) under the linear trend. This is the relevant estimand for a question such as: over 2001--2004, what is the average effect of fiscal rules on fiscal outcomes?

\textbf{Weights.} Focusing on the constant discontinuity trend assumption, up to now we have used equal weights for the comparison periods. Section~\ref{sec:est-weights} presents weights that minimize the estimator's variance. Figure~\ref{fig:app-mv} plots the estimated standard error of $\widehat{\atu}(2004)$ by $\kappa_{1999}$, noting that $\kappa_{2000}=1-\kappa_{1999}$, holding fixed the common AMSE-optimal bandwidth estimated with equal weights ($\SsixItHone$). Using the estimator in Section~\ref{sec:est-weights}, the minimum-variance weight on 1999 is $\widehat \kappa^{\mathrm{MV}}_{1999}=\SsixMVhOnew$ for the conventional estimator and $\widehat{\kappa}^{\mathrm{MV,BC}}_{1999}=\SsixMVhOnewbc$ for the bias-corrected one, which produce $\widehat{\atu}(2004)=\SsixMVhOneATU$ (SE $\SsixMVhOneATUse$) and $\SsixMVhOneATUbc$ (SE $\SsixMVhOneATUbcse$), respectively. Using equal weights the conventional and bias-corrected ATU estimates are $\SsixATUConstJointF$ (SE $\SsixATUseConstJointF$) and $\SsixATUbcConstJointF$ (SE $\SsixATUbcseConstJointF$), respectively. The gain in precision is small: most of the variance is that of the 2004 discontinuity, which no choice of comparison weights affects, and the 1999 and 2000 discontinuities are highly correlated, so moving weight between them changes little.

\begin{figure}[t]
  \centering
  \caption{Standard error of $\widehat{\atu}(2004)$ by the weight on 1999}\label{fig:app-mv}
  \includegraphics[width=0.85\linewidth]{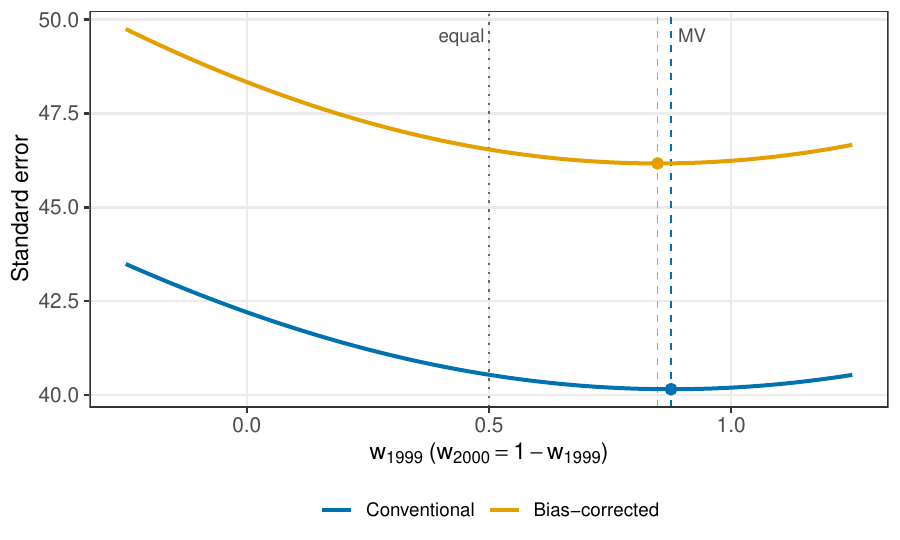}
  \caption*{\footnotesize \emph{Notes:} Constant confounding trend, comparison years 1999 and 2000 with weights $(\kappa_{1999},1-\kappa_{1999})$, per-period CCT bandwidths, standard errors under the PV sampling scheme. Dotted line: equal weights. Dashed lines: the estimated minimum-variance weight of each estimator.}
\end{figure}

\textbf{Choosing both bandwidth and weights.} In Section~\ref{sec:est-weights} we proposed a procedure to estimate both the weights and the bandwidth, in three steps: (i) set initial weights, e.g., equal weights in the constant case, and estimate the common optimal bandwidth, (ii) using the estimated bandwidth, estimate the MV weights, and (iii) using the MV weights, estimate the common optimal bandwidth. The reason we can stop after step three is that, at a common bandwidth, the MV weights are not affected by the bandwidth to first order. 
We now illustrate this procedure, and also illustrate if we were to continue iterating beyond the third step, and find that, as expected, further steps change the weights and the bandwidth only by negligible amounts. 

Consider only the ATU estimator under the constant trend assumption in 2004. Setting the comparison-period weights to be equal, the estimated common AMSE-optimal bandwidth is $\SsixItHone$. Taking this bandwidth as fixed, the estimated MV weights are $\widehat \kappa^{\mathrm{MV}}_{1999}=\SsixItWtwo$ for 1999 and $\widehat \kappa^{\mathrm{MV}}_{2000}=\SsixItWtwoB$ for 2000. Treating these weights as fixed, re-estimating the common AMSE-optimal bandwidth gives $\SsixItHthree$. Our proposed procedure would have stopped here, as we do not expect first-order differences beyond this iteration. A further iteration, treating the bandwidth again as fixed, produces MV weights of $\SsixItWfour$ for 1999 and $\SsixItWfourB$ for 2000, very close to the second-step weights. The weights and bandwidth stop moving (up to small rounding) after two further iterations, which give $\widehat \kappa^{\mathrm{MV}}_{1999}=\SsixItWlim$ and a bandwidth of $\SsixItHlim$, quite close to the second-step weights and the third-step bandwidth.   

\subsection{Composition effects validation tests}\label{sec:app-validation}

As discussed above, the sampling scheme is a panel with a time-varying running variable. Theorem~\ref{thm:route-stable} shows that the above estimates can still be interpreted as the local ATU under additional assumptions: continuous side shares (Assumption~\ref{ass:type-cont}) and constant conditional confounding discontinuity (Assumption~\ref{ass:trend-cell}), together with either composition stability (Assumption~\ref{ass:comp-stable}) or homogeneous confounding (Assumption~\ref{ass:homog}). We test the assumptions as described in Section~\ref{sec:tvrv-assess}.
Figure~\ref{fig:app-validation} visualizes the tests and Table~\ref{tab:app-validation} reports formal test statistics.

\paragraph{Continuous side shares (Assumption~\ref{ass:type-cont}).} In each period we estimate the RD discontinuity in the share of units on the other period's above side, that is, an RD with the side indicator as outcome, as visualized in Figure~\ref{fig:app-val-a7}. The estimated discontinuities are $\SfiveAsevenJumpZZ$ (SE $\SfiveAsevenJumpZZse$, $p=\SfiveAsevenPZZ$) in 2000 and $\SfiveAsevenJumpZF$ (SE $\SfiveAsevenJumpZFse$, $p=\SfiveAsevenPZF$) in 2004. The conventional joint Wald test is close to rejecting at the 10\% level ($\chi^2(2)=\SfiveAsevenWald$, $p=\SfiveAsevenP$), while the bias-corrected version does reject at 10\% ($p=\SfiveAsevenPBC$). Since the assumption must hold in every period, we read this as evidence against continuous side shares, driven by 2000.

\paragraph{Composition stability (Assumption~\ref{ass:comp-stable}).} As the target is the local ATU, the test uses only below-cutoff units, and the outcome is whether the unit is above the cutoff in the other census, as in the side-share test. The reflected RD is shown in Figure~\ref{fig:app-val-a8}, with the period 2000 units as they are to the left of zero and the 2004 units mirrored to the right. The estimated discontinuity is $\SfiveAeightJump$ (SE $\SfiveAeightJumpSE$, $p=\SfiveAeightJumpP$) using the conventional estimator and $\SfiveAeightJumpBC$ (SE $\SfiveAeightJumpBCse$, $p=\SfiveAeightJumpPBC$) under bias correction. The rejection reflects the unbalanced crossings, discussed further below with respect to treatment switchers: many municipalities just below the cutoff in 2000 are above it in 2004, while few just below it in 2004 were above it in 2000, so the side shares just below the cutoff differ between the two years.

\paragraph{Homogeneous confounding (Assumption~\ref{ass:homog}).} In each comparison year we estimate the discontinuity separately within the subsample of municipalities below the 2004 cutoff and within the subsample above it, four discontinuities in all (Figure~\ref{fig:app-val-a9}). In 1999 the estimated outcome discontinuity, using the conventional estimator and single-period-CCT bandwidths, is $\SfiveAnineDAbelow$ (SE $\SfiveAnineDAbelowSE$) for units below the cutoff in 2004 and $\SfiveAnineDAabove$ (SE $\SfiveAnineDAaboveSE$) for units above it. The estimates for 2000 are $\SfiveAnineDBbelow$ (SE $\SfiveAnineDBbelowSE$) below and $\SfiveAnineDBabove$ (SE $\SfiveAnineDBaboveSE$) above. Under Assumption~\ref{ass:homog} the two discontinuities within the same year are equal. The test of equality gives $p=\SfiveAninePA$ in 1999 and $p=\SfiveAninePB$ in 2000, and the joint Wald is $\chi^2(2)=\SfiveAnineWaldJoint$ ($p=\SfiveAninePJoint$). The data do not reject homogeneous confounding at conventional significance levels.

\paragraph{Constant conditional confounding discontinuity (Assumption~\ref{ass:trend-cell}).} The starting point is the four estimated outcome discontinuities from above, two for each comparison period 1999 and 2000, by side (above or below) in 2004. We test whether the estimated outcome discontinuities across comparison periods are equal, holding the 2004 side constant. Testing whether the outcome discontinuities for 1999 and 2000 are equal for units above the cutoff in 2004 gives $p=\SfiveAtenPabove$. The similar test for units below the cutoff in 2004 gives $p=\SfiveAtenPbelow$. The joint Wald test gives $\chi^2(2)=\SfiveAtenWald$, $p=\SfiveAtenP$, and the bias-corrected Wald test gives $p=\SfiveAtenPBC$. That is, the data do not reject a constant conditional confounding discontinuity at conventional significance levels.. 

\paragraph{Summary.} The above illustrates how to conduct validation tests regarding composition effects. With respect to the considered application, as composition stability is rejected while homogeneous confounding is not, under Theorem~\ref{thm:route-stable} the estimates of Section~\ref{sec:app-estimates} rest on the homogeneous-confounding route to identification under a time-varying running variable, that is, Assumptions~\ref{ass:trend-cell}, \ref{ass:type-cont} and~\ref{ass:homog}. However, continuous side shares are rejected in 2000, so the 2000 discontinuity may carry a within-period composition term $S_{2000}$ (Lemma~\ref{prop:sdisc}), and the estimates may be biased. 

\begin{figure}[p]
  \centering
  \caption{Validation tests by assumption}\label{fig:app-validation}
  \begin{subfigure}{\textwidth}
    \centering
    \includegraphics[width=\textwidth]{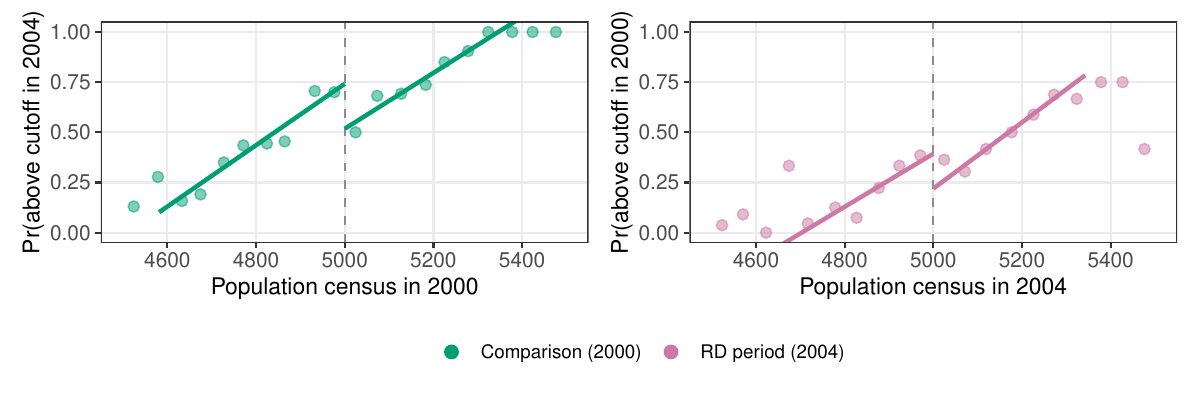}
    \caption{Continuous side shares (Assumption~\ref{ass:type-cont}).}\label{fig:app-val-a7}
  \end{subfigure}\\[1ex]
  \begin{subfigure}{\textwidth}
    \centering
    \includegraphics[width=\textwidth]{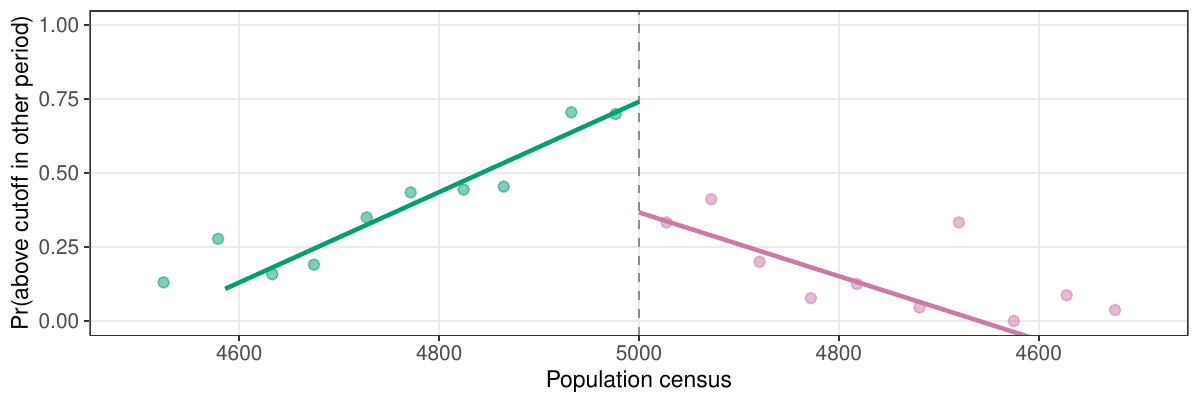}
    \caption{Composition stability (Assumption~\ref{ass:comp-stable}).}\label{fig:app-val-a8}
  \end{subfigure}\\[1ex]
  \begin{subfigure}{\textwidth}
    \centering
    \includegraphics[width=\textwidth]{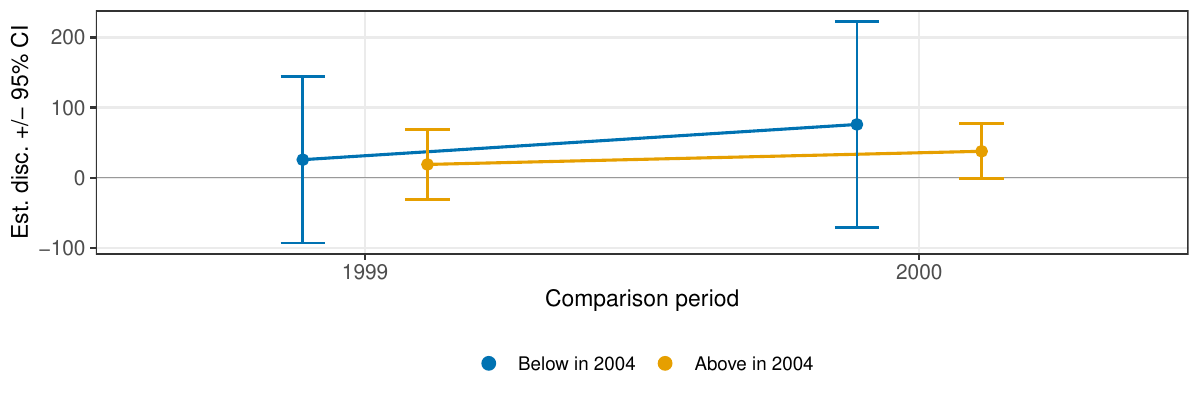}
    \caption{Conditional discontinuities in the comparison years (Assumptions~\ref{ass:homog} and~\ref{ass:trend-cell}).}\label{fig:app-val-a9}
  \end{subfigure}
  \caption*{\footnotesize \emph{Notes:} The figure visualizes validation tests regarding composition effects. \emph{Panel (a):} the left plot shows an RD plot where the outcome variable is whether the unit is above the cutoff in 2004 and the running variable is the municipality population in 2000. The right plot shows the opposite: the outcome is whether the unit is above the cutoff in 2000 and the running variable is the population in 2004. Dots represent means binned in groups of 50. \emph{Panel (b):} below-cutoff units only, as the target is the ATU: to the left of zero the 2000 units as they are, to the right the 2004 units mirrored; the outcome is whether the unit is above the cutoff in the other census, as in panel (a), and the axis shows each group's own census population. \emph{Panel (c):} the discontinuity in years 1999 and 2000 within subsamples of units by whether the unit is above or below the cutoff in 2004. Lines in panels (a) and (b) and estimates in (c) are single cross-section RDs using triangular kernels, first-order polynomials and CCT bandwidths.}
\end{figure}

\begin{table}[t]
  \centering
  \caption{Validation tests.}
  \label{tab:app-validation}
\begin{tabular}{l rr rr}
\toprule
 & \multicolumn{2}{c}{Conventional} & \multicolumn{2}{c}{Bias-corrected} \\
\cmidrule(lr){2-3}\cmidrule(lr){4-5}
 & $\chi^2$(df) & $p$ & $\chi^2$(df) & $p$ \\
\midrule
\multicolumn{5}{l}{\textit{Assumption~\ref{ass:type-cont}: Continuous side shares}} \\
  $2000$ & $3.73\,(1)$ & 0.053 & $4.64\,(1)$ & 0.031 \\
  $2004$ & $1.50\,(1)$ & 0.221 & $2.21\,(1)$ & 0.137 \\
  Joint $(2000,2004)$ & $4.36\,(2)$ & 0.113 & $5.81\,(2)$ & 0.055 \\
\addlinespace
\multicolumn{5}{l}{\textit{Assumption~\ref{ass:comp-stable}: Composition stability}} \\
  $(2000,2004)$ & $10.67\,(1)$ & 0.001 & $6.51\,(1)$ & 0.011 \\
\addlinespace
\multicolumn{5}{l}{\textit{Assumption~\ref{ass:homog}: Homogeneous confounding}} \\
  $(1999\text{ below} = 1999\text{ above})$ & $0.01\,(1)$ & 0.918 & $0.01\,(1)$ & 0.938 \\
  $(2000\text{ below} = 2000\text{ above})$ & $0.24\,(1)$ & 0.621 & $0.08\,(1)$ & 0.784 \\
  Joint & $0.43\,(2)$ & 0.807 & $0.29\,(2)$ & 0.866 \\
\addlinespace
\multicolumn{5}{l}{\textit{Assumption~\ref{ass:trend-cell}: Constant conditional confounding discontinuity}} \\
  $(1999\text{ below} = 2000\text{ below})$ & $1.20\,(1)$ & 0.273 & $0.97\,(1)$ & 0.324 \\
  $(1999\text{ above} = 2000\text{ above})$ & $0.84\,(1)$ & 0.359 & $0.61\,(1)$ & 0.433 \\
  Joint & $2.04\,(2)$ & 0.360 & $1.59\,(2)$ & 0.452 \\
\bottomrule
\end{tabular}
   \caption*{\footnotesize \emph{Notes:} The table reports tests results: Wald statistics and $p$-values from local linear RDs at CCT bandwidths, triangular kernel, first-order polynomials. Each block reports the per-component tests and the joint test. \emph{Continuous side shares:} the RD jump in the other census's above-cutoff share, in 2000 and 2004 and jointly (the cross-year covariance accounts for the shared municipalities). \emph{Composition stability:} the reflected-cutoff RD on the 2000--2004 pair, on the below-cutoff units. \emph{Homogeneous confounding:} in each comparison year, equality of the discontinuity between units below and above the 2004 cutoff, and jointly. \emph{Constant conditional confounding discontinuity:} for each 2004 side, equality of the discontinuity across 1999 and 2000, and jointly. With two comparison years only the constant form is testable. Bias-corrected estimates use the robust variance of Appendix~\ref{app:est-var}.}
\end{table}

\subsection{Dropping switchers} 

A tempting response to a rejection of composition stability is to drop the treatment switchers.
To cross the cutoff between censuses a municipality must lie near it. The median distance of a treatment switcher from the 5{,}000 threshold is $\SfiveSwMed$, compared to $\SfiveNonSwMed$ for non-switchers. This is shown visually in the left panel of Figure~\ref{fig:app-switchers}.
Dropping treatment switchers to stabilize the composition therefore removes mass concentrated at the cutoff and carves a ``hole'' in the running-variable density around the cutoff, resembling a donut-hole RD \citep{barreca2011saving,bajari2011regression}, as seen in the right panel of Figure~\ref{fig:app-switchers}.
This is not specific to this application: when switchers are more likely to be near the cutoff, omitting them will disproportionately lower the density at the cutoff.
If the crossings are also unbalanced, as here ($\SfiveSwBelowAbove$ crossing upward against $\SfiveSwAboveBelow$ downward), omitting treatment switchers can create a density discontinuity: a density test \citep{mccrary2008manipulation} moves from
$p=\SfiveMcKeep$ on the full sample to $p=\SfiveMcOmit$ once switchers are dropped, so that a mechanical drop in density risks being read as manipulation. Hence dropping switchers does not solve the composition problem. It removes the units closest to the cutoff, so the estimates rely on units farther away, as in a donut-hole RD, and it makes a McCrary density test more likely to reject.

\begin{figure}[t]
  \centering
  \caption{Census switchers and the cutoff}\label{fig:app-switchers}
  \includegraphics[width=\linewidth]{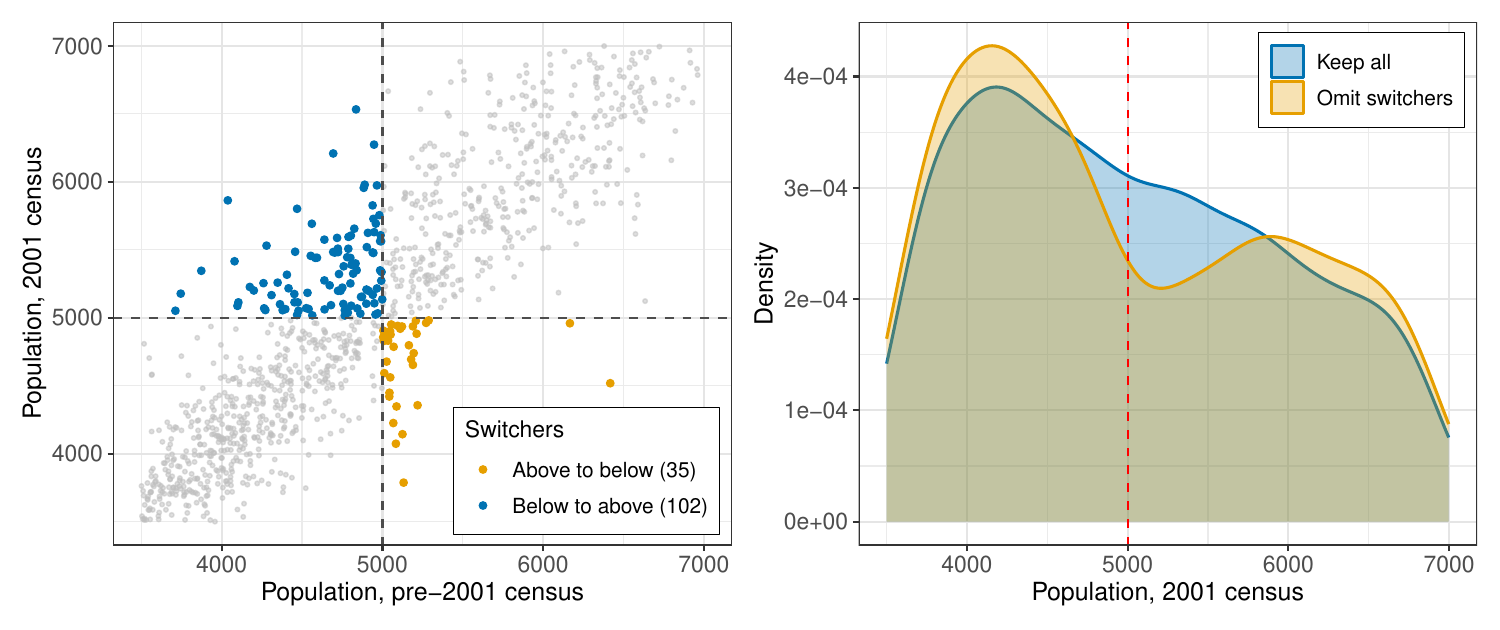}
  \caption*{\footnotesize \emph{Notes:}
    \emph{Left}: each municipality's population in the pre-2001 and 2001 censuses;
    the $\SfiveSwTotal$ switchers (coloured) cross the dashed cutoff lines and cluster at
    $5{,}000$, while non-switchers (grey) spread away from it.
    \emph{Right}: kernel density of 2001-census population keeping all
    municipalities versus omitting the switchers; omission thins the density around
    the cutoff (the donut hole). Here the crossings are unbalanced (more above than
    below), so the hole is lopsided; a McCrary test marginally rejects on the omitted sample
    ($p=\SfiveMcOmit$, versus $p=\SfiveMcKeep$ for the full sample).}

\end{figure}

\section{Conclusion}\label{sec:conclusion}

In this paper, we studied identification and estimation in diff-in-disc designs, in which a confounding treatment is assigned at the same cutoff as the treatment of interest and multiple time periods are observed. We extended identification to several comparison and RD periods and to time trends in the confounding discontinuity. We then analyzed the setup where the running variable varies over time, showed how this introduces composition effects, derived assumptions under which the diff-in-disc identifies the effect of interest, and proposed tests of these assumptions. For estimation, we proposed to estimate the discontinuity in each period and aggregate, and derived the variance of the aggregate estimator under different sampling schemes, a bandwidth common to all periods that minimizes its MSE, and comparison-period weights that minimize its variance. The theory was illustrated on the application of \citet{grembi2016fiscal}.

Although throughout we have discussed the framework as comparing time periods, it can be adapted in a straightforward manner to settings where data can be classified into comparison and RD sets. For example, one can compare a hospital where units are treated according to some sharp RD assignment with a hospital where the RD is not implemented, and no unit, or all units, are treated.

While our formulation focused on the continuity approach to RD, a second common approach to identification in RD designs is the local randomization framework \citep{cattaneo2016inference}. Extending our framework to this approach is left for future research. Our work also assumed that the cutoff is fixed over time. Future research may consider time-varying cutoffs and how they might be used to identify causal parameters. For example, multiple cutoffs might allow identification of treatment effects at different points of the support of the running variable, as discussed in \citet{cattaneo2021extrapolating}.

\appendix
\setcounter{assumption}{0}
\renewcommand{\theassumption}{\thesection.\arabic{assumption}}
\renewcommand{\theHassumption}{\thesection.\arabic{assumption}}
\counterwithin*{lemma}{section}
\renewcommand{\thelemma}{\thesection.\arabic{lemma}}
\renewcommand{\theHlemma}{\thesection.\arabic{lemma}}
\counterwithin*{proposition}{section}
\renewcommand{\theproposition}{\thesection.\arabic{proposition}}
\renewcommand{\theHproposition}{\thesection.\arabic{proposition}}

\section{Proofs}\label{app:proofs}
\begin{proof}[Proof of Lemma~\ref{prop:bias}]
Consider a single cross section where units above the cutoff have $W_{i,t}=V_{i,t}=1$ and units below it have $W_{i,t}=V_{i,t}=0$. Hence $Y_{i,t}=Y_{i,t}(1,1)$ when $R_{i,t}\ge c$ and $Y_{i,t}=Y_{i,t}(0,0)$ when $R_{i,t}<c$, so that $m_{t,(+)}=\mu_{t,(1,1),(+)}$ and $m_{t,(-)}=\mu_{t,(0,0),(-)}$, where the limits exist by Assumption~\ref{ass:relax_cont}. Substiuting these definitions into $D_t$, and adding and subtracting $\mu_{t,(0,1),(+)}$ we obtain
\begin{align*}
  D_t&=m_{t,(+)}-m_{t,(-)}\\
  &=\mu_{t,(1,1),(+)}-\mu_{t,(0,0),(-)}\\
  &=\bigl[\mu_{t,(1,1),(+)}-\mu_{t,(0,1),(+)}\bigr]+\bigl[\mu_{t,(0,1),(+)}-\mu_{t,(0,0),(-)}\bigr]\\
  &=\att(t)+\alpha_{t,0},
\end{align*}
where the last equality follows from the definitions of $\att(t)$ and $\alpha_{t,0}$.
\end{proof}

\begin{proof}[Proof of Proposition~\ref{prop:twoperiod}]
In period $t_0$, units above the cutoff have $(W_{i,t_0},V_{i,t_0})=(0,1)$ and units below it have $(W_{i,t_0},V_{i,t_0})=(0,0)$, so
\begin{align*}
  D_{t_0}&=m_{t_0,(+)}-m_{t_0,(-)}\\
  &=\mu_{t_0,(0,1),(+)}-\mu_{t_0,(0,0),(-)}\\
  &=\alpha_{t_0,0}.
\end{align*}
In period $t_\mathrm{RD}$, Lemma~\ref{prop:bias} gives $D_{t_\mathrm{RD}}=\att(t_\mathrm{RD})+\alpha_{t_\mathrm{RD},0}$. Hence under Assumption~\ref{ass:stable},
\[
  D_{t_\mathrm{RD}}-D_{t_0}=\att(t_\mathrm{RD})+\alpha_{t_\mathrm{RD},0}-\alpha_{t_0,0}=\att(t_\mathrm{RD}).
\]
\end{proof}

\begin{proof}[Proof of Theorem~\ref{thm:id}]
In every period $t$, units above the cutoff have $(\mathbf W_i,\mathbf V_i)=(\mathbf 1_w,\mathbf 1)$ and units below it have $(\mathbf W_i,\mathbf V_i)=(\mathbf 0,\mathbf 0)$, where, recall, $\mathbf 1_w$ is a vector with zeros in the elements which represent comparison periods and ones in the elements which represent RD periods. Hence $D_t=\mu_{t,(\mathbf 1_w,\mathbf 1),(+)}-\mu_{t,(\mathbf 0,\mathbf 0),(-)}$. Adding and subtracting $\mu_{t,(\mathbf 0,\mathbf 1),(+)}$, as in the proof of Lemma~\ref{prop:bias},
\[
  D_t=\bigl[\mu_{t,(\mathbf 1_w,\mathbf 1),(+)}-\mu_{t,(\mathbf 0,\mathbf 1),(+)}\bigr]+\bigl[\mu_{t,(\mathbf 0,\mathbf 1),(+)}-\mu_{t,(\mathbf 0,\mathbf 0),(-)}\bigr]=\att(t)+\alpha_{t,0}.
\]
For $t_0\in\mathcal{T}_0$, Assumption~\ref{ass:noant} gives $D_{t_0}=\alpha_{t_0,0}$. Hence, under Assumption~\ref{ass:trend}
\[
  D_{t_{\mathrm{RD}}}-g_{0,t_{\mathrm{RD}}}\big(\{D_{t_0}:t_0\in\mathcal{T}_0\}\big)
  =\att(t_{\mathrm{RD}})+\alpha_{t_{\mathrm{RD}},0}-g_{0,t_{\mathrm{RD}}}\big(\{\alpha_{t_0,0}:t_0\in\mathcal{T}_0\}\big)
  =\att(t_{\mathrm{RD}}).
\]
\end{proof}

\begin{proof}[Proof of Corollary~\ref{cor:trend}]
If $\alpha_{t,0}\equiv\theta_1$ for some constant $\theta_1$ and $\sum_{t_0\in\mathcal{T}_0} \kappa_{t_0} = 1$, then $\sum_{t_0} \kappa_{t_0}\alpha_{t_0,0}=\theta_1=\alpha_{t_{\mathrm{RD}},0}$.
If $\alpha_{t,0}=\theta_1+\theta_2t$ for some constants $\theta_1,\theta_2$, $\sum_{t_0\in\mathcal{T}_0} \kappa_{t_0} = 1$, and $\sum_{t_0\in\mathcal{T}_0} \kappa_{t_0}\, t_0 = t_{\mathrm{RD}}$, then $\sum_{t_0} \kappa_{t_0}\alpha_{t_0,0}=\theta_1\sum_{t_0} \kappa_{t_0}+\theta_2\sum_{t_0} \kappa_{t_0} t_0=\theta_1+\theta_2\,t_{\mathrm{RD}}=\alpha_{t_{\mathrm{RD}},0}$. In both cases Assumption~\ref{ass:trend} holds with $g_{0,t_{\mathrm{RD}}}\big(\{\alpha_{t_0,0}\}\big)=\sum_{t_0} \kappa_{t_0}\alpha_{t_0,0}$, and Theorem~\ref{thm:id} gives the result in the corollary.
\end{proof}

\begin{proof}[Proof of Corollary~\ref{cor:agg}]
By Corollary~\ref{cor:trend}, for each $t_{\mathrm{RD}}\in\mathcal T_{\mathrm{RD}}$, $\att(t_{\mathrm{RD}})=D_{t_{\mathrm{RD}}}-\sum_{t_0\in\mathcal T_0}\kappa_{t_0}(t_{\mathrm{RD}})\,D_{t_0}$. Hence
\[
  \att^{\phi}=\sum_{t_{\mathrm{RD}}\in\mathcal T_{\mathrm{RD}}}\phi_{t_{\mathrm{RD}}}\att(t_{\mathrm{RD}})
  =\sum_{t_{\mathrm{RD}}\in\mathcal T_{\mathrm{RD}}}\phi_{t_{\mathrm{RD}}}D_{t_{\mathrm{RD}}}-\sum_{t_{\mathrm{RD}}\in\mathcal T_{\mathrm{RD}}}\phi_{t_{\mathrm{RD}}}\sum_{t_0\in\mathcal T_0}\kappa_{t_0}(t_{\mathrm{RD}})\,D_{t_0}.
\]
Exchanging the order of summation in the last term,
\[
  \sum_{t_{\mathrm{RD}}\in\mathcal T_{\mathrm{RD}}}\phi_{t_{\mathrm{RD}}}\sum_{t_0\in\mathcal T_0}\kappa_{t_0}(t_{\mathrm{RD}})\,D_{t_0}
  =\sum_{t_0\in\mathcal T_0}D_{t_0}\sum_{t_{\mathrm{RD}}\in\mathcal T_{\mathrm{RD}}}\phi_{t_{\mathrm{RD}}}\kappa_{t_0}(t_{\mathrm{RD}})
  =\sum_{t_0\in\mathcal T_0}\kappa^{\phi}_{t_0}D_{t_0},
\]
which gives the result.
\end{proof}

\begin{proof}[Proof of Lemma~\ref{prop:sdisc}]
By the law of total expectation,
\[
  D_t=m_{t,(+)}-m_{t,(-)}=\sum_{v\in\{0,1\}}\pi_{t,(+)}(v)\,m_{t,(+)}(v)-\sum_{v\in\{0,1\}}\pi_{t,(-)}(v)\,m_{t,(-)}(v).
\]
Adding and subtracting $\sum_{v}\pi_{t,(+)}(v)\,m_{t,(-)}(v)$,
\begin{equation}\label{eq:sdisc-split}
  D_t=\sum_{v\in\{0,1\}}\pi_{t,(+)}(v)\,D_t(v)+\sum_{v\in\{0,1\}}\bigl[\pi_{t,(+)}(v)-\pi_{t,(-)}(v)\bigr]m_{t,(-)}(v).
\end{equation}
For the second sum, substituting $\pi_{t,(+)}(0)=1-\pi_{t,(+)}(1)$, $\pi_{t,(-)}(0)=1-\pi_{t,(-)}(1)$, and $m_{t,(-)}(v)=\mu_{t,(0,0),(-)}(v)$,
\[
\begin{aligned}
  \bigl[\pi_{t,(+)}(1)-\pi_{t,(-)}(1)\bigr]m_{t,(-)}(1)&+\bigl[\pi_{t,(+)}(0)-\pi_{t,(-)}(0)\bigr]m_{t,(-)}(0)\\
  &=\bigl[\pi_{t,(+)}(1)-\pi_{t,(-)}(1)\bigr]\bigl[\mu_{t,(0,0),(-)}(1)-\mu_{t,(0,0),(-)}(0)\bigr]\equiv S_t.
\end{aligned}
\]
Next we turn to the first sum in \eqref{eq:sdisc-split}. In period $t_0$, using that the conditional observed means are equal to the conditional potential-outcome means at the observed treatments,
\[
  D_{t_0}(v)=m_{t_0,(+)}(v)-m_{t_0,(-)}(v)=\mu_{t_0,(0,1),(+)}(v)-\mu_{t_0,(0,0),(-)}(v)=\alpha_{t_0,0}(v).
\]
In period $t_\mathrm{RD}$, adding and subtracting $\mu_{t_\mathrm{RD},(0,1),(+)}(v)$,
\[
\begin{aligned}
  D_{t_\mathrm{RD}}(v)&=m_{t_\mathrm{RD},(+)}(v)-m_{t_\mathrm{RD},(-)}(v)\\
  &=\mu_{t_\mathrm{RD},(1,1),(+)}(v)-\mu_{t_\mathrm{RD},(0,0),(-)}(v)\\
  &=\bigl[\mu_{t_\mathrm{RD},(1,1),(+)}(v)-\mu_{t_\mathrm{RD},(0,1),(+)}(v)\bigr]+\bigl[\mu_{t_\mathrm{RD},(0,1),(+)}(v)-\mu_{t_\mathrm{RD},(0,0),(-)}(v)\bigr]\\
  &=\att(t_\mathrm{RD}\mid v)+\alpha_{t_\mathrm{RD},0}(v).
\end{aligned}
\]
Averaging the above $D_t(v)$ with $\pi_{t,(+)}(v)$ and substituting definitions shows that the first sum in \eqref{eq:sdisc-split} is $\bar\alpha_{t_0,0}$ in period $t_0$ and $\att(t_\mathrm{RD})+\bar\alpha_{t_\mathrm{RD},0}$ in period $t_\mathrm{RD}$.
\end{proof}

\begin{proof}[Proof of Lemma~\ref{prop:cdisc}]
By Section~\ref{sec:tvrv-problem}, in each period
\[
  \bar\alpha_{t,0}=\alpha_{t,0}(0)+\pi_{t,(+)}(1)\bigl[\alpha_{t,0}(1)-\alpha_{t,0}(0)\bigr].
\]
Subtract the period-$t_0$ expression from the period-$t_\mathrm{RD}$ expression, and add and subtract $\pi_{t_\mathrm{RD},(+)}(1)\bigl[\alpha_{t_0,0}(1)-\alpha_{t_0,0}(0)\bigr]$:
\[
\begin{aligned}
  \bar\alpha_{t_\mathrm{RD},0}-\bar\alpha_{t_0,0}={}&\bigl[\pi_{t_\mathrm{RD},(+)}(1)-\pi_{t_0,(+)}(1)\bigr]\bigl[\alpha_{t_0,0}(1)-\alpha_{t_0,0}(0)\bigr]\\
  &+\bigl[1-\pi_{t_\mathrm{RD},(+)}(1)\bigr]\bigl[\alpha_{t_\mathrm{RD},0}(0)-\alpha_{t_0,0}(0)\bigr]+\pi_{t_\mathrm{RD},(+)}(1)\bigl[\alpha_{t_\mathrm{RD},0}(1)-\alpha_{t_0,0}(1)\bigr].
\end{aligned}
\]
The first line is $C_{t_0,t_\mathrm{RD}}$. Since $\pi_{t_\mathrm{RD},(+)}(0)=1-\pi_{t_\mathrm{RD},(+)}(1)$, the second line is $T_{t_0,t_\mathrm{RD}}$.
\end{proof}

\begin{proof}[Proof of Theorem~\ref{thm:route-stable}]
Assumption~\ref{ass:trend-cell} implies $\alpha_{t_\mathrm{RD},0}(v)-\alpha_{t_0,0}(v)=0$ for $v\in\{0,1\}$, so
\[
  T_{t_0,t_{\mathrm{RD}}}=\sum_{v\in\{0,1\}}\pi_{t_\mathrm{RD},(+)}(v)\bigl[\alpha_{t_\mathrm{RD},0}(v)-\alpha_{t_0,0}(v)\bigr]=\sum_{v\in\{0,1\}}\pi_{t_\mathrm{RD},(+)}(v)\cdot0=0.
\]
Assumption~\ref{ass:type-cont} implies $\pi_{t,(+)}(1)-\pi_{t,(-)}(1)=0$, so for $t\in\{t_0,t_\mathrm{RD}\}$,
\[
  S_t=\bigl[\pi_{t,(+)}(1)-\pi_{t,(-)}(1)\bigr]\bigl[\mu_{t,(0,0),(-)}(1)-\mu_{t,(0,0),(-)}(0)\bigr]=0\cdot\bigl[\mu_{t,(0,0),(-)}(1)-\mu_{t,(0,0),(-)}(0)\bigr]=0.
\]
Assumption~\ref{ass:comp-stable} implies $\pi_{t_\mathrm{RD},(+)}(1)-\pi_{t_0,(+)}(1)=0$, so the first factor of $C_{t_0,t_\mathrm{RD}}$ is zero. Assumption~\ref{ass:homog} implies $\alpha_{t_0,0}(1)-\alpha_{t_0,0}(0)=0$, so the second factor of $C_{t_0,t_\mathrm{RD}}$ is zero. Hence, under either assumption,
\[
  C_{t_0,t_{\mathrm{RD}}}=\bigl[\pi_{t_\mathrm{RD},(+)}(1)-\pi_{t_0,(+)}(1)\bigr]\bigl[\alpha_{t_0,0}(1)-\alpha_{t_0,0}(0)\bigr]=0.
\]
Substituting $S_{t_0}=S_{t_\mathrm{RD}}=C_{t_0,t_\mathrm{RD}}=T_{t_0,t_\mathrm{RD}}=0$ into \eqref{eq:anatomy},
\[
  \att(t_{\mathrm{RD}})=\bigl([D_{t_{\mathrm{RD}}}-0]-[D_{t_0}-0]\bigr)-(0+0)=D_{t_{\mathrm{RD}}}-D_{t_0}.
\]
\end{proof}

\section{Estimation and inference}\label{app:est}

This appendix derives the bias and variance of the estimators of Section~\ref{sec:estimation}, the MSE-optimal bandwidths, and the variance-minimizing weights. We begin with some notation, setup and regularity conditions. 

\paragraph{Notation conventions.} Throughout this appendix, bold letters denote vectors and matrices, and their entries are plain. A hat marks an estimator. A bar marks a probability limit after scaling by the relevant rate: for example, $\bar V_t$ is the limit of $n_th_t\,V(\widehat D_t\mid\mathbf R_t)$ and $\bar{\boldsymbol\Xi}$ is the limit of $nh\,\boldsymbol\Xi$. A tilde marks an integral of the kernel, which depends only on the kernel $K$ and the polynomial order $p$. A check marks a design matrix rescaled by the bandwidth. The typewriter letter $\mathtt B$ denotes a bias term, and calligraphic letters denote sets. Following \citet{calonico2014robust}, the coefficient vectors $\beta$ and the basis vectors $e_{v,p}$ are not bold. The indicator function is written $\ind\{\cdot\}$.

\paragraph{Units, periods and sampling schemes.} Assume there are $n$ units, indexed by $i=1,\dots,n$, each observed in at least one period. For each period $t$, let $\mathcal{N}_t$ be the set of units observed in $t$, of size $n_t=|\mathcal{N}_t|$, and for $t\neq s$ let $\mathcal{N}_{t,s}=\mathcal{N}_t\cap\mathcal{N}_s$ be the set observed in both, of size $n_{t,s}=|\mathcal{N}_{t,s}|$.
We consider three sampling schemes, indexed by $\mathrm{S}\in\{\mathrm{CS},\mathrm{PC},\mathrm{PV}\}$. Under a repeated cross section (CS), each unit is observed in exactly one period, so $\mathcal{N}_{t,s}=\varnothing$ for $t\ne s$. Under a panel, every unit is observed in every period, $\mathcal{N}_t=\{1,\dots,n\}$ for all $t$. For a panel with a time-constant running variable (PC), $R_{i,t}=R_i$ for all $t$, and with a time-varying running variable (PV), $R_{i,t}$ may differ within $i$ across $t$. 

\paragraph{Local polynomial RD in matrix notation.} We develop here the matrix form of the local polynomial RD estimator, building on the development for a single cross section in \citet{calonico2014robust}. Throughout, for brevity we abuse notation and suppress arguments. The terms $\blacktriangle,\blacktriangledown\in\{+,-\}$ index the side of the cutoff, $(+)$ above and $(-)$ below. Period $t$ uses the bandwidth $h_t$ and the bias-correction bandwidth $b_t$. The kernel $K$ and the polynomial orders $p$ and $q\ge p+1$ are assumed the same in every period.

Let $\mathbf{Y}_t=(Y_{1,t},\dots,Y_{n,t})'$ and $\mathbf{R}_t=(R_{1,t},\dots,R_{n,t})'$ be the $n\times 1$ vectors of period-$t$ outcomes and running variables. For $i\notin\mathcal N_t$ the entries $Y_{i,t}$ and $R_{i,t}$ are arbitrarily set to zero; below we set their regression weights to zero as well. Conditioning on $\mathbf{R}_t$ is understood throughout to include the membership indicators $\ind\{i\in\mathcal{N}_t\}$, $i=1,\dots,n$. Let $f_t$ denote the density of $R_{i,t}$ and $f_{t,s}$ denote the joint density of $(R_{i,t},R_{i,s})$ for $t\ne s$; under PC, where $R_{i,t}=R_i$, write $f=f_t$ for the common density. Recall that $\mathbf{X}_p(r)=\bigl[1,\,(r-c),\,\dots,\,(r-c)^p\bigr]'$, and define $\mathbf{X}_{t,p}$ as the $n\times(p+1)$ matrix with $i$-th row $\mathbf{X}_p(R_{i,t})'$. Let $K_h(u)=K(u/h)/h$ for kernel $K$, and
\[
  \mathbf{A}_{t,(+)}(h_t)
  = \mathrm{diag}\!\Bigl(
      \ind\{i\in\mathcal{N}_t\}\,\ind\{R_{i,t}\ge c\}\,K_{h_t}(R_{i,t}-c)
    \Bigr)_{i=1}^{n},
\]
and analogously $\mathbf{A}_{t,(-)}(h_t)$ with $\ind\{R_{i,t}<c\}$. For $\blacktriangle\in\{+,-\}$ define $\boldsymbol{\Gamma}_{t,(\blacktriangle),p}(h_t)=\tfrac{1}{n}\mathbf{X}_{t,p}'\mathbf{A}_{t,(\blacktriangle)}(h_t)\mathbf{X}_{t,p}$. The one-sided local polynomial RD estimator from \eqref{eq:wls_rd} can be re-written as
\[
  \widehat\beta_{t,(\blacktriangle),p}(h_t)
  = \tfrac{1}{n}\,
    \boldsymbol{\Gamma}_{t,(\blacktriangle),p}(h_t)^{-1}\,
    \mathbf{X}_{t,p}'\,\mathbf{A}_{t,(\blacktriangle)}(h_t)\,
    \mathbf{Y}_t,
\]
with $\widehat\beta^{(v)}_{t,(\blacktriangle),p}(h_t)=v!\,e_{v,p}'\,\widehat\beta_{t,(\blacktriangle),p}(h_t)$, where $e_{v,p}$ is the $(v+1)$-th standard basis vector in $\mathbb{R}^{p+1}$.

\paragraph{Additional notation.} Recall that $m_t(r)=\E[Y_{i,t}\mid R_{i,t}=r]$ is the conditional mean of the observed outcome, with one-sided limits $m_{t,(+)}$ and $m_{t,(-)}$ at the cutoff. Let $m^{(v)}_t$ denote the $v$-th derivative of $m_t$, with one-sided limits $m^{(v)}_{t,(\blacktriangle)}$ at $c$, so that the limits $m_{t,(\blacktriangle)}$ defined above are the case $v=0$. Throughout, $\varepsilon_{i,t}=Y_{i,t}-m_t(R_{i,t})$ is the period-$t$ regression error. Let $\sigma^2_t(r)=V(\varepsilon_{i,t}\mid R_{i,t}=r)$ be the conditional error variance, with one-sided limits $\sigma^2_{t,(+)}=\lim_{r\downarrow c}\sigma^2_t(r)$ and $\sigma^2_{t,(-)}=\lim_{r\uparrow c}\sigma^2_t(r)$. Under a panel, let $\sigma_{t,s}(r,r')=\mathrm{Cov}(\varepsilon_{i,t},\varepsilon_{i,s}\mid R_{i,t}=r,R_{i,s}=r')$ be the conditional cross-period error covariance for $t\ne s$, and let
\[
\begin{aligned}
  \sigma_{t,s,(+,+)}&=\lim_{r\downarrow c,\;r'\downarrow c}\sigma_{t,s}(r,r'), \qquad
  \sigma_{t,s,(+,-)}=\lim_{r\downarrow c,\;r'\uparrow c}\sigma_{t,s}(r,r'), \\
  \sigma_{t,s,(-,+)}&=\lim_{r\uparrow c,\;r'\downarrow c}\sigma_{t,s}(r,r'), \qquad
  \sigma_{t,s,(-,-)}=\lim_{r\uparrow c,\;r'\uparrow c}\sigma_{t,s}(r,r'),
\end{aligned}
\]
be the four one-sided limits of the conditional cross-period covariance at $(c,c)$. Under PC, $R_{i,t}=R_{i,s}$, so $\sigma_{t,s}(r,r')$ is defined on $r=r'$ only, and $\sigma_{t,s,(+,+)}=\lim_{r\downarrow c}\sigma_{t,s}(r,r)$ and $\sigma_{t,s,(-,-)}=\lim_{r\uparrow c}\sigma_{t,s}(r,r)$. Finally, let $\eta_t=P(i\in\mathcal{N}_t)$ denote the probability that a unit is observed in period $t$, and for $t\ne s$, when the limit exists, write $\bar\rho_{t,s}=\lim_{n\to\infty}h_t/h_s$ for the limiting bandwidth ratio. The population objects $f_t$, $m_t$, and $\sigma^2_t$ refer to the distribution of $(R_{i,t},Y_{i,t})$ conditional on $i\in\mathcal{N}_t$, and $f_{t,s}$ and $\sigma_{t,s}$ to the distribution conditional on $i\in\mathcal{N}_{t,s}$.

\paragraph{The ATT estimator in matrix form.} Recall that $\widehat D_t = \widehat\beta^{(0)}_{t,(+),p}(h_t)-\widehat\beta^{(0)}_{t,(-),p}(h_t)$ is the outcome discontinuity estimator in period $t$. Denote by $\widehat{\mathbf D}_0=(\widehat D_{t_0})_{t_0\in\mathcal T_0}$ the vector of comparison-period estimators, with the comparison periods in increasing order, by $\mathbf D_0=(D_{t_0})_{t_0\in\mathcal T_0}$ the corresponding outcome discontinuities, by $\widehat{\mathbf D}=(\widehat D_{t_{\mathrm{RD}}},\widehat{\mathbf D}_0')'$ the vector of all per-period estimators with the RD period first, and by $\boldsymbol{\kappa}=(\kappa_{t_0})_{t_0\in\mathcal T_0}$ the weight vector. Partition the covariance matrix of $\widehat{\mathbf D}$ to
\[
  V(\widehat{\mathbf D})=\begin{pmatrix}V(\widehat D_{t_{\mathrm{RD}}})&\boldsymbol\xi'\\ \boldsymbol\xi&\boldsymbol\Xi\end{pmatrix},\qquad
  \boldsymbol\Xi=V(\widehat{\mathbf D}_0),\quad \boldsymbol\xi=\mathrm{Cov}(\widehat{\mathbf D}_0,\widehat D_{t_{\mathrm{RD}}}).
\]
Throughout, variances and covariances are conditional on the running variables of the periods involved, as in Section~\ref{sec:est-inf}, and we suppress the conditioning in the notation.
Using these definitions, \eqref{eq:estimation_multiple} can be written as $\widehat{\att}(t_{\mathrm{RD}};\boldsymbol{\kappa})=\widehat D_{t_{\mathrm{RD}}}-\boldsymbol{\kappa}'\widehat{\mathbf D}_0$. Similarly, \eqref{eq:agg-var}, the variance of \eqref{eq:var-cs}--\eqref{eq:var-pv} under sampling scheme $\mathrm S$, can be written as $V^{\mathrm S}(t_{\mathrm{RD}};\boldsymbol{\kappa})=V(\widehat D_{t_{\mathrm{RD}}})-2\boldsymbol{\kappa}'\boldsymbol\xi+\boldsymbol{\kappa}'\boldsymbol\Xi\boldsymbol{\kappa}$. When the weights are fixed we drop the argument $\boldsymbol{\kappa}$ and write $\widehat{\att}(t_{\mathrm{RD}})$ and $V^{\mathrm S}(t_{\mathrm{RD}})$.

\paragraph{Regularity conditions.} The following regularity conditions are maintained throughout the appendix.
\begin{assumption}[Regularity conditions]\label{ass:reg}
$\;$
\begin{enumerate}
\item[(a)] \emph{Sampling.} The vectors $\bigl(Y_{i,t},R_{i,t},\ind\{i\in\mathcal{N}_t\}\bigr)_{t\in\mathcal{T}_0\cup\mathcal{T}_{\mathrm{RD}}}$, $i=1,\dots,n$, are independent and identically distributed.
\item[(b)] \emph{Kernel.} $K$ is bounded, nonnegative, has compact support, is continuous on its support, and is positive in a neighborhood of zero.
\item[(c)] \emph{Density.} For each $t$, $f_t$ exists and is continuous and positive at $c$. Under PV, $f_{t,s}$ exists and is continuous at $(c,c)$ for $t\ne s$.
\item[(d)] \emph{Regression functions.} For each $t$, $m_t$ is $(q+1)$-times continuously differentiable on each side of $c$. $m^{(v)}_{t,(\blacktriangle)}$ exist and are finite for $v\le q+1$.
\item[(e)] \emph{Errors.} $\E[\varepsilon_{i,t}^4\mid R_{i,t}]$ is bounded. The limits $\sigma^2_{t,(\blacktriangle)}$ exist and are finite. Under PC, for $t\ne s$, $\sigma_{t,s,(+,+)}$ and $\sigma_{t,s,(-,-)}$ exist and are finite. Under PV, for $t\ne s$, $\sigma_{t,s,(+,+)}$, $\sigma_{t,s,(+,-)}$, $\sigma_{t,s,(-,+)}$, $\sigma_{t,s,(-,-)}$ exist and are finite.
\item[(f)] \emph{Non-degeneracy.} For each period $t$ and side $\blacktriangle\in\{+,-\}$, $\sigma^2_{t,(\blacktriangle)}>0$. Under PC, for each side $\blacktriangle$, the matrix $\boldsymbol\Sigma_{(\blacktriangle)}=[\sigma_{t,s,(\blacktriangle,\blacktriangle)}]_{t,s\in\mathcal T_0\cup\mathcal T_{\mathrm{RD}}}$, where $\sigma_{t,t,(\blacktriangle,\blacktriangle)}=\sigma^2_{t,(\blacktriangle)}$, is positive definite.
\item[(g)] \emph{Rates.} For each period $t$, $\eta_t>0$, $h_t\to0$, $n_th_t\to\infty$, $n_th_t^{2p+3}=O(1)$, $b_t\to0$, and $h_t/b_t$ converges to a limit in $[0,\infty)$. Under PC, for $t\ne s$, $\bar\rho_{t,s}$ exists and $\bar\rho_{t,s}\in(0,\infty)$. Under PV, $n_{t,s}h_th_s\to\infty$ for $t\ne s$.
\end{enumerate}
\end{assumption}
\noindent
These conditions are similar to the conventional regularity conditions in RD analysis. The main additions are due to extending these conditions to the multiple time period setup. (a) requires independence across units, but allows dependence across periods within a unit. (c) adds continuity of the joint density function to the conventional single cross-section assumption. (e) adds existence of limits of error covariances between periods, and (f) adds that, under PC, the matrix of error covariances across periods at the cutoff is positive definite. (g) adds that each period is observed with positive probability, $\eta_t>0$, that each bandwidth is no wider than the AMSE-optimal order, $n_th_t^{2p+3}=O(1)$, and conditions on relationships of bandwidths across periods. Following \citet{calonico2014robust}, (g) allows the bias-correction bandwidth $b_t$ to shrink at the same rate as $h_t$ or more slowly. 

\subsection{First-order bias}\label{app:est-bias}

Fix period $t$ and a side $\blacktriangle\in\{+,-\}$. Let $z_{i,t}(h_t)=(R_{i,t}-c)/h_t$ be the scaled running variable, $\mathbf{Z}_{t,p}(h_t)=(z_{1,t}(h_t)^{p+1},\dots,z_{n,t}(h_t)^{p+1})'$, and, for $v\le p$,
\begin{equation}\label{eq:design-moment}
  \boldsymbol{\vartheta}_{t,(\blacktriangle),p}(h_t)
  = \tfrac{1}{n}\,\mathbf{X}_{t,p}'\,\mathbf{A}_{t,(\blacktriangle)}(h_t)\,\mathbf{Z}_{t,p}(h_t),
  \qquad
  B_{t,(\blacktriangle),v,p}(h_t)
  = v!\,h_t^{v}\,e_{v,p}'\,\boldsymbol{\Gamma}_{t,(\blacktriangle),p}(h_t)^{-1}\,
    \boldsymbol{\vartheta}_{t,(\blacktriangle),p}(h_t).
\end{equation}
For brevity, below we abuse notation and write $B_{t,(\blacktriangle),p}$ for $B_{t,(\blacktriangle),0,p}$.
The bias term of the period-$t$ discontinuity estimator and its estimate are
\begin{align*}
  \mathtt{B}_t(h_t)&=\tfrac{h_t^{p+1}}{(p+1)!}\bigl[m^{(p+1)}_{t,(+)}B_{t,(+),p}(h_t)-m^{(p+1)}_{t,(-)}B_{t,(-),p}(h_t)\bigr],\\
  \widehat{\mathtt{B}}_t(h_t,b_t)&=\tfrac{h_t^{p+1}}{(p+1)!}\bigl[\widehat\beta^{(p+1)}_{t,(+),q}(b_t)B_{t,(+),p}(h_t)-\widehat\beta^{(p+1)}_{t,(-),q}(b_t)B_{t,(-),p}(h_t)\bigr],
\end{align*}
where $\widehat{\beta}^{(p+1)}_{t,(\blacktriangle),q}(b_t)$, the $(p+1)$-th derivative from the order-$q$ fit at the bias-correction bandwidth $b_t$, estimates $m^{(p+1)}_{t,(\blacktriangle)}$.

\paragraph{Single-period bias and bias correction.} First we state bias results for a single cross section, building on \citet{calonico2014robust}. A one-sided Taylor expansion gives, for $v\le p$,
\begin{equation}\label{eq:bias-v}
\mathbb{E}[\widehat\beta^{(v)}_{t,(\blacktriangle),p}(h_t)\mid\mathbf{R}_t]-m^{(v)}_{t,(\blacktriangle)}=\tfrac{h_t^{p+1-v}}{(p+1)!}m^{(p+1)}_{t,(\blacktriangle)}B_{t,(\blacktriangle),v,p}(h_t)+o_p(h_t^{p+1-v}).
\end{equation}
For $v=0$, subtracting the below-cutoff expansion from the above-cutoff one gives the first-order bias of the discontinuity estimator,
\begin{equation}\label{eq:bias_single}
\mathbb{E}\bigl[\widehat{D}_{t}(h_{t})\mid\mathbf{R}_{t}\bigr]-D_{t}=\mathtt{B}_{t}(h_{t})+o_{p}(h_{t}^{p+1}). 
\end{equation}
The bias-corrected (BC) intercept subtracts the estimated leading bias,
\[
  \widehat\beta^{(0),\mathrm{BC}}_{t,(\blacktriangle),p,q}(h_t,b_t)
  = \widehat\beta^{(0)}_{t,(\blacktriangle),p}(h_t)
    - \tfrac{h_t^{p+1}}{(p+1)!}\,\widehat{\beta}^{(p+1)}_{t,(\blacktriangle),q}(b_t)\,B_{t,(\blacktriangle),p}(h_t).
\]
As both terms are linear in $\mathbf{Y}_t$, substituting definitions allows us to re-write the BC estimator as 
\begin{equation}\label{eq:bc_Q}
  \widehat\beta^{(0),\mathrm{BC}}_{t,(\blacktriangle),p,q}(h_t,b_t)
  = e_{0,p}'\,\boldsymbol{\Gamma}_{t,(\blacktriangle),p}(h_t)^{-1}\,\tfrac{1}{n}\,\mathbf{Q}_{t,(\blacktriangle),p,q}(h_t,b_t)\,\mathbf{Y}_t,
\end{equation}
where
\[
  \mathbf{Q}_{t,(\blacktriangle),p,q}(h_t,b_t)
  = \mathbf{X}_{t,p}'\mathbf{A}_{t,(\blacktriangle)}(h_t)
    - h_t^{p+1}\,\boldsymbol{\vartheta}_{t,(\blacktriangle),p}(h_t)\,
      e_{p+1,q}'\,\boldsymbol{\Gamma}_{t,(\blacktriangle),q}(b_t)^{-1}\,\mathbf{X}_{t,q}'\mathbf{A}_{t,(\blacktriangle)}(b_t).
\]
That is, the BC intercept is the intercept estimator with the weight matrix $\mathbf{X}_{t,p}'\mathbf{A}_{t,(\blacktriangle)}(h_t)$ replaced by $\mathbf{Q}_{t,(\blacktriangle),p,q}(h_t,b_t)$ \citep{calonico2014robust}. The BC discontinuity estimator is the difference of the two BC intercepts,
\[
  \widehat D^{\mathrm{BC}}_t(h_t,b_t)
  = \widehat\beta^{(0),\mathrm{BC}}_{t,(+),p,q}(h_t,b_t)-\widehat\beta^{(0),\mathrm{BC}}_{t,(-),p,q}(h_t,b_t).
\]

\paragraph{Multiple-period bias correction.}
The aggregate estimator and estimand are linear in the per-period discontinuities, so the bias of the aggregate estimator is the corresponding linear combination of the per-period conditional biases, each conditional on its own period's running variables. Substituting \eqref{eq:bias_single} by period gives
\begin{multline}\label{eq:bias-agg}
  \bigl(\mathbb{E}[\widehat D_{t_{\mathrm{RD}}}(h_{t_{\mathrm{RD}}})\mid\mathbf{R}_{t_{\mathrm{RD}}}]-D_{t_{\mathrm{RD}}}\bigr)
    - \sum_{t_0\in\mathcal{T}_0} \kappa_{t_0}\bigl(\mathbb{E}[\widehat D_{t_0}(h_{t_0})\mid\mathbf{R}_{t_0}]-D_{t_0}\bigr)\\
  = \mathtt{B}_{t_{\mathrm{RD}}}(h_{t_{\mathrm{RD}}})
    - \sum_{t_0\in\mathcal{T}_0} \kappa_{t_0}\,\mathtt{B}_{t_0}(h_{t_0})
    + o_p\bigl(\max_{t\in\mathcal{T}_0\cup\{t_{\mathrm{RD}}\}} h_t^{p+1}\bigr).
\end{multline}
The aggregate BC estimator in \eqref{eq:est_agg_bc} can be re-arranged as subtracting an estimate of the first-order bias term on the right-hand side of \eqref{eq:bias-agg},
\[
\widehat{\att}^{\mathrm{BC}}(t_{\mathrm{RD}};\{h_t\},\{b_t\})=\widehat{\att}(t_{\mathrm{RD}};\{h_t\})-\Bigl[\widehat{\mathtt{B}}_{t_{\mathrm{RD}}}(h_{t_{\mathrm{RD}}},b_{t_{\mathrm{RD}}})-\sum_{t_0\in\mathcal{T}_0} \kappa_{t_0}\,\widehat{\mathtt{B}}_{t_0}(h_{t_0},b_{t_0})\Bigr].
\]

\subsection{Variance}\label{app:est-var}

The variances below are sandwich forms in the error covariances. Let $\boldsymbol\Sigma_t$ be the $n\times n$ diagonal matrix with $[\boldsymbol\Sigma_t]_{ii}=\sigma^2_t(R_{i,t})$ for $i\in\mathcal N_t$ and $0$ otherwise, and, for $t\neq s$, let $\boldsymbol\Sigma_{t,s}$ be the cross-period error covariance matrix, diagonal by Assumption~\ref{ass:reg}(a), with $[\boldsymbol\Sigma_{t,s}]_{ii}=\sigma_{t,s}(R_{i,t},R_{i,s})$ for $i\in\mathcal N_{t,s}$ and $0$ otherwise. Set $\boldsymbol\Sigma_{t,t}=\boldsymbol\Sigma_t$. For sides $\blacktriangle,\blacktriangledown\in\{+,-\}$ let
\begin{align*}
  \boldsymbol\Psi_{t,s,(\blacktriangle,\blacktriangledown),p}(h_t,h_s)&=\tfrac1n\mathbf X_{t,p}'\mathbf A_{t,(\blacktriangle)}(h_t)\,\boldsymbol\Sigma_{t,s}\,\mathbf A_{s,(\blacktriangledown)}(h_s)\mathbf X_{s,p},\\
  \boldsymbol\Psi^{\mathrm{BC}}_{t,s,(\blacktriangle,\blacktriangledown),p,q}(h_t,h_s,b_t,b_s)&=\tfrac1n\mathbf Q_{t,(\blacktriangle),p,q}(h_t,b_t)\,\boldsymbol\Sigma_{t,s}\,\mathbf Q_{s,(\blacktriangledown),p,q}(h_s,b_s)',
\end{align*}
which include the single-period case $s=t$.

\paragraph{Single-period variances.}
Here as well, we begin by stating variance results for a single cross section, building on \citet{calonico2014robust}.
For a single period $t$ and side $\blacktriangle\in\{+,-\}$, the sandwich variance is
\[
  V\!\bigl(\widehat\beta_{t,(\blacktriangle),p}(h_t)\mid\mathbf{R}_t\bigr)
  = \tfrac{1}{n}\,
    \boldsymbol{\Gamma}_{t,(\blacktriangle),p}(h_t)^{-1}\,
    \boldsymbol{\Psi}_{t,t,(\blacktriangle,\blacktriangle),p}(h_t,h_t)\,
    \boldsymbol{\Gamma}_{t,(\blacktriangle),p}(h_t)^{-1}.
\]
The variance of the $v$-th derivative is
\[
V(\widehat\beta^{(v)}_{t,(\blacktriangle),p}(h_t)\mid\mathbf{R}_t)=(v!)^2\,e_{v,p}'\,V(\widehat\beta_{t,(\blacktriangle),p}(h_t)\mid\mathbf{R}_t)\,e_{v,p},
\]
which for $v=0$ is the intercept variance. Let
\begin{equation}\label{eq:dhat-v}
  \widehat D^{(v)}_{t,p}(h_t)=\widehat\beta^{(v)}_{t,(+),p}(h_t)-\widehat\beta^{(v)}_{t,(-),p}(h_t)
\end{equation}
be the period-$t$ discontinuity in the $v$-th derivative, so that $\widehat D^{(0)}_{t,p}(h_t)=\widehat D_t(h_t)$. Within a single period the two sides use disjoint units, so
\[
  V(\widehat D^{(v)}_{t,p}(h_t)\mid\mathbf{R}_t)=V(\widehat\beta^{(v)}_{t,(+),p}(h_t)\mid\mathbf{R}_t)+V(\widehat\beta^{(v)}_{t,(-),p}(h_t)\mid\mathbf{R}_t).
\]

\paragraph{Multiple-period covariances.}
The cross-period covariances decompose as
\begin{align}\label{eq:cross-decomp}
    &\mathrm{Cov}(\widehat D^{(v)}_{t,p}(h_t),\widehat D^{(v)}_{s,p}(h_s)\mid\mathbf{R}_t,\mathbf{R}_s)\nonumber\\
  &\quad= \mathrm{Cov}\!\bigl(\widehat\beta^{(v)}_{t,(+),p}(h_t),\widehat\beta^{(v)}_{s,(+),p}(h_s)\mid\mathbf{R}_t,\mathbf{R}_s\bigr)
  + \mathrm{Cov}\!\bigl(\widehat\beta^{(v)}_{t,(-),p}(h_t),\widehat\beta^{(v)}_{s,(-),p}(h_s)\mid\mathbf{R}_t,\mathbf{R}_s\bigr)\nonumber\\
  &\quad- \mathrm{Cov}\!\bigl(\widehat\beta^{(v)}_{t,(+),p}(h_t),\widehat\beta^{(v)}_{s,(-),p}(h_s)\mid\mathbf{R}_t,\mathbf{R}_s\bigr)
  - \mathrm{Cov}\!\bigl(\widehat\beta^{(v)}_{t,(-),p}(h_t),\widehat\beta^{(v)}_{s,(+),p}(h_s)\mid\mathbf{R}_t,\mathbf{R}_s\bigr).
\end{align}
Each of the four terms in~\eqref{eq:cross-decomp} follows from a cross-period sandwich: for sides $\blacktriangle,\blacktriangledown\in\{+,-\}$,
\[
  \mathrm{Cov}\!\bigl(\widehat\beta_{t,(\blacktriangle),p}(h_t),\widehat\beta_{s,(\blacktriangledown),p}(h_s)\mid\mathbf{R}_t,\mathbf{R}_s\bigr)
  = \tfrac{1}{n}\,
    \boldsymbol{\Gamma}_{t,(\blacktriangle),p}(h_t)^{-1}\,
    \boldsymbol{\Psi}_{t,s,(\blacktriangle,\blacktriangledown),p}(h_t,h_s)\,
    \boldsymbol{\Gamma}_{s,(\blacktriangledown),p}(h_s)^{-1},
\]
and 
\[
\mathrm{Cov}(\widehat\beta^{(v)}_{t,(\blacktriangle),p}(h_t),\widehat\beta^{(v)}_{s,(\blacktriangledown),p}(h_s)\mid\mathbf{R}_t,\mathbf{R}_s)=(v!)^2\,e_{v,p}'\,\mathrm{Cov}(\widehat\beta_{t,(\blacktriangle),p}(h_t),\widehat\beta_{s,(\blacktriangledown),p}(h_s)\mid\mathbf{R}_t,\mathbf{R}_s)\,e_{v,p}.
\]

\paragraph{Sampling scheme.}
Collect the cross-period covariances in \eqref{eq:cross-decomp} for $v=0$ on the same side and on opposite sides of the cutoff,
\begin{align}
\begin{split}
\label{eq:cov-same-opp}
  C^{\mathrm{same}}_{t,s}&=\mathrm{Cov}(\widehat\beta^{(0)}_{t,(+),p},\widehat\beta^{(0)}_{s,(+),p})+\mathrm{Cov}(\widehat\beta^{(0)}_{t,(-),p},\widehat\beta^{(0)}_{s,(-),p}), \\
  C^{\mathrm{opp}}_{t,s}&=\mathrm{Cov}(\widehat\beta^{(0)}_{t,(+),p},\widehat\beta^{(0)}_{s,(-),p})+\mathrm{Cov}(\widehat\beta^{(0)}_{t,(-),p},\widehat\beta^{(0)}_{s,(+),p}),
\end{split}
\end{align}
so that $\mathrm{Cov}(\widehat D_t,\widehat D_s)=C^{\mathrm{same}}_{t,s}-C^{\mathrm{opp}}_{t,s}$. Under CS, $\boldsymbol{\Sigma}_{t,s}=\boldsymbol{0}$ for $t\neq s$ and all four cross-period covariances vanish. Under PC each unit keeps a fixed side, so the opposite-side covariances are zero, $C^{\mathrm{opp}}_{t,s}=0$. Under PV, a unit can switch sides, and all four may be non-zero. Substituting into \eqref{eq:agg-var}, the variance by sampling scheme is
\begin{align}
  V^{\mathrm{CS}}(t_{\mathrm{RD}}) &= V(\widehat D_{t_{\mathrm{RD}}}) + \sum_{t_0\in\mathcal T_0} \kappa_{t_0}^2\, V(\widehat D_{t_0}), \label{eq:var-cs}\\
  V^{\mathrm{PC}}(t_{\mathrm{RD}}) &= V^{\mathrm{CS}}(t_{\mathrm{RD}}) - 2\sum_{t_0\in\mathcal T_0} \kappa_{t_0}\, C^{\mathrm{same}}_{t_{\mathrm{RD}},t_0} + \sum_{t_0\in\mathcal T_0}\sum_{\substack{s_0\in\mathcal T_0\\s_0\ne t_0}} \kappa_{t_0} \kappa_{s_0}\, C^{\mathrm{same}}_{t_0,s_0}, \label{eq:var-pc}\\
  V^{\mathrm{PV}}(t_{\mathrm{RD}}) &= V^{\mathrm{PC}}(t_{\mathrm{RD}}) + 2\sum_{t_0\in\mathcal T_0} \kappa_{t_0}\, C^{\mathrm{opp}}_{t_{\mathrm{RD}},t_0} - \sum_{t_0\in\mathcal T_0}\sum_{\substack{s_0\in\mathcal T_0\\s_0\ne t_0}} \kappa_{t_0} \kappa_{s_0}\, C^{\mathrm{opp}}_{t_0,s_0}. \label{eq:var-pv}
\end{align}

\paragraph{Variance robust to bias correction.} By \eqref{eq:bc_Q}, the BC intercept has the same form as the intercept estimator, with $\mathbf{Q}_{t,(\blacktriangle),p,q}(h_t,b_t)$ in place of $\mathbf{X}_{t,p}'\mathbf{A}_{t,(\blacktriangle)}(h_t)$. The variances and covariances of the BC intercepts are therefore the sandwich forms above with $\boldsymbol{\Psi}$ replaced by $\boldsymbol{\Psi}^{\mathrm{BC}}$, building on \citet{calonico2014robust}. Applying the decomposition~\eqref{eq:cross-decomp} and the three sampling cases to the BC intercepts, with these variances and covariances in place of the conventional ones, delivers $V^{\mathrm{BC,CS}}(t_{\mathrm{RD}})$, $V^{\mathrm{BC,PC}}(t_{\mathrm{RD}})$, and $V^{\mathrm{BC,PV}}(t_{\mathrm{RD}})$.

\paragraph{Estimation.} Let
\[
  \widehat\varepsilon_{i,t}=Y_{i,t}-\mathbf{X}_p(R_{i,t})'\widehat\beta_{t,(+),p}(h_t)\ind\{R_{i,t}\ge c\}-\mathbf{X}_p(R_{i,t})'\widehat\beta_{t,(-),p}(h_t)\ind\{R_{i,t}< c\}
\]
be the period-$t$ residual, using the estimated coefficients from the same side as the unit. The variances and covariances above are estimated by replacing the error variances and covariances with products of residuals,
\[
  [\widehat{\boldsymbol\Sigma}_t]_{ii}=\widehat\varepsilon_{i,t}^{\,2}\ \ (i\in\mathcal N_t),\qquad
  [\widehat{\boldsymbol\Sigma}_{t,s}]_{ii}=\widehat\varepsilon_{i,t}\,\widehat\varepsilon_{i,s}\ \ (i\in\mathcal N_{t,s},\ t\neq s),
\]
and $0$ otherwise, with $\widehat{\boldsymbol\Sigma}_{t,t}=\widehat{\boldsymbol\Sigma}_t$. That is, the estimates are the sandwich forms above with $\boldsymbol\Psi_{t,s,(\blacktriangle,\blacktriangledown),p}(h_t,h_s)$ replaced by
\[
  \widehat{\boldsymbol\Psi}_{t,s,(\blacktriangle,\blacktriangledown),p}(h_t,h_s)=\tfrac1n\mathbf X_{t,p}'\mathbf A_{t,(\blacktriangle)}(h_t)\,\widehat{\boldsymbol\Sigma}_{t,s}\,\mathbf A_{s,(\blacktriangledown)}(h_s)\mathbf X_{s,p},
\]
and we write $\widehat V(\cdot)$ and $\widehat{\mathrm{Cov}}(\cdot,\cdot)$ for them. For the variance robust to bias correction above we follow \citet{calonico2014robust} and its implementation in the \texttt{rdrobust} package: the residuals are those of the order-$q$ fit at $b_t$, $\widehat\varepsilon^{\,\mathrm{BC}}_{i,t}=Y_{i,t}-\mathbf{X}_q(R_{i,t})'\widehat\beta_{t,(\blacktriangle),q}(b_t)$, again from the same side as the unit.

\subsection{Asymptotics}\label{app:est-asy}

\paragraph{Kernel integrals.} Write $\mathbf{r}_p(u)=[1,u,\dots,u^p]'$ and, for each side, the kernel matrices (a tilde marks an integral of the kernel)
\[
\begin{aligned}
  \widetilde{\boldsymbol{\Gamma}}_{(+),p}&=\int_0^\infty K(u)\,\mathbf{r}_p(u)\mathbf{r}_p(u)'\,du, \qquad 
  \;\;\widetilde{\boldsymbol{\Gamma}}_{(-),p}=\int_{-\infty}^0 K(u)\,\mathbf{r}_p(u)\mathbf{r}_p(u)'\,du
  \\  
  \widetilde{\boldsymbol{\Psi}}_{(+),p}&=\int_0^\infty K(u)^2\,\mathbf{r}_p(u)\mathbf{r}_p(u)'\,du, \qquad
  \widetilde{\boldsymbol{\Psi}}_{(-),p}=\int_{-\infty}^0 K(u)^2\,\mathbf{r}_p(u)\mathbf{r}_p(u)'\,du,
  \\
  \widetilde{\boldsymbol{\vartheta}}_{(+),p}&=\int_0^\infty K(u)\,u^{p+1}\,\mathbf{r}_p(u)\,du, \qquad\;\;\;\;
  \widetilde{\boldsymbol{\vartheta}}_{(-),p}=\int_{-\infty}^0 K(u)\,u^{p+1}\,\mathbf{r}_p(u)\,du,
\end{aligned}
\]
and, for $\rho>0$,
\[
\begin{aligned}
  \widetilde{\boldsymbol{\Omega}}_{(+),p}(\rho)&=\int_0^\infty K(v)\,K(\rho v)\,\mathbf{r}_p(v)\mathbf{r}_p(\rho v)'\,dv, \qquad
  \widetilde{\boldsymbol{\Omega}}_{(-),p}(\rho)=\int_{-\infty}^0 K(v)\,K(\rho v)\,\mathbf{r}_p(v)\mathbf{r}_p(\rho v)'\,dv.
\end{aligned}
\]
For $\blacktriangle\in\{+,-\}$ define the kernel constants
\[
  \tilde\nu_{(\blacktriangle),p}=e_{0,p}'\,\widetilde{\boldsymbol{\Gamma}}_{(\blacktriangle),p}^{-1}\widetilde{\boldsymbol{\vartheta}}_{(\blacktriangle),p},
  \qquad
  \tilde\omega_{(\blacktriangle),p}=e_{0,p}'\,\widetilde{\boldsymbol{\Gamma}}_{(\blacktriangle),p}^{-1}\widetilde{\boldsymbol{\Psi}}_{(\blacktriangle),p}\widetilde{\boldsymbol{\Gamma}}_{(\blacktriangle),p}^{-1}e_{0,p},
\]
and
\[
  \tilde\omega_{(\blacktriangle),p}(\rho)=e_{0,p}'\,\widetilde{\boldsymbol{\Gamma}}_{(\blacktriangle),p}^{-1}\widetilde{\boldsymbol{\Omega}}_{(\blacktriangle),p}(\rho)\widetilde{\boldsymbol{\Gamma}}_{(\blacktriangle),p}^{-1}e_{0,p},
\]
which depend only on $K$ and $p$, and on $\rho$ for the last. As $K$ is positive on an interval and a nonzero polynomial of degree $p$ has at most $p$ roots, $\widetilde{\boldsymbol{\Gamma}}_{(\blacktriangle),p}$ is positive definite under Assumption~\ref{ass:reg}(b), so its inverse exists. Since $\widetilde{\boldsymbol{\Omega}}_{(\blacktriangle),p}(1)=\widetilde{\boldsymbol{\Psi}}_{(\blacktriangle),p}$, $\tilde\omega_{(\blacktriangle),p}(1)$ is the variance constant $\tilde\omega_{(\blacktriangle),p}$ above. Substituting $v=\rho u$ gives $\widetilde{\boldsymbol{\Omega}}_{(\blacktriangle),p}(1/\rho)=\rho\,\widetilde{\boldsymbol{\Omega}}_{(\blacktriangle),p}(\rho)'$ and hence $\tilde\omega_{(\blacktriangle),p}(1/\rho)=\rho\,\tilde\omega_{(\blacktriangle),p}(\rho)$.

\paragraph{Rescaled design matrices.} Let $\mathbf{H}_p(h)=\mathrm{diag}(1,h,\dots,h^p)$, so that $\mathbf{X}_p(r)=\mathbf{H}_p(h)\,\mathbf{r}_p((r-c)/h)$. Let $a_{i,t,(\blacktriangle)}(h_t)$ be the $i$-th diagonal entry of $\mathbf{A}_{t,(\blacktriangle)}(h_t)$, and recall that $z_{i,t}(h_t)=(R_{i,t}-c)/h_t$ denotes the scaled running variable. Then for $\blacktriangle,\blacktriangledown\in\{+,-\}$, writing the design matrices rescaled by $\mathbf H_p$ with a check,
\[
\begin{aligned}
\boldsymbol{\Gamma}_{t,(\blacktriangle),p}(h_t)&=\mathbf{H}_p(h_t)\check{\boldsymbol{\Gamma}}_{t,(\blacktriangle),p}(h_t)\mathbf{H}_p(h_t)\\
\boldsymbol{\Psi}_{t,s,(\blacktriangle,\blacktriangledown),p}(h_t,h_s)&=\mathbf{H}_p(h_t)\check{\boldsymbol{\Psi}}_{t,s,(\blacktriangle,\blacktriangledown),p}(h_t,h_s)\mathbf{H}_p(h_s)
\end{aligned}
\]
where
\[
\begin{aligned}
\check{\boldsymbol{\Gamma}}_{t,(\blacktriangle),p}(h_t)&=\frac{1}{n}\sum_{i=1}^{n} a_{i,t,(\blacktriangle)}(h_t)\,\mathbf{r}_p(z_{i,t}(h_t))\mathbf{r}_p(z_{i,t}(h_t))', \\
\check{\boldsymbol{\Psi}}_{t,s,(\blacktriangle,\blacktriangledown),p}(h_t,h_s)&=\frac{1}{n}\sum_{i=1}^{n} a_{i,t,(\blacktriangle)}(h_t)\,a_{i,s,(\blacktriangledown)}(h_s)\,[\boldsymbol{\Sigma}_{t,s}]_{ii}\,\mathbf{r}_p(z_{i,t}(h_t))\mathbf{r}_p(z_{i,s}(h_s))'.
\end{aligned}
\]

\paragraph{Limits.} By \citet{calonico2014robust},
\[
\begin{aligned}
    (p+1)!\,\mathtt B_t(h_t)/h_t^{p+1} &\to_p m^{(p+1)}_{t,(+)}\tilde\nu_{(+),p}-m^{(p+1)}_{t,(-)}\tilde\nu_{(-),p}\equiv \bar{\mathtt B}_t,\\
    n_th_t\,V(\widehat D_t\mid\mathbf R_t)&\to_p\frac{\sigma^2_{t,(+)}\tilde\omega_{(+),p}+\sigma^2_{t,(-)}\tilde\omega_{(-),p}}{f_t(c)}=\bar V_t,
\end{aligned}
\]
with $\bar V_t$ as defined in Section~\ref{sec:est-inf}. Under PC and for $t\neq s$, recall from Section~\ref{sec:est-inf} that
\[
  \bar C^{\mathrm{PC}}_{t,s}=\frac{\sigma_{t,s,(+,+)}\,\tilde\omega_{(+),p}(\bar\rho_{t,s})+\sigma_{t,s,(-,-)}\,\tilde\omega_{(-),p}(\bar\rho_{t,s})}{f(c)}.
\]
Lemma~\ref{lem:cov-pc} below shows that $nh_s\,\mathrm{Cov}(\widehat D_t,\widehat D_s\mid\mathbf R_t,\mathbf R_s)$ converges in probability to $\bar C^{\mathrm{PC}}_{t,s}$.

\paragraph{Single-period asymptotics.}
For completeness we restate here several asymptotic results for a single cross section from \citet{calonico2014robust}. The leading bias and variance of $\widehat D_t(h_t)$ are governed by one constant each. Under Assumption~\ref{ass:reg}, for $\blacktriangle\in\{+,-\}$,
\[
B_{t,(\blacktriangle),p}(h_t)\to_p\tilde\nu_{(\blacktriangle),p}, \qquad 
n_th_t\,V(\widehat\beta^{(0)}_{t,(\blacktriangle),p}(h_t)\mid\mathbf{R}_t)\to_p\sigma^2_{t,(\blacktriangle)}\tilde\omega_{(\blacktriangle),p}/f_t(c).
\] 
The bias term and variance satisfy
\begin{equation}\label{eq:per-period-orders}
  \mathtt{B}_t(h_t)=\frac{h_t^{p+1}}{(p+1)!}\,\bar{\mathtt B}_t+o_p\bigl(h_t^{p+1}\bigr),
  \qquad
  V(\widehat D_t(h_t)\mid\mathbf{R}_t)=\frac{\bar V_t}{n_th_t}+o_p\bigl((n_th_t)^{-1}\bigr).
\end{equation}

Together, the two parts of \eqref{eq:per-period-orders} bound the estimation error of each per-period discontinuity estimator.

\begin{lemma}\label{lem:dhat-order}
Under Assumption~\ref{ass:reg}, for each period $t$, $\widehat D_t(h_t)-D_t=O_p\bigl(h_t^{p+1}+(nh_t)^{-1/2}\bigr)$.
\end{lemma}
\begin{proof}
Adding and subtracting the conditional mean,
\[
  \widehat D_t(h_t)-D_t
  =\bigl(\mathbb{E}[\widehat D_t(h_t)\mid\mathbf R_t]-D_t\bigr)
  +U_t,
  \qquad
  U_t=\widehat D_t(h_t)-\mathbb{E}[\widehat D_t(h_t)\mid\mathbf R_t].
\]
The first term is the conditional bias. By \eqref{eq:bias_single} and \eqref{eq:per-period-orders} it equals $\tfrac{h_t^{p+1}}{(p+1)!}\bar{\mathtt B}_t+o_p(h_t^{p+1})$, which is $O_p(h_t^{p+1})$ as $\bar{\mathtt B}_t$ is finite.

Given $\mathbf R_t$, the second term $U_t$ has conditional expectation zero and conditional variance equal to $V(\widehat D_t(h_t)\mid\mathbf R_t)$. By \eqref{eq:per-period-orders}, $n_th_t\,V(\widehat D_t(h_t)\mid\mathbf R_t)\to_p\bar V_t$. As $n_th_t=(n_t/n)\,nh_t$ and $n_t/n\to_p\eta_t>0$ by the law of large numbers and Assumption~\ref{ass:reg}(a) and (g), $nh_t\,V(\widehat D_t(h_t)\mid\mathbf R_t)\to_p\bar V_t/\eta_t$, so
\[
  V(\widehat D_t(h_t)\mid\mathbf R_t)=O_p\bigl((nh_t)^{-1}\bigr).
\]
Next, by the conditional Chebyshev inequality, for every $x>0$,
\[
  P\bigl(|U_t|>x\mid\mathbf R_t\bigr)\le\frac{V(\widehat D_t(h_t)\mid\mathbf R_t)}{x^2}.
\]
Setting $x=M\,V(\widehat D_t(h_t)\mid\mathbf R_t)^{1/2}$ with $M>0$ and substituting into the above gives
\[
  P\Bigl(|U_t|>M\,V(\widehat D_t(h_t)\mid\mathbf R_t)^{1/2}\ \Big|\ \mathbf R_t\Bigr)\le\frac{1}{M^2}.
\]
Since the bound $M^{-2}$ does not depend on $\mathbf R_t$, by the law of iterated expectations the unconditional probability has the same bound. Finally, for any $\epsilon>0$, setting $M=\epsilon^{-1/2}$ makes the probability at most $\epsilon$ for every $n$, which is the definition of
\[
  U_t=O_p\Bigl(V(\widehat D_t(h_t)\mid\mathbf R_t)^{1/2}\Bigr).
\]
With the order of the variance found above, $V(\widehat D_t(h_t)\mid\mathbf R_t)=O_p((nh_t)^{-1})$, this implies $U_t=O_p((nh_t)^{-1/2})$. Adding the two terms gives the result.
\end{proof}

\paragraph{Rescaled covariances.} Since $e_{0,p}'\mathbf{H}_p(h)^{-1}=e_{0,p}'$,
\begin{equation}\label{eq:cov-rescaled}
  \mathrm{Cov}\bigl(\widehat\beta^{(0)}_{t,(\blacktriangle),p}(h_t),\widehat\beta^{(0)}_{s,(\blacktriangledown),p}(h_s)\mid\mathbf{R}_t,\mathbf{R}_s\bigr)
  =\frac{1}{n}\,e_{0,p}'\,\check{\boldsymbol{\Gamma}}_{t,(\blacktriangle),p}(h_t)^{-1}\,\check{\boldsymbol{\Psi}}_{t,s,(\blacktriangle,\blacktriangledown),p}(h_t,h_s)\,\check{\boldsymbol{\Gamma}}_{s,(\blacktriangledown),p}(h_s)^{-1}\,e_{0,p}.
\end{equation}
The order of \eqref{eq:cov-rescaled} is determined by the limits of the rescaled matrices $\check{\boldsymbol{\Gamma}}_{t,(\blacktriangle),p}(h_t)$ and $\check{\boldsymbol{\Psi}}_{t,s,(\blacktriangle,\blacktriangledown),p}(h_t,h_s)$. We first establish these limits in auxiliary lemmas below, and then derive the asymptotics of the covariances by sampling scheme. Two periods may use different bandwidths, which enters the limits through the ratio $\bar\rho_{t,s}$. By the symmetry of $\tilde\omega_{(\blacktriangle),p}(\rho)$ above, the normalization by $h_s$ in the lemmas below is symmetric in $t$ and $s$.

\paragraph{Auxiliary lemmas.}
Each proof focuses on side $+$, as the side $-$ is analogous.

\begin{lemma}\label{lem:aux-gamma}
For $\blacktriangle\in\{+,-\}$, $\check{\boldsymbol{\Gamma}}_{t,(\blacktriangle),p}(h_t)\to_p \eta_t f_t(c)\,\widetilde{\boldsymbol{\Gamma}}_{(\blacktriangle),p}$.
\end{lemma}
\begin{proof}
By definition, $a_{i,t,(+)}(h_t)=\ind\{i\in\mathcal{N}_t\}\ind\{R_{i,t}\ge c\}K(z_{i,t}(h_t))/h_t$. Then
\begin{align}\label{eq:gamma_exp_lim}
\mathbb{E}\bigl[\check{\boldsymbol{\Gamma}}_{t,(+),p}(h_t)\bigr]
&=\frac{\eta_t}{h_t}\int_c^\infty K\!\Bigl(\tfrac{r-c}{h_t}\Bigr)\,\mathbf{r}_p\!\Bigl(\tfrac{r-c}{h_t}\Bigr)\mathbf{r}_p\!\Bigl(\tfrac{r-c}{h_t}\Bigr)'\,f_t(r)\,dr \nonumber\\
&=\eta_t\int_0^\infty K(u)\,\mathbf{r}_p(u)\mathbf{r}_p(u)'\,f_t(c+h_tu)\,du\nonumber\\
&\to \eta_t\int_0^\infty K(u)\mathbf{r}_p(u)\mathbf{r}_p(u)'f_t(c)du\nonumber\\
&=\eta_tf_t(c)\widetilde{\boldsymbol{\Gamma}}_{(+),p},  
\end{align}
where the first equality follows from definitions $P(i\in\mathcal{N}_t)=\eta_t$ and $f_t$ denoting the density of $R_{i,t}$ given $i\in\mathcal{N}_t$ and assuming units are identically distributed. The second equality follows from substituting $u=(r-c)/h_t$. The limit assumes $h_t\to0$ and continuity of $f_t$ at $c$, and the last equality follows by definition of $\widetilde{\boldsymbol{\Gamma}}_{(+),p}$. 

Next, let $\boldsymbol{\gamma}_i=a_{i,t,(+)}(h_t)\,\mathbf{r}_p(z_{i,t}(h_t))\mathbf{r}_p(z_{i,t}(h_t))'$, so that $\check{\boldsymbol{\Gamma}}_{t,(+),p}(h_t)=\tfrac{1}{n}\sum_{i=1}^n\boldsymbol{\gamma}_i$. By Assumption~\ref{ass:reg}(a) the $\boldsymbol{\gamma}_i$ are independent and identically distributed, so each entry of $\check{\boldsymbol{\Gamma}}_{t,(+),p}(h_t)-\mathbb{E}[\check{\boldsymbol{\Gamma}}_{t,(+),p}(h_t)]$ is an average of $n$ independent, identically distributed, mean-zero terms, whose variance is $1/n$ times the variance of one term. Summing the entrywise variances, for any single $i$,
\[
\begin{aligned}
\mathbb{E}\Bigl\|\check{\boldsymbol{\Gamma}}_{t,(+),p}(h_t)-\mathbb{E}\bigl[\check{\boldsymbol{\Gamma}}_{t,(+),p}(h_t)\bigr]\Bigr\|^2
&=\frac{1}{n}\,\mathbb{E}\bigl\|\boldsymbol{\gamma}_i-\mathbb{E}\boldsymbol{\gamma}_i\bigr\|^2\\
&\le\frac{1}{n}\,\mathbb{E}\bigl\|\boldsymbol{\gamma}_i\bigr\|^2
=\frac{1}{n}\,\mathbb{E}\Bigl[a_{i,t,(+)}(h_t)^2\,\bigl\|\mathbf{r}_p(z_{i,t}(h_t))\bigr\|^4\Bigr]\\
&=\frac{\eta_t}{nh_t}\int_0^\infty K(u)^2\,\bigl\|\mathbf{r}_p(u)\bigr\|^4\,f_t(c+h_tu)\,du\\
&\to 0,
\end{aligned}
\]
where the inequality is $\mathbb{E}\|\mathbf{M}-\mathbb{E}\mathbf{M}\|^2\le\mathbb{E}\|\mathbf{M}\|^2$ for a random matrix $\mathbf M$, the equality next to it uses $\|\mathbf{r}\mathbf{r}'\|=\|\mathbf{r}\|^2$, the following equality substitutes $u=(r-c)/h_t$ as above, and the limit follows from $nh_t\to\infty$ and continuity of $f_t$ at $c$. By Markov's inequality, $\check{\boldsymbol{\Gamma}}_{t,(+),p}(h_t)-\mathbb{E}[\check{\boldsymbol{\Gamma}}_{t,(+),p}(h_t)]\to_p\boldsymbol{0}$, and by substituting \eqref{eq:gamma_exp_lim} the result follows.
\end{proof}

\begin{lemma}\label{lem:aux-psi-pc}
Under PC, for $t\neq s$ and $\blacktriangle\in\{+,-\}$,
\[
  h_s\,\check{\boldsymbol{\Psi}}_{t,s,(\blacktriangle,\blacktriangle),p}(h_t,h_s)\to_p \sigma_{t,s,(\blacktriangle,\blacktriangle)}\,f(c)\,\widetilde{\boldsymbol{\Omega}}_{(\blacktriangle),p}(\bar\rho_{t,s}).
\]
\end{lemma}
\begin{proof}
Since $\ind\{i\in\mathcal N_t\}=\ind\{i\in\mathcal N_s\}=1$ for all $i$ and $R_{i,t}=R_{i,s}=R_i$,
\begin{equation}\label{eq:pc-weights}
a_{i,t,(+)}(h_t)a_{i,s,(+)}(h_s)=\ind\{R_i\geq c\}K(z_{i,t}(h_t))K(z_{i,s}(h_s))/(h_th_s),
\end{equation}
and $[\boldsymbol{\Sigma}_{t,s}]_{ii}=\sigma_{t,s}(R_i,R_i)$. Then
\begin{align}\label{eq:psi_pc_exp_lim}
\mathbb{E}\bigl[h_s\,\check{\boldsymbol{\Psi}}_{t,s,(+,+),p}(h_t,h_s)\bigr]
&=\frac{1}{h_t}\int_c^\infty K\!\Bigl(\tfrac{r-c}{h_t}\Bigr)K\!\Bigl(\tfrac{r-c}{h_s}\Bigr)\,\sigma_{t,s}(r,r)\,\mathbf{r}_p\!\Bigl(\tfrac{r-c}{h_t}\Bigr)\mathbf{r}_p\!\Bigl(\tfrac{r-c}{h_s}\Bigr)'\,f(r)\,dr\nonumber\\
&=\int_0^\infty K(u)\,K\!\Bigl(\tfrac{h_t}{h_s}u\Bigr)\,\sigma_{t,s}(c+h_tu,c+h_tu)\,\mathbf{r}_p(u)\,\mathbf{r}_p\!\Bigl(\tfrac{h_t}{h_s}u\Bigr)'\,f(c+h_tu)\,du\nonumber\\
&\to \int_0^\infty K(u)K(\bar\rho_{t,s}u)\,\sigma_{t,s,(+,+)}\,\mathbf{r}_p(u)\mathbf{r}_p(\bar\rho_{t,s}u)'\,f(c)\,du\nonumber\\
&=\sigma_{t,s,(+,+)}\,f(c)\,\widetilde{\boldsymbol{\Omega}}_{(+),p}(\bar\rho_{t,s}),
\end{align}
where the first equality follows as in the proof of Lemma~\ref{lem:aux-gamma}, with $\eta_t=1$ and $f$ the density of $R_i$, after substituting \eqref{eq:pc-weights}. The second equality follows from substituting $u=(r-c)/h_t$. The limit assumes $h_t\to0$ and $h_t/h_s\to\bar\rho_{t,s}$, continuity of $K$, continuity of $f$ at $c$, and existence of the limit $\sigma_{t,s,(+,+)}$. The last equality follows by definition of $\widetilde{\boldsymbol{\Omega}}_{(+),p}$.

Next, let 
\[
\boldsymbol{\psi}^{\mathrm{PC}}_i=h_s\,a_{i,t,(+)}(h_t)a_{i,s,(+)}(h_s)\,\sigma_{t,s}(R_i,R_i)\,\mathbf{r}_p(z_{i,t}(h_t))\mathbf{r}_p(z_{i,s}(h_s))',
\]
so that $h_s\check{\boldsymbol{\Psi}}_{t,s,(+,+),p}(h_t,h_s)=\tfrac{1}{n}\sum_{i=1}^n\boldsymbol{\psi}^{\mathrm{PC}}_i$. As in the proof of Lemma~\ref{lem:aux-gamma}, for any single $i$,
\[
\begin{aligned}
&\mathbb{E}\Bigl\|h_s\check{\boldsymbol{\Psi}}_{t,s,(+,+),p}(h_t,h_s)-\mathbb{E}\bigl[h_s\check{\boldsymbol{\Psi}}_{t,s,(+,+),p}(h_t,h_s)\bigr]\Bigr\|^2\\
&\le\frac{1}{n}\,\mathbb{E}\bigl\|\boldsymbol{\psi}^{\mathrm{PC}}_i\bigr\|^2\\
&=\frac{1}{nh_t}\int_0^\infty K(u)^2K\!\Bigl(\tfrac{h_t}{h_s}u\Bigr)^2\sigma_{t,s}(c+h_tu,c+h_tu)^2\,\bigl\|\mathbf{r}_p(u)\bigr\|^2\Bigl\|\mathbf{r}_p\!\Bigl(\tfrac{h_t}{h_s}u\Bigr)\Bigr\|^2 f(c+h_tu)\,du\\
&\to 0,
\end{aligned}
\]
where the equality substitutes \eqref{eq:pc-weights} and $u=(r-c)/h_t$, using $\|\mathbf{a}\mathbf{b}'\|=\|\mathbf{a}\|\,\|\mathbf{b}\|$, and the limit follows from $nh_t\to\infty$, continuity of $f$ at $c$, and existence of the limit $\sigma_{t,s,(+,+)}$. By Markov's inequality, $h_s\check{\boldsymbol{\Psi}}_{t,s,(+,+),p}(h_t,h_s)-\mathbb{E}[h_s\check{\boldsymbol{\Psi}}_{t,s,(+,+),p}(h_t,h_s)]\to_p\boldsymbol{0}$, and by substituting \eqref{eq:psi_pc_exp_lim} the result follows.
\end{proof}

\begin{lemma}\label{lem:aux-psi-pv}
Under PV, for $t\neq s$ and $\blacktriangle,\blacktriangledown\in\{+,-\}$,
\[
  \check{\boldsymbol{\Psi}}_{t,s,(\blacktriangle,\blacktriangledown),p}(h_t,h_s)\to_p \sigma_{t,s,(\blacktriangle,\blacktriangledown)}\,f_{t,s}(c,c)\,\widetilde{\boldsymbol{\Gamma}}_{(\blacktriangle),p}e_{0,p}\,e_{0,p}'\widetilde{\boldsymbol{\Gamma}}_{(\blacktriangledown),p}.
\]
\end{lemma}
\begin{proof}
Take $(\blacktriangle,\blacktriangledown)=(+,+)$; the other three pairs are analogous with the corresponding integration limits. Since $\ind\{i\in\mathcal N_t\}=\ind\{i\in\mathcal N_s\}=1$ for all $i$,
\begin{equation}\label{eq:pv-weights}
  a_{i,t,(+)}(h_t)a_{i,s,(+)}(h_s)=\ind\{R_{i,t}\ge c\}\ind\{R_{i,s}\ge c\}K(z_{i,t}(h_t))K(z_{i,s}(h_s))/(h_th_s),
\end{equation}
and $[\boldsymbol{\Sigma}_{t,s}]_{ii}=\sigma_{t,s}(R_{i,t},R_{i,s})$. Then
\begin{align}\label{eq:psi_pv_exp_lim}
&\mathbb{E}\bigl[\check{\boldsymbol{\Psi}}_{t,s,(+,+),p}(h_t,h_s)\bigr] \nonumber\\
&=\frac{1}{h_th_s}\int_c^\infty\!\!\int_c^\infty K\!\Bigl(\tfrac{r-c}{h_t}\Bigr)K\!\Bigl(\tfrac{r'-c}{h_s}\Bigr)\,\sigma_{t,s}(r,r')\,\mathbf{r}_p\!\Bigl(\tfrac{r-c}{h_t}\Bigr)\mathbf{r}_p\!\Bigl(\tfrac{r'-c}{h_s}\Bigr)'\,f_{t,s}(r,r')\,dr\,dr'\nonumber\\
&=\int_0^\infty\!\!\int_0^\infty K(u)K(u')\,\sigma_{t,s}(c+h_tu,c+h_su')\,\mathbf{r}_p(u)\mathbf{r}_p(u')'\,f_{t,s}(c+h_tu,c+h_su')\,du\,du'\nonumber\\
&\to \int_0^\infty\!\!\int_0^\infty K(u)K(u')\,\sigma_{t,s,(+,+)}\,\mathbf{r}_p(u)\mathbf{r}_p(u')'\,f_{t,s}(c,c)\,du\,du'\nonumber\\
&=\sigma_{t,s,(+,+)}\,f_{t,s}(c,c)\,\widetilde{\boldsymbol{\Gamma}}_{(+),p}e_{0,p}\,e_{0,p}'\widetilde{\boldsymbol{\Gamma}}_{(+),p},
\end{align}
where the first equality follows as in the proof of Lemma~\ref{lem:aux-gamma}, with $\eta_t=1$ and $f_{t,s}$ the joint density of $(R_{i,t},R_{i,s})$, after substituting \eqref{eq:pv-weights}. The second equality follows from substituting $u=(r-c)/h_t$ and $u'=(r'-c)/h_s$. The limit assumes $h_t,h_s\to0$, continuity of $f_{t,s}$ at $(c,c)$, and existence of the limit $\sigma_{t,s,(+,+)}$. The last equality follows because the double integral factors into $\bigl[\int_0^\infty K(u)\mathbf{r}_p(u)\,du\bigr]\bigl[\int_0^\infty K(u')\mathbf{r}_p(u')\,du'\bigr]'$. The first factor equals $\widetilde{\boldsymbol{\Gamma}}_{(+),p}e_{0,p}$, the first column of $\widetilde{\boldsymbol{\Gamma}}_{(+),p}$, since $\mathbf{r}_p(u)'e_{0,p}=1$, and the second factor is its transpose, $e_{0,p}'\widetilde{\boldsymbol{\Gamma}}_{(+),p}$; for mixed sides the second factor is $e_{0,p}'\widetilde{\boldsymbol{\Gamma}}_{(\blacktriangledown),p}$.

Next, let
\[
\boldsymbol{\psi}^{\mathrm{PV}}_i=a_{i,t,(+)}(h_t)a_{i,s,(+)}(h_s)\,\sigma_{t,s}(R_{i,t},R_{i,s})\,\mathbf{r}_p(z_{i,t}(h_t))\mathbf{r}_p(z_{i,s}(h_s))',
\]
so that $\check{\boldsymbol{\Psi}}_{t,s,(+,+),p}(h_t,h_s)=\tfrac{1}{n}\sum_{i=1}^n\boldsymbol{\psi}^{\mathrm{PV}}_i$. As in the proofs of Lemmas~\ref{lem:aux-gamma} and~\ref{lem:aux-psi-pc}, for any single $i$,
\[
\begin{aligned}
&\mathbb{E}\Bigl\|\check{\boldsymbol{\Psi}}_{t,s,(+,+),p}(h_t,h_s)-\mathbb{E}\bigl[\check{\boldsymbol{\Psi}}_{t,s,(+,+),p}(h_t,h_s)\bigr]\Bigr\|^2\\
&\le\frac{1}{n}\,\mathbb{E}\bigl\|\boldsymbol{\psi}^{\mathrm{PV}}_i\bigr\|^2\\
&=\frac{1}{nh_th_s}\int_0^\infty\!\!\int_0^\infty K(u)^2K(u')^2\,\sigma_{t,s}(c+h_tu,c+h_su')^2\,\bigl\|\mathbf{r}_p(u)\bigr\|^2\bigl\|\mathbf{r}_p(u')\bigr\|^2\,f_{t,s}(c+h_tu,c+h_su')\,du\,du'\\
&\to 0,
\end{aligned}
\]
where the equality substitutes \eqref{eq:pv-weights}, $u=(r-c)/h_t$, and $u'=(r'-c)/h_s$, and the limit follows from $nh_th_s\to\infty$, continuity of $f_{t,s}$ at $(c,c)$, and existence of the limit $\sigma_{t,s,(+,+)}$. By Markov's inequality, $\check{\boldsymbol{\Psi}}_{t,s,(+,+),p}(h_t,h_s)-\mathbb{E}[\check{\boldsymbol{\Psi}}_{t,s,(+,+),p}(h_t,h_s)]\to_p\boldsymbol{0}$, and by substituting \eqref{eq:psi_pv_exp_lim} the result follows.
\end{proof}

\paragraph{Covariance asymptotics by sampling scheme.}

\begin{lemma}\label{lem:cov-cs} Under CS, for $t\neq s$ the covariance in \eqref{eq:cov-rescaled} is zero.
\end{lemma}
\begin{proof}
By definition $\mathcal{N}_{t,s}=\varnothing$ implies $\boldsymbol{\Sigma}_{t,s}=\boldsymbol{0}$. Hence for any $\blacktriangle,\blacktriangledown\in\{+,-\}$, $\check{\boldsymbol{\Psi}}_{t,s,(\blacktriangle,\blacktriangledown),p}(h_t,h_s)=\boldsymbol{0}$, and in turn \eqref{eq:cov-rescaled} is zero.
\end{proof}

\begin{lemma}\label{lem:cov-pc}
Under PC, for $t\neq s$, 
\[
\mathrm{Cov}\bigl(\widehat\beta^{(0)}_{t,(+),p}(h_t),\widehat\beta^{(0)}_{s,(-),p}(h_s)\mid\mathbf{R}_t,\mathbf{R}_s\bigr)=\mathrm{Cov}\bigl(\widehat\beta^{(0)}_{t,(-),p}(h_t),\widehat\beta^{(0)}_{s,(+),p}(h_s)\mid\mathbf{R}_t,\mathbf{R}_s\bigr)=0
\]
and for $\blacktriangle\in\{+,-\}$
\[
  nh_s\,\mathrm{Cov}\bigl(\widehat\beta^{(0)}_{t,(\blacktriangle),p}(h_t),\widehat\beta^{(0)}_{s,(\blacktriangle),p}(h_s)\mid\mathbf{R}_t,\mathbf{R}_s\bigr)\to_p \frac{\sigma_{t,s,(\blacktriangle,\blacktriangle)}\,\tilde\omega_{(\blacktriangle),p}(\bar\rho_{t,s})}{f(c)}.
\]
Hence 
\[
nh_s\,\mathrm{Cov}(\widehat D_t(h_t),\widehat D_s(h_s)\mid\mathbf{R}_t,\mathbf{R}_s)\to_p\bar C^{\mathrm{PC}}_{t,s}.
\]
\end{lemma}
\begin{proof}
Each unit keeps a fixed side, so $a_{i,t,(+)}(h_t)a_{i,s,(-)}(h_s)=0$ for every $i$, hence $\check{\boldsymbol{\Psi}}_{t,s,(+,-),p}(h_t,h_s)=\boldsymbol{0}$, and symmetrically $\check{\boldsymbol{\Psi}}_{t,s,(-,+),p}(h_t,h_s)=\boldsymbol{0}$. Combined with \eqref{eq:cov-rescaled} this implies the opposite-side covariances are zero. 

For the same side, for $\blacktriangle\in\{+,-\}$,
\[
\begin{aligned}
nh_s\,\mathrm{Cov}\bigl(\widehat\beta^{(0)}_{t,(\blacktriangle),p}(h_t),\widehat\beta^{(0)}_{s,(\blacktriangle),p}(h_s)\mid\mathbf{R}_t,\mathbf{R}_s\bigr)
&=e_{0,p}'\,\check{\boldsymbol{\Gamma}}_{t,(\blacktriangle),p}(h_t)^{-1}\,\bigl[h_s\check{\boldsymbol{\Psi}}_{t,s,(\blacktriangle,\blacktriangle),p}(h_t,h_s)\bigr]\,\check{\boldsymbol{\Gamma}}_{s,(\blacktriangle),p}(h_s)^{-1}\,e_{0,p}\\
&\to_p e_{0,p}'\,\bigl[f(c)\widetilde{\boldsymbol{\Gamma}}_{(\blacktriangle),p}\bigr]^{-1}\,\bigl[\sigma_{t,s,(\blacktriangle,\blacktriangle)}\,f(c)\,\widetilde{\boldsymbol{\Omega}}_{(\blacktriangle),p}(\bar\rho_{t,s})\bigr]\,\bigl[f(c)\widetilde{\boldsymbol{\Gamma}}_{(\blacktriangle),p}\bigr]^{-1}\,e_{0,p}\\
&=\frac{\sigma_{t,s,(\blacktriangle,\blacktriangle)}\,\tilde\omega_{(\blacktriangle),p}(\bar\rho_{t,s})}{f(c)},
\end{aligned}
\]
where the first equality follows from \eqref{eq:cov-rescaled}. The limit substitutes Lemma~\ref{lem:aux-gamma} with $\eta_t=\eta_s=1$ and $f_t=f_s=f$ and Lemma~\ref{lem:aux-psi-pc}. The last equality collects the $f(c)$ factors, and substitutes the definition of $\tilde\omega_{(\blacktriangle),p}(\rho)$. Summing the two same-side terms in \eqref{eq:cross-decomp} gives $\bar C^{\mathrm{PC}}_{t,s}$.
\end{proof}

\begin{lemma}\label{lem:cov-pv}
Under PV, for $t\neq s$ and $\blacktriangle,\blacktriangledown\in\{+,-\}$,
\[
  n\,\mathrm{Cov}\bigl(\widehat\beta^{(0)}_{t,(\blacktriangle),p}(h_t),\widehat\beta^{(0)}_{s,(\blacktriangledown),p}(h_s)\mid\mathbf{R}_t,\mathbf{R}_s\bigr)\to_p \frac{\sigma_{t,s,(\blacktriangle,\blacktriangledown)}\,f_{t,s}(c,c)}{f_t(c)\,f_s(c)}.
\]
Hence
\[
  n\,\mathrm{Cov}(\widehat D_t(h_t),\widehat D_s(h_s)\mid\mathbf{R}_t,\mathbf{R}_s)\to_p \frac{f_{t,s}(c,c)}{f_t(c)\,f_s(c)}\bigl[\sigma_{t,s,(+,+)}+\sigma_{t,s,(-,-)}-\sigma_{t,s,(+,-)}-\sigma_{t,s,(-,+)}\bigr]=\bar C^{\mathrm{PV}}_{t,s},
\]
with $\bar C^{\mathrm{PV}}_{t,s}$ as defined in Section~\ref{sec:est-inf}, which is zero when $\sigma_{t,s}$ is continuous at $(c,c)$.
\end{lemma}
\begin{proof}
For $\blacktriangle,\blacktriangledown\in\{+,-\}$,
\begin{align}\label{eq:asy_cov_pv}
&n\,\mathrm{Cov}\bigl(\widehat\beta^{(0)}_{t,(\blacktriangle),p}(h_t),\widehat\beta^{(0)}_{s,(\blacktriangledown),p}(h_s)\mid\mathbf{R}_t,\mathbf{R}_s\bigr) \nonumber\\
&=e_{0,p}'\,\check{\boldsymbol{\Gamma}}_{t,(\blacktriangle),p}(h_t)^{-1}\,\check{\boldsymbol{\Psi}}_{t,s,(\blacktriangle,\blacktriangledown),p}(h_t,h_s)\,\check{\boldsymbol{\Gamma}}_{s,(\blacktriangledown),p}(h_s)^{-1}\,e_{0,p} \nonumber\\
&\to_p e_{0,p}'\,\bigl[f_t(c)\widetilde{\boldsymbol{\Gamma}}_{(\blacktriangle),p}\bigr]^{-1}\,\bigl[\sigma_{t,s,(\blacktriangle,\blacktriangledown)}\,f_{t,s}(c,c)\,\widetilde{\boldsymbol{\Gamma}}_{(\blacktriangle),p}e_{0,p}\,e_{0,p}'\widetilde{\boldsymbol{\Gamma}}_{(\blacktriangledown),p}\bigr]\,\bigl[f_s(c)\widetilde{\boldsymbol{\Gamma}}_{(\blacktriangledown),p}\bigr]^{-1}\,e_{0,p} \nonumber\\
&=\frac{\sigma_{t,s,(\blacktriangle,\blacktriangledown)}\,f_{t,s}(c,c)}{f_t(c)f_s(c)}\,
\bigl(e_{0,p}'\widetilde{\boldsymbol{\Gamma}}_{(\blacktriangle),p}^{-1}\widetilde{\boldsymbol{\Gamma}}_{(\blacktriangle),p}e_{0,p}\bigr)\bigl(e_{0,p}'\widetilde{\boldsymbol{\Gamma}}_{(\blacktriangledown),p}\widetilde{\boldsymbol{\Gamma}}_{(\blacktriangledown),p}^{-1}e_{0,p}\bigr) \nonumber\\
&=\frac{\sigma_{t,s,(\blacktriangle,\blacktriangledown)}\,f_{t,s}(c,c)}{f_t(c)f_s(c)},
\end{align}
where the equality follows from \eqref{eq:cov-rescaled}. The limit substitutes Lemma~\ref{lem:aux-gamma}, with $\eta_t=\eta_s=1$ as the panel is balanced, and Lemma~\ref{lem:aux-psi-pv}. The second equality collects the density factors, and the last equality follows because each bracket equals $e_{0,p}'e_{0,p}=1$. The second result follows from substituting \eqref{eq:asy_cov_pv} into \eqref{eq:cross-decomp} for each sign combination.
\end{proof}

\noindent
The $1/n$ rate in Lemma~\ref{lem:cov-pv} is a limit statement with the joint distribution of $(R_{i,t},R_{i,s})$ fixed as $h_t,h_s\to0$, so that the bandwidths are eventually small relative to typical movements in the running variable. The cross-period covariance is driven by units carrying kernel weight in both periods.

\paragraph{Aggregate variance.}
The lemma below gives the leading terms of the variances and covariances of the per-period estimators under each sampling scheme, and hence of the variance of the aggregate estimator. Its remainders are stated in terms of $r_n=\max_{t\in\mathcal T_0\cup\{t_{\mathrm{RD}}\}}1/(nh_t)$, the largest of the per-period variance orders.

\begin{lemma}\label{lem:agg-var}
Under Assumption~\ref{ass:reg}, for $\mathrm{S}\in\{\mathrm{CS},\mathrm{PC},\mathrm{PV}\}$ and periods $t\neq s$,
\[
  V(\widehat D_t\mid\mathbf R_t)=\frac{\bar V_t}{\eta_t nh_t}+o_p(r_n),\qquad
  \mathrm{Cov}(\widehat D_t,\widehat D_s\mid\mathbf R_t,\mathbf R_s)=\ind\{\mathrm S=\mathrm{PC}\}\,\frac{\bar C^{\mathrm{PC}}_{t,s}}{nh_s}+o_p(r_n).
\]
Hence
\begin{multline*}
V^{\mathrm S}(t_{\mathrm{RD}})=\frac{\bar V_{t_{\mathrm{RD}}}}{\eta_{t_{\mathrm{RD}}}nh_{t_{\mathrm{RD}}}}+\sum_{t_0\in\mathcal T_0}\frac{\kappa_{t_0}^2\,\bar V_{t_0}}{\eta_{t_0}nh_{t_0}}\\
+\ind\{\mathrm S=\mathrm{PC}\}\Bigl(-2\sum_{t_0\in\mathcal T_0}\frac{\kappa_{t_0}\,\bar C^{\mathrm{PC}}_{t_{\mathrm{RD}},t_0}}{nh_{t_0}}+\sum_{t_0\in\mathcal T_0}\sum_{\substack{s_0\in\mathcal T_0\\s_0\neq t_0}}\frac{\kappa_{t_0}\kappa_{s_0}\,\bar C^{\mathrm{PC}}_{t_0,s_0}}{nh_{s_0}}\Bigr)+o_p(r_n).
\end{multline*}
\end{lemma}
\begin{proof}
By the law of large numbers, $n_t/n\to_p\eta_t$, so by \eqref{eq:per-period-orders}, $V(\widehat D_t(h_t)\mid\mathbf{R}_t)=\bar V_t/(\eta_tnh_t)+o_p(1/(nh_t))$ for each $t$, with $\eta_t=1$ under a balanced panel. The cross-period covariances are zero under a repeated cross section (Lemma~\ref{lem:cov-cs}), equal $\bar C^{\mathrm{PC}}_{t,s}/(nh_s)+o_p(1/(nh_s))$ under a time-constant running variable (Lemma~\ref{lem:cov-pc}), and are $O_p(1/n)=o_p(r_n)$ under a time-varying running variable (Lemma~\ref{lem:cov-pv}). The two orderings of a pair agree to first order, since $\bar C^{\mathrm{PC}}_{s,t}=\bar\rho_{t,s}\,\bar C^{\mathrm{PC}}_{t,s}$ and $h_t/h_s\to\bar\rho_{t,s}$. As the weights are fixed, substituting into the block form of $V^{\mathrm S}(t_{\mathrm{RD}})$ at the beginning of this appendix gives the last claim.
\end{proof}

\begin{proof}[Proof of Proposition~\ref{prop:cov-orders}]
By \eqref{eq:per-period-orders} and $n_t/n\to_p\eta_t$, $nh_t\,V(\widehat D_t)\to_p\bar V_t/\eta_t$, which is positive by Assumption~\ref{ass:reg}(b), (c), (f) and (g), as shown in the proof of Lemma~\ref{lem:nondegenerate}. The covariance results are Lemmas~\ref{lem:cov-cs}, \ref{lem:cov-pc} and~\ref{lem:cov-pv}.
\end{proof}

\subsection{MSE-optimal bandwidth}\label{app:est-bw}

From here on, every period uses the same bandwidth, $h_t=h$ for every $t$, so that $\bar\rho_{t,s}=1$ for every pair.

\paragraph{Limits at a common bandwidth.} Let $\bar{\boldsymbol\Xi}$ and $\bar{\boldsymbol\xi}$ be a matrix and vector, respectively, with entries
\[
  [\bar{\boldsymbol\Xi}]_{t_0t_0}=\frac{\bar V_{t_0}}{\eta_{t_0}},\qquad
  [\bar{\boldsymbol\Xi}]_{t_0s_0}=\ind\{\mathrm S=\mathrm{PC}\}\,\bar C^{\mathrm{PC}}_{t_0,s_0}\ (s_0\neq t_0),\qquad
  [\bar{\boldsymbol\xi}]_{t_0}=\ind\{\mathrm S=\mathrm{PC}\}\,\bar C^{\mathrm{PC}}_{t_{\mathrm{RD}},t_0}.
\]
Finally, let 
\[
  \bar V^{\mathrm S}(t_{\mathrm{RD}})=\frac{\bar V_{t_{\mathrm{RD}}}}{\eta_{t_{\mathrm{RD}}}}-2\boldsymbol{\kappa}'\bar{\boldsymbol\xi}+\boldsymbol{\kappa}'\bar{\boldsymbol\Xi}\boldsymbol{\kappa},\qquad
  \bar{\mathtt B}(t_{\mathrm{RD}})=\bar{\mathtt B}_{t_{\mathrm{RD}}}-\sum_{t_0\in\mathcal T_0}\kappa_{t_0}\bar{\mathtt B}_{t_0},
\]
be the aggregates of variance and bias limits.
The first result below shows the convergence to these limits. 

\begin{proposition}\label{prop:common-limits}
Under Assumption~\ref{ass:reg}, assume that all periods use a common bandwidth $h$ and that $\boldsymbol{\kappa}$ is fixed. Then $nh\,V(\widehat D_{t_{\mathrm{RD}}})\to_p\bar V_{t_{\mathrm{RD}}}/\eta_{t_{\mathrm{RD}}}$, $nh\,\boldsymbol\Xi\to_p\bar{\boldsymbol\Xi}$, $nh\,\boldsymbol\xi\to_p\bar{\boldsymbol\xi}$, and hence $nh\,V^{\mathrm S}(t_{\mathrm{RD}})\to_p\bar V^{\mathrm S}(t_{\mathrm{RD}})$.
\end{proposition}
\begin{proof}
From $h_t=h$ for all $t$, $r_n=1/(nh)$ and $\bar\rho_{t,s}=1$ for every pair of time periods. Multiplying the expansions of Lemma~\ref{lem:agg-var} by $nh$ gives, for each period $t$ and each pair $t\neq s$,
\[
  nh\,V(\widehat D_t\mid\mathbf R_t)=\frac{\bar V_t}{\eta_t}+o_p(1),\qquad
  nh\,\mathrm{Cov}(\widehat D_t,\widehat D_s\mid\mathbf R_t,\mathbf R_s)=\ind\{\mathrm S=\mathrm{PC}\}\,\bar C^{\mathrm{PC}}_{t,s}+o_p(1),
\]
with $\bar C^{\mathrm{PC}}_{t,s}$ evaluated at $\bar\rho_{t,s}=1$, where it is symmetric in $t$ and $s$. Using the above by time period shows that $nh\,V(\widehat D_{t_{\mathrm{RD}}})\to_p\bar V_{t_{\mathrm{RD}}}/\eta_{t_{\mathrm{RD}}}$, $nh\,\boldsymbol\Xi\to_p\bar{\boldsymbol\Xi}$ and $nh\,\boldsymbol\xi\to_p\bar{\boldsymbol\xi}$. Finally, multiplying the aggregate variance by $nh$, $nh\,V^{\mathrm S}(t_{\mathrm{RD}})=nh\,V(\widehat D_{t_{\mathrm{RD}}})-2\boldsymbol{\kappa}'(nh\,\boldsymbol\xi)+\boldsymbol{\kappa}'(nh\,\boldsymbol\Xi)\boldsymbol{\kappa}$, which from the above converges to $\bar V^{\mathrm S}(t_{\mathrm{RD}})$ as $\boldsymbol{\kappa}$ is fixed.
\end{proof}
\noindent
Next, we show that the regularity conditions rule out degenerate limits. 
\begin{lemma}\label{lem:nondegenerate}
Under Assumption~\ref{ass:reg}, $\bar{\boldsymbol\Xi}$ is positive definite and $\bar V^{\mathrm S}(t_{\mathrm{RD}})>0$ for every $\boldsymbol{\kappa}$.
\end{lemma}
\begin{proof}
By Proposition~\ref{prop:common-limits}
\[
  nh\,V(\widehat{\mathbf D}) \to_p \bar V(\widehat{\mathbf D})\equiv\begin{pmatrix}\bar V_{t_{\mathrm{RD}}}/\eta_{t_{\mathrm{RD}}}&\bar{\boldsymbol\xi}'\\ \bar{\boldsymbol\xi}&\bar{\boldsymbol\Xi}\end{pmatrix}.
\]
As $\bar{\boldsymbol\Xi}$ is a diagonal block of $\bar V(\widehat{\mathbf D})$ and $\bar V^{\mathrm S}(t_{\mathrm{RD}})=(1,-\boldsymbol{\kappa}')\,\bar V(\widehat{\mathbf D})\,(1,-\boldsymbol{\kappa}')'$, both claims follow once $\bar V(\widehat{\mathbf D})$ is positive definite.

First, $\sigma^2_{t,(\blacktriangle)}>0$ by Assumption~\ref{ass:reg}(f), $f_t(c)>0$ by Assumption~\ref{ass:reg}(c) and $\eta_t>0$ by Assumption~\ref{ass:reg}(g). Furthermore, $\tilde\omega_{(\blacktriangle),p}>0$, as $\widetilde{\boldsymbol\Psi}_{(\blacktriangle),p}$ is positive definite by the same argument given for $\widetilde{\boldsymbol\Gamma}_{(\blacktriangle),p}$ in Appendix~\ref{app:est-asy}, with $K^2$ in place of $K$. Together, this implies that per time-period, $\bar V_t=(\sigma^2_{t,(+)}\tilde\omega_{(+),p}+\sigma^2_{t,(-)}\tilde\omega_{(-),p})/f_t(c)>0$.

Under CS and PV, $\bar{\boldsymbol\xi}=\mathbf 0$ and $\bar{\boldsymbol\Xi}$ is diagonal, so $\bar V(\widehat{\mathbf D})$ is diagonal with positive entries $\bar V_t/\eta_t$, hence positive definite. Under PC, $\eta_t=1$ and $f_t=f$, as every unit is observed in every period with the same running variable, and $\tilde\omega_{(\blacktriangle),p}(1)=\tilde\omega_{(\blacktriangle),p}$ at a common bandwidth. For each side $\blacktriangle$, let $\boldsymbol\Sigma_{(\blacktriangle)}(t_\mathrm{RD})=[\sigma_{t,s,(\blacktriangle,\blacktriangle)}]_{t,s\in\{t_\mathrm{RD}\}\cup\mathcal T_0}$, ordered as $\widehat{\mathbf D}$, be the block of $\boldsymbol\Sigma_{(\blacktriangle)}$ for the periods in $\widehat{\mathbf D}$. The definitions of $\bar V_t$ and $\bar C^{\mathrm{PC}}_{t,s}$ then give
\[
  \bar V(\widehat{\mathbf D})=\frac{\tilde\omega_{(+),p}}{f(c)}\,\boldsymbol\Sigma_{(+)}(t_\mathrm{RD})+\frac{\tilde\omega_{(-),p}}{f(c)}\,\boldsymbol\Sigma_{(-)}(t_\mathrm{RD}),
\]
a combination, with positive coefficients, of two matrices that are positive definite, as each is a principal submatrix of a matrix that is positive definite by Assumption~\ref{ass:reg}(f).
\end{proof}

\paragraph{Single-period optimal bandwidth.}
For completeness we restate the MSE-optimal bandwidth for a single cross section from \citet{imbens2012optimal,calonico2014robust}. For a single period $t$, by \eqref{eq:per-period-orders} the leading terms of the conditional mean squared error (MSE) of $\widehat D_t(h_t)$ are
\[
  \mathrm{AMSE}_t(h_t)=\Bigl(\tfrac{h_t^{p+1}}{(p+1)!}\,\bar{\mathtt B}_t\Bigr)^2+\frac{\bar V_t}{n_th_t}.
\]
For $\bar{\mathtt B}_t\neq0$ the solution to the first-order condition of $\mathrm{AMSE}_t(h_t)$ has the form
\[
  h_t^{\mathrm{CCT}}
  =\Bigl(\tfrac{((p+1)!)^2}{2(p+1)}\,\tfrac{\bar V_t}{\bar{\mathtt B}_t^2}\Bigr)^{1/(2p+3)}n_t^{-1/(2p+3)}.
\]
For $p=1$, it is $h_t^{\mathrm{CCT}}=(\bar V_t/\bar{\mathtt B}_t^2)^{1/5}n_t^{-1/5}$.

\paragraph{Common multiple-period optimal bandwidth.} The main text discusses the bandwidth that targets the AMSE of the ATT estimator. First, we prove Proposition~\ref{prop:common-h} from the main text. 
\begin{proof}[Proof of Proposition~\ref{prop:common-h}]
We prove the proposition for a general polynomial order $p$. By \eqref{eq:bias-agg} and \eqref{eq:per-period-orders} at $h_t=h$, the conditional bias of $\widehat{\att}(t_{\mathrm{RD}})$ is $\tfrac{h^{p+1}}{(p+1)!}\bigl(\bar{\mathtt B}_{t_{\mathrm{RD}}}-\sum_{t_0\in\mathcal T_0}\kappa_{t_0}\bar{\mathtt B}_{t_0}\bigr)+o_p(h^{p+1})=\tfrac{h^{p+1}}{(p+1)!}\bar{\mathtt B}(t_{\mathrm{RD}})+o_p(h^{p+1})$. By Proposition~\ref{prop:common-limits}, its conditional variance is $V^{\mathrm S}(t_{\mathrm{RD}})=\bar V^{\mathrm S}(t_{\mathrm{RD}})/(nh)+o_p(1/(nh))$. The mean squared error is the squared bias plus the variance, and substituting the leading terms of each gives
\begin{equation}\label{eq:amse-att}
\mathrm{AMSE}^{\mathrm S}(h)=\Bigl(\tfrac{h^{p+1}}{(p+1)!}\,\bar{\mathtt B}(t_{\mathrm{RD}})\Bigr)^2+\frac{\bar V^{\mathrm S}(t_{\mathrm{RD}})}{nh},
\end{equation}
which for $p=1$ is \eqref{eq:joint_amse}. It has the form of the single-period objective $\mathrm{AMSE}_t$ of \citet{imbens2012optimal,calonico2014robust}, with the aggregate $\bar{\mathtt B}(t_{\mathrm{RD}})$ and $\bar V^{\mathrm S}(t_{\mathrm{RD}})$ in place of the per-period ones.

For the minimizer, $\bar V^{\mathrm S}(t_{\mathrm{RD}})>0$ by Lemma~\ref{lem:nondegenerate}. Write $\delta_1=\bar{\mathtt B}(t_{\mathrm{RD}})^2/((p+1)!)^2>0$ and $\delta_2=\bar V^{\mathrm S}(t_{\mathrm{RD}})/n>0$, so that $\mathrm{AMSE}^{\mathrm S}(h)=\delta_1h^{2p+2}+\delta_2/h$. Its derivative $(2p+2)\delta_1h^{2p+1}-\delta_2/h^2$ is negative for small $h$, positive for large $h$, and has a single root, so the root is the unique minimizer. Solving the first-order condition gives
\begin{equation}\label{eq:common_h_opt}
  h^{\star,\mathrm{S}}(t_{\mathrm{RD}})=\Bigl(\tfrac{((p+1)!)^2}{2(p+1)}\,\tfrac{\bar V^{\mathrm{S}}(t_{\mathrm{RD}})}{\bar{\mathtt B}(t_{\mathrm{RD}})^2}\Bigr)^{1/(2p+3)}n^{-1/(2p+3)},
\end{equation}
which for $p=1$ is the bandwidth in Proposition~\ref{prop:common-h}.
\end{proof}

\noindent
In practice, under the common bandwidth every period also uses the same bandwidth for the bias correction, $b(t_{\mathrm{RD}})$ of the estimation procedure below. The ratio of the main bandwidth to it, $h^{\star,\mathrm{S}}(t_{\mathrm{RD}})/b(t_{\mathrm{RD}})$, tends to zero, a case covered by Assumption~\ref{ass:reg}(g).
When $\mathcal{T}_0=\varnothing$, $h^{\star,\mathrm{CS}}(t_{\mathrm{RD}})$ reduces to $h^{\mathrm{CCT}}_{t_{\mathrm{RD}}}$, with $\eta_{t_{\mathrm{RD}}}n$ in place of $n_{t_{\mathrm{RD}}}$. Next we discuss estimation and regularization.

\paragraph{Bandwidth estimation.} The proposed estimation of the bandwidth builds on \citet{calonico2014robust} and computes four bandwidths in turn. For brevity we suppress in the notation the dependence on $\boldsymbol{\kappa}$ and the sampling scheme.

Before stating the procedure, it is convenient to formulate several general formulas. Let
\[
\widehat{\mathrm{DiDisc}}(t_\mathrm{RD};v,o,h)=\widehat D^{(v)}_{t_\mathrm{RD},o}(h)-\sum_{t_0\in\mathcal T_0}\kappa_{t_0}\widehat D^{(v)}_{t_0,o}(h)
\]
be the diff-in-disc estimator in the $v$-th derivative, that is, the discontinuity in $t_\mathrm{RD}$ minus the $\boldsymbol{\kappa}$-weighted sum of the discontinuities in the comparison periods, using polynomial order $o$, at bandwidth $h$, where $\widehat D^{(v)}_{t,o}(h)$ is \eqref{eq:dhat-v} with $o$ in place of $p$ and $h_t=h$, and $\widehat{\mathrm{DiDisc}}(t_\mathrm{RD};0,p,h)$ is the estimator \eqref{eq:estimation_multiple} with $h_t=h$ for every period. Let
\begin{multline}\label{eq:vhat-didisc}
\widehat V\bigl(\widehat{\mathrm{DiDisc}}(t_\mathrm{RD};v,o,h)\bigr)=\widehat V\bigl(\widehat D^{(v)}_{t_\mathrm{RD},o}(h)\bigr)+\sum_{t_0\in\mathcal T_0}\kappa_{t_0}^2\,\widehat V\bigl(\widehat D^{(v)}_{t_0,o}(h)\bigr)\\
+\ind\{\mathrm S=\mathrm{PC}\}\Bigl(-2\sum_{t_0\in\mathcal T_0}\kappa_{t_0}\,\widehat{\mathrm{Cov}}\bigl(\widehat D^{(v)}_{t_{\mathrm{RD}},o}(h),\widehat D^{(v)}_{t_0,o}(h)\bigr)+\sum_{t_0\in\mathcal T_0}\sum_{\substack{s_0\in\mathcal T_0\\s_0\neq t_0}}\kappa_{t_0}\kappa_{s_0}\,\widehat{\mathrm{Cov}}\bigl(\widehat D^{(v)}_{t_0,o}(h),\widehat D^{(v)}_{s_0,o}(h)\bigr)\Bigr),
\end{multline}
where variances and covariances are estimated as discussed in Appendix~\ref{app:est-var}, with $o$ in place of $p$ and $h_t=h$. The cross-period covariances enter under PC only, as they are zero under CS and of lower order under PV (Proposition~\ref{prop:cov-orders}). Let
\[
\widehat{\bar V}(t_\mathrm{RD};v,o,h)=nh^{2v+1}\,\widehat V\bigl(\widehat{\mathrm{DiDisc}}(t_\mathrm{RD};v,o,h)\bigr)
\]
be the estimator of its variance limit, and 
\begin{multline}\label{eq:bbar-hat}
\widehat{\bar{\mathtt B}}(t_\mathrm{RD};v,o,o^\prime,h,h^\prime)=\widehat\beta^{(o+1)}_{t_{\mathrm{RD}},(+),o'}(h')\,B_{t_{\mathrm{RD}},(+),v,o}(h)-\widehat\beta^{(o+1)}_{t_{\mathrm{RD}},(-),o'}(h')\,B_{t_{\mathrm{RD}},(-),v,o}(h)\\
-\sum_{t_0\in\mathcal T_0}\kappa_{t_0}\bigl(\widehat\beta^{(o+1)}_{t_0,(+),o'}(h')\,B_{t_0,(+),v,o}(h)-\widehat\beta^{(o+1)}_{t_0,(-),o'}(h')\,B_{t_0,(-),v,o}(h)\bigr)
\end{multline}
be the estimator of the bias constant, that is, the $\boldsymbol{\kappa}$-weighted combination of $m^{(o+1)}_{t,(\blacktriangle)}B_{t,(\blacktriangle),v,o}(h)$ in \eqref{eq:bias-v} with $o$ in place of $p$, with the $(o+1)$-th derivatives estimated from regressions using polynomial order $o'$ at a bandwidth $h'$, where $o'\ge o+1$, and the design moments \eqref{eq:design-moment} at $h$. Finally, for $V>0$ and $B\neq0$, let
\begin{equation}\label{eq:h_opt_formula}
h^{\mathrm{opt}}(v,o,V,B)=\left(\frac{((o+1)!)^2\,(2v+1)}{2(o+1-v)}\frac{V}{B^2}\right)^{1/(2o+3)}n^{-1/(2o+3)},
\end{equation}
be the general formula for the AMSE-optimal bandwidth in derivative $v$ using polynomial order $o\ge v$, given a variance limit $V$ and a bias constant $B$. The formula is a generalized version of \eqref{eq:common_h_opt}, as the bandwidth estimation procedure requires estimation of higher-order derivatives. The procedure to estimate $h^{\star,\mathrm S}(t_{\mathrm{RD}})$ computes, in turn,
\begin{align*}
h^{\mathrm{rot}}&:=2.576\,\min\bigl\{\widehat\sigma_R,\ \widehat{\mathrm{IQR}}_R/1.349\bigr\}\,\bar n^{-1/5},\\
d(t_{\mathrm{RD}})&:=h^{\mathrm{opt}}\bigl(q+1,q+1,\ \widehat{\bar V}(t_{\mathrm{RD}};q+1,q+1,h^{\mathrm{rot}}),\ \widehat{\bar{\mathtt B}}(t_{\mathrm{RD}};q+1,q+1,q+2,h^{\mathrm{rot}},\bar r)\bigr),\\
b(t_{\mathrm{RD}})&:=h^{\mathrm{opt}}\bigl(p+1,q,\ \widehat{\bar V}(t_{\mathrm{RD}};p+1,q,h^{\mathrm{rot}}),\ \widehat{\bar{\mathtt B}}(t_{\mathrm{RD}};p+1,q,q+1,h^{\mathrm{rot}},d(t_{\mathrm{RD}}))\bigr),\\
\widehat h^{\star}(t_{\mathrm{RD}})&:=h^{\mathrm{opt}}\bigl(0,p,\ \widehat{\bar V}(t_{\mathrm{RD}};0,p,h^{\mathrm{rot}}),\ \widehat{\bar{\mathtt B}}(t_{\mathrm{RD}};0,p,q,h^{\mathrm{rot}},b(t_{\mathrm{RD}}))\bigr),
\end{align*}
where $\widehat\sigma_R$ and $\widehat{\mathrm{IQR}}_R$ are the standard deviation and interquartile range of the running variables of all periods pooled, respectively, $\bar n$ is the average number of units per period, and $\bar r$ is the largest distance of the running variable to the cutoff in each period and on each side. The rule-of-thumb (ROT) bandwidth $h^{\mathrm{rot}}$ is based on the ROT bandwidth of \citet{calonico2014robust} for the triangular kernel. Each bias constant needs a higher-order derivative: $\widehat h^{\star}$ uses the $(p+1)$-th derivatives at $b$, $b$ the $(q+1)$-th derivatives at $d$, and $d$ the $(q+2)$-th derivatives at $\bar r$, that is, from a kernel-weighted fit over the whole support. By construction, $b$ and $d$ are the AMSE-optimal bandwidths for the derivatives that the line below uses. The last line is \eqref{eq:common_h_opt} with these estimates in place of $\bar V^{\mathrm S}(t_{\mathrm{RD}})$ and $\bar{\mathtt B}(t_{\mathrm{RD}})$.
Every period uses the same bandwidths, i.e., $h^{\mathrm{rot}}$, $d$ and $b$ are common to all periods. Under PC this is the procedure of \citet{calonico2014robust} applied to the composite outcome $Y_{i,t_{\mathrm{RD}}}-\sum_{t_0\in\mathcal T_0}\kappa_{t_0}Y_{i,t_0}$. As in the software of \citet{calonico2014robust}, the implementation caps each bandwidth at the largest distance of the running variable to the cutoff, and adjusts the bandwidths for mass points in the running variable. We omit these details. 

\paragraph{Regularization.}
Following \citet{imbens2012optimal,calonico2014robust}, we regularize the bandwidth. Returning to the estimator of the optimal bandwidth $\widehat h^{\star}(t_{\mathrm{RD}})$ discussed above, 
a spuriously near-zero estimated bias constant widens it. Since the expected square of an estimated bias constant is approximately its square plus the variance of its estimator, \citet{imbens2012optimal} and \citet{calonico2014robust} use instead the squared estimated bias constant plus a tuning parameter (denoted $\lambda$) times its estimated variance. Implementing a similar procedure, we use
\[
\widehat{\bar{\mathtt B}}(t_{\mathrm{RD}};0,p,q,h^{\mathrm{rot}},b(t_\mathrm{RD}))^2+\lambda\,\widehat V\bigl(\widehat{\bar{\mathtt B}}(t_{\mathrm{RD}};0,p,q,h^{\mathrm{rot}},b(t_\mathrm{RD}))\bigr)
\]
in place of $\widehat{\bar{\mathtt B}}(t_{\mathrm{RD}};0,p,q,h^{\mathrm{rot}},b(t_\mathrm{RD}))^2$ in $\widehat h^{\star}(t_{\mathrm{RD}})$, with default $\lambda=3$, where, for general arguments,
\begin{equation}\label{eq:vhat-bbar}
\begin{aligned}
&\widehat V\bigl(\widehat{\bar{\mathtt B}}(t_\mathrm{RD};v,o,o',h,h')\bigr)\\
&\quad=\sum_{\blacktriangle\in\{+,-\}}\Bigl[B_{t_\mathrm{RD},(\blacktriangle),v,o}(h)^2\,\widehat V\bigl(\widehat\beta^{(o+1)}_{t_\mathrm{RD},(\blacktriangle),o'}(h')\bigr)\\
&\qquad+\sum_{t_0\in\mathcal T_0}\kappa_{t_0}^2\,B_{t_0,(\blacktriangle),v,o}(h)^2\,\widehat V\bigl(\widehat\beta^{(o+1)}_{t_0,(\blacktriangle),o'}(h')\bigr)\\
&\qquad+\ind\{\mathrm S=\mathrm{PC}\}\Bigl(-2\sum_{t_0\in\mathcal T_0}\kappa_{t_0}\,B_{t_\mathrm{RD},(\blacktriangle),v,o}(h)\,B_{t_0,(\blacktriangle),v,o}(h)\,\widehat{\mathrm{Cov}}\bigl(\widehat\beta^{(o+1)}_{t_\mathrm{RD},(\blacktriangle),o'}(h'),\widehat\beta^{(o+1)}_{t_0,(\blacktriangle),o'}(h')\bigr)\\
&\qquad\quad+\sum_{t_0\in\mathcal T_0}\sum_{\substack{s_0\in\mathcal T_0\\ s_0\neq t_0}}\kappa_{t_0}\kappa_{s_0}\,B_{t_0,(\blacktriangle),v,o}(h)\,B_{s_0,(\blacktriangle),v,o}(h)\,\widehat{\mathrm{Cov}}\bigl(\widehat\beta^{(o+1)}_{t_0,(\blacktriangle),o'}(h'),\widehat\beta^{(o+1)}_{s_0,(\blacktriangle),o'}(h')\bigr)\Bigr)\Bigr],
\end{aligned}
\end{equation}
and the variances and covariances of the $\widehat\beta^{(o+1)}$ are estimated as discussed in Appendix~\ref{app:est-var}, with $o+1$ in place of $v$, $o'$ in place of $p$ and $h_t=h'$. The bandwidth $b(t_{\mathrm{RD}})$ is regularized in the same way, with the arguments of its line. As in \citet{calonico2014robust}, $d(t_{\mathrm{RD}})$ is not regularized.

\subsection{Variance-minimizing weights}\label{app:est-weights}

This subsection studies the weights that minimize the variance (MV), denoted $\boldsymbol{\kappa}^{\mathrm{MV}}$, of \eqref{eq:w-mv}. First, Lemma~\ref{lem:w-mv} gives their closed form. Second, Lemma~\ref{lem:w-mv-h} shows that at a common bandwidth they do not depend on the bandwidth to first order. Third, Lemmas~\ref{lem:w-hat-consistent} and~\ref{lem:w-mv-plugin} show that estimating the weights does not change the estimator to first order. Building on these results, Lemma~\ref{lem:three-steps} shows that the three-step procedure to estimate the bandwidth and weights of Section~\ref{sec:est-weights} suffices. Finally, we prove Proposition~\ref{prop:w-mv} from the main text. 

Let $\mathbf L\boldsymbol{\kappa}=\boldsymbol\ell$ be the constraints of Corollary~\ref{cor:trend}. In the constant case $\mathbf L$ is a row of ones of length $|\mathcal T_0|$ and $\boldsymbol\ell=1$. In the linear case $\mathbf L$ also has the comparison periods $t_0$ as a second row and $\boldsymbol\ell=(1,t_{\mathrm{RD}})'$. Since Corollary~\ref{cor:trend} requires at least one comparison period in the constant case and at least two in the linear case, $\mathbf L$ has full row rank. Recall $\boldsymbol\Xi=V(\widehat{\mathbf D}_0)$, $\boldsymbol\xi=\mathrm{Cov}(\widehat{\mathbf D}_0,\widehat D_{t_{\mathrm{RD}}})$, and
$V^{\mathrm S}(t_{\mathrm{RD}};\boldsymbol{\kappa})=V(\widehat D_{t_{\mathrm{RD}}})-2\boldsymbol{\kappa}'\boldsymbol\xi+\boldsymbol{\kappa}'\boldsymbol\Xi\boldsymbol{\kappa}$.

\begin{lemma}\label{lem:w-mv}
Assume $\boldsymbol\Xi$ is positive definite. Then the problem
\[
  \min_{\boldsymbol{\kappa}}\ V^{\mathrm S}(t_{\mathrm{RD}};\boldsymbol{\kappa})
  \quad\text{subject to}\quad \mathbf L\boldsymbol{\kappa}=\boldsymbol\ell
\]
has the solution
\[
  \boldsymbol{\kappa}^{\mathrm{MV}}=\boldsymbol\Xi^{-1}\Bigl[\boldsymbol\xi+\mathbf L'\bigl(\mathbf L\boldsymbol\Xi^{-1}\mathbf L'\bigr)^{-1}\bigl(\boldsymbol\ell-\mathbf L\boldsymbol\Xi^{-1}\boldsymbol\xi\bigr)\Bigr].
\]
Under a repeated cross section, in the constant case $\kappa^{\mathrm{MV}}_{t_0}=V(\widehat D_{t_0})^{-1}\big/\sum_{s_0\in\mathcal T_0}V(\widehat D_{s_0})^{-1}$.
\end{lemma}
\begin{proof}
The Lagrangian of the problem, with multipliers $2\boldsymbol\zeta$, is
\[
  V(\widehat D_{t_{\mathrm{RD}}})-2\boldsymbol{\kappa}'\boldsymbol\xi+\boldsymbol{\kappa}'\boldsymbol\Xi\boldsymbol{\kappa}
  -2\boldsymbol\zeta'(\mathbf L\boldsymbol{\kappa}-\boldsymbol\ell),
\]
with first-order conditions $-2\boldsymbol{\xi}+2\boldsymbol{\Xi}\boldsymbol{\kappa}-2\mathbf{L}^{\prime}\boldsymbol{\zeta}  =\mathbf 0$ and $\mathbf{L}\boldsymbol{\kappa}-\boldsymbol{\ell}  =\mathbf 0$.
Rearranging the first gives \begin{equation}\label{eq:w_foc}
\boldsymbol{\kappa}=\boldsymbol\Xi^{-1}(\boldsymbol\xi+\mathbf L'\boldsymbol\zeta).
\end{equation} As $\boldsymbol\Xi$ is positive definite and $\mathbf L$ has full row rank, $\mathbf L\boldsymbol\Xi^{-1}\mathbf L'$ is positive definite, hence invertible, so substituting $\boldsymbol{\kappa}$ into the second first-order condition and rearranging gives
\begin{equation}\label{eq:zet_foc}
\boldsymbol\zeta=(\mathbf L\boldsymbol\Xi^{-1}\mathbf L')^{-1}(\boldsymbol\ell-\mathbf L\boldsymbol\Xi^{-1}\boldsymbol\xi).
\end{equation}
Substituting $\boldsymbol\zeta$ from \eqref{eq:zet_foc} back into \eqref{eq:w_foc} gives the formula of $\boldsymbol{\kappa}^{\mathrm{MV}}$ in the lemma. As $\boldsymbol\Xi$ is positive definite, the objective is strictly convex in $\boldsymbol{\kappa}$, and as the constraints are linear, $\boldsymbol{\kappa}^{\mathrm{MV}}$ is its unique minimizer under the constraints.

Under a repeated cross section the cross-period covariances are zero (Lemma~\ref{lem:cov-cs}), so $\boldsymbol\xi=\boldsymbol 0$ and $\boldsymbol\Xi$ is diagonal with entries $V(\widehat D_{t_0})$. In the constant case $\mathbf L$ is a row of ones, so every entry of $\mathbf L'\boldsymbol\zeta$ equals $\zeta$, and \eqref{eq:w_foc} can be simplified into $\kappa_{t_0}=\zeta/V(\widehat D_{t_0})$ for each $t_0\in\mathcal T_0$ and \eqref{eq:zet_foc} can be simplified into $\zeta=1/\sum_{s_0\in\mathcal T_0}V(\widehat D_{s_0})^{-1}$. Substituting $\zeta$ into $\kappa_{t_0}$ gives the inverse-variance weights.
\end{proof}

\noindent
The weights $\boldsymbol{\kappa}^{\mathrm{MV}}$ derived in Lemma \ref{lem:w-mv} depend on the variances and covariances in $\boldsymbol\Xi$ and $\boldsymbol\xi$, which are computed using bandwidths, so one might expect $\boldsymbol{\kappa}^{\mathrm{MV}}$ to depend on the bandwidth. Below we show that when all periods use the same bandwidth $h$, $\boldsymbol{\kappa}^{\mathrm{MV}}$ does not depend on $h$ to first order. Let 
\[
  \bar{\boldsymbol{\kappa}}^{\mathrm{MV}}=\bar{\boldsymbol\Xi}^{-1}\Bigl[\bar{\boldsymbol\xi}+\mathbf L'\bigl(\mathbf L\bar{\boldsymbol\Xi}^{-1}\mathbf L'\bigr)^{-1}\bigl(\boldsymbol\ell-\mathbf L\bar{\boldsymbol\Xi}^{-1}\bar{\boldsymbol\xi}\bigr)\Bigr],
\]
which does not depend on $h$. Here $\bar{\boldsymbol\Xi}$ and $\bar{\boldsymbol\xi}$ are defined as in Section~\ref{app:est-bw}.
\begin{lemma}\label{lem:w-mv-h}
Under Assumption~\ref{ass:reg}, if all periods use a common bandwidth $h$, then $\boldsymbol{\kappa}^{\mathrm{MV}}\to_p\bar{\boldsymbol{\kappa}}^{\mathrm{MV}}$.
\end{lemma}
\begin{proof}
By Proposition~\ref{prop:common-limits}, $nh\,\boldsymbol\Xi\to_p\bar{\boldsymbol\Xi}$ and $nh\,\boldsymbol\xi\to_p\bar{\boldsymbol\xi}$. Multiplying $\boldsymbol\Xi$ and $\boldsymbol\xi$ by the same $nh$ does not change the formula of Lemma~\ref{lem:w-mv}, since 
\[
  (nh\boldsymbol\Xi)^{-1}(nh\boldsymbol\xi)=\boldsymbol\Xi^{-1}\boldsymbol\xi,\qquad
  (nh\boldsymbol\Xi)^{-1}\mathbf L'\bigl(\mathbf L(nh\boldsymbol\Xi)^{-1}\mathbf L'\bigr)^{-1}=\boldsymbol\Xi^{-1}\mathbf L'\bigl(\mathbf L\boldsymbol\Xi^{-1}\mathbf L'\bigr)^{-1}.
\]
As $\bar{\boldsymbol\Xi}$ is positive definite (Lemma~\ref{lem:nondegenerate}), the formula is continuous at $(\bar{\boldsymbol\Xi},\bar{\boldsymbol\xi})$, and $\boldsymbol\Xi$ is positive definite with probability approaching one. Hence by the continuous mapping theorem
\[\boldsymbol{\kappa}^{\mathrm{MV}}=\left(nh\boldsymbol{\Xi}\right)^{-1}\Bigl[nh\boldsymbol{\xi}+\mathbf{L}'\bigl(\mathbf{L}\left(nh\boldsymbol{\Xi}\right)^{-1}\mathbf{L}'\bigr)^{-1}\bigl(\boldsymbol{\ell}-\mathbf{L}\left(nh\boldsymbol{\Xi}\right)^{-1}nh\boldsymbol{\xi}\bigr)\Bigr]\to_{p}\bar{\boldsymbol{\kappa}}^{\mathrm{MV}}.\]
\end{proof}

\noindent
For intuition on the above result, note that the variances and covariances in $\boldsymbol\Xi$ and $\boldsymbol\xi$ all shrink at the same rate as $h$ grows, each being approximately a constant divided by $nh$, and multiplying all of them by the same number does not change $\boldsymbol{\kappa}^{\mathrm{MV}}$. So $\boldsymbol{\kappa}^{\mathrm{MV}}$ is approximately the weights computed from these constants, the matrix $\bar{\boldsymbol\Xi}$ and vector $\bar{\boldsymbol\xi}$, which do not involve $h$.

Next, we turn to analyze an estimator of the weights. 
Let $\widehat{\boldsymbol{\kappa}}^{\mathrm{MV}}$ be the plug-in weights: the formula of Lemma~\ref{lem:w-mv} with $\boldsymbol\Xi$ and $\boldsymbol\xi$ replaced by estimates $\widehat{\boldsymbol\Xi}$ and $\widehat{\boldsymbol\xi}$, such as the estimators discussed in Appendix~\ref{app:est-var}. We begin by showing that the plug-in weights $\widehat{\boldsymbol{\kappa}}^{\mathrm{MV}}$ converge to the same limit as the infeasible weights $\boldsymbol{\kappa}^{\mathrm{MV}}$ if the estimated variances and covariances are consistent.

\begin{lemma}\label{lem:w-hat-consistent}
Under Assumption~\ref{ass:reg}, if all periods use a common bandwidth $h$, $nh\,\widehat{\boldsymbol\Xi}\to_p\bar{\boldsymbol\Xi}$ and $nh\,\widehat{\boldsymbol\xi}\to_p\bar{\boldsymbol\xi}$, then $\widehat{\boldsymbol{\kappa}}^{\mathrm{MV}}\to_p\bar{\boldsymbol{\kappa}}^{\mathrm{MV}}$, so $\widehat{\boldsymbol{\kappa}}^{\mathrm{MV}}-\boldsymbol{\kappa}^{\mathrm{MV}}=o_p(1)$.
\end{lemma}
\begin{proof}
As $nh\,\widehat{\boldsymbol\Xi}\to_p\bar{\boldsymbol\Xi}$ and $\bar{\boldsymbol\Xi}$ is positive definite (Lemma~\ref{lem:nondegenerate}), $\widehat{\boldsymbol\Xi}$ is positive definite with probability approaching one. By the proof of Lemma~\ref{lem:w-mv} and the definition of $\widehat{\boldsymbol{\kappa}}^{\mathrm{MV}}$ as a plug-in estimator using the same formula as $\boldsymbol{\kappa}^{\mathrm{MV}}$ with $(\widehat{\boldsymbol\Xi},\widehat{\boldsymbol\xi})$ in place of $(\boldsymbol\Xi,\boldsymbol\xi)$, it satisfies the constraints $\mathbf{L}\boldsymbol{\kappa}=\boldsymbol{\ell}$ with probability approaching one. The proof of Lemma~\ref{lem:w-mv-h} with $(nh\,\widehat{\boldsymbol\Xi},nh\,\widehat{\boldsymbol\xi})$ in place of $(nh\,\boldsymbol\Xi,nh\,\boldsymbol\xi)$ gives $\widehat{\boldsymbol{\kappa}}^{\mathrm{MV}}\to_p\bar{\boldsymbol{\kappa}}^{\mathrm{MV}}$. Combining with $\boldsymbol{\kappa}^{\mathrm{MV}}\to_p\bar{\boldsymbol{\kappa}}^{\mathrm{MV}}$ (Lemma~\ref{lem:w-mv-h}) concludes the proof.
\end{proof}
\noindent
The next result shows that using the plug-in weights $\widehat{\boldsymbol{\kappa}}^{\mathrm{MV}}$ does not change the estimator to first order. The intuition is that, under the identification assumption on the trend, all weights that satisfy the constraints give the same combination of the true comparison discontinuities. The estimation error in the weights therefore multiplies only the estimation error of the discontinuities, a product of two small terms.

\begin{lemma}\label{lem:w-mv-plugin}
Under the conditions of Lemma~\ref{lem:w-hat-consistent}, assume $\att(t_{\mathrm{RD}})=D_{t_{\mathrm{RD}}}-\boldsymbol{\kappa}'\mathbf D_0$ for every $\boldsymbol{\kappa}$ with $\mathbf L\boldsymbol{\kappa}=\boldsymbol\ell$ as in Corollary~\ref{cor:trend}. Then
\[
  \widehat{\att}(t_{\mathrm{RD}};\widehat{\boldsymbol{\kappa}}^{\mathrm{MV}})-\widehat{\att}(t_{\mathrm{RD}};\boldsymbol{\kappa}^{\mathrm{MV}})=o_p\bigl((nh)^{-1/2}\bigr).
\]
\end{lemma}
\begin{proof}
As the RD period gets the same weight, the difference is $\widehat{\att}(t_{\mathrm{RD}};\widehat{\boldsymbol{\kappa}}^{\mathrm{MV}})-\widehat{\att}(t_{\mathrm{RD}};\boldsymbol{\kappa}^{\mathrm{MV}})=-(\widehat{\boldsymbol{\kappa}}^{\mathrm{MV}}-\boldsymbol{\kappa}^{\mathrm{MV}})'\widehat{\mathbf D}_0$. We begin by considering the difference under the population-level observed discontinuities, i.e., $(\widehat{\boldsymbol{\kappa}}^{\mathrm{MV}}-\boldsymbol{\kappa}^{\mathrm{MV}})'\mathbf D_0$. Next, for any two weight vectors $\boldsymbol{\kappa}_1,\boldsymbol{\kappa}_2$ that satisfy the constraints, by assumption 
\[
  \boldsymbol{\kappa}_1^{\prime}\mathbf D_0=\boldsymbol{\kappa}_2^{\prime}\mathbf D_0=D_{t_{\mathrm{RD}}}-\att(t_{\mathrm{RD}}) \quad\Longrightarrow\quad
  (\boldsymbol{\kappa}_1-\boldsymbol{\kappa}_2)'\mathbf D_0=0.
\]
By the proof of Lemma~\ref{lem:w-hat-consistent}, $\mathbf L\widehat{\boldsymbol{\kappa}}^{\mathrm{MV}}=\boldsymbol\ell$ with probability approaching one, and $\boldsymbol{\kappa}^{\mathrm{MV}}$ satisfies the constraints by construction.
This implies that $-(\widehat{\boldsymbol{\kappa}}^{\mathrm{MV}}-\boldsymbol{\kappa}^{\mathrm{MV}})'\widehat{\mathbf D}_0=-(\widehat{\boldsymbol{\kappa}}^{\mathrm{MV}}-\boldsymbol{\kappa}^{\mathrm{MV}})'(\widehat{\mathbf D}_0-\mathbf D_0)$, as this adds a zero to the initial difference. 

Next, we consider the order. By Lemma~\ref{lem:dhat-order}, $\widehat D_{t_0}-D_{t_0}=O_p\bigl(h^{p+1}+(nh)^{-1/2}\bigr)$ for each $t_0\in\mathcal T_0$, and multiplying by $\widehat{\boldsymbol{\kappa}}^{\mathrm{MV}}-\boldsymbol{\kappa}^{\mathrm{MV}}=o_p(1)$ gives the order $o_p\bigl(h^{p+1}+(nh)^{-1/2}\bigr)$. As $h^{p+1}=O((nh)^{-1/2})$ by Assumption~\ref{ass:reg}(g) at a common bandwidth, this is $o_p((nh)^{-1/2})$.
\end{proof}

\noindent
Having finished stating the properties of the MV weights, we next consider two implications. The first, formalized in the result below, returns to the three-step algorithm discussed in Section~\ref{sec:est-weights}: estimate the common optimal bandwidth at some initial weights, estimate the MV weights at that bandwidth, and finally estimate the common optimal bandwidth again at the second step's MV weights. Lemmas~\ref{lem:w-mv-h} and~\ref{lem:w-hat-consistent} imply that these three steps suffice, as the bandwidth does not affect the infeasible weights to first order, and the plug-in and infeasible weights converge to the same limit. For the formal result, denote by $h^{\star,\mathrm S}(t_{\mathrm{RD}};\boldsymbol{\kappa})$ the infeasible common bandwidth \eqref{eq:common_h_opt} with weights $\boldsymbol{\kappa}$, and by $\widehat h_3$ the estimated bandwidth from the above three-step procedure, that is, the estimated bandwidth $\widehat h^{\star}(t_{\mathrm{RD}})$ of Appendix~\ref{app:est-bw} computed at the second-step weights $\widehat{\boldsymbol{\kappa}}^{\mathrm{MV}}$. For simplicity, we abstract from regularization, i.e., set $\lambda=0$. 

\begin{lemma}\label{lem:three-steps}
Let the conditions of Lemma~\ref{lem:w-hat-consistent} hold at a non-random first-step bandwidth, and let $\bar{\mathtt B}(t_{\mathrm{RD}})\neq0$ at $\bar{\boldsymbol{\kappa}}^{\mathrm{MV}}$. Let the estimators of Appendix~\ref{app:est-bw} be consistent: for every period $t$, the period-$t$ term of \eqref{eq:bbar-hat} in $\widehat h_3$ converges to its limit in Appendix~\ref{app:est-asy},
\[
\widehat\beta^{(p+1)}_{t,(+),q}(b(t_{\mathrm{RD}}))\,B_{t,(+),0,p}(h^{\mathrm{rot}})-\widehat\beta^{(p+1)}_{t,(-),q}(b(t_{\mathrm{RD}}))\,B_{t,(-),0,p}(h^{\mathrm{rot}})\to_p\bar{\mathtt B}_t,
\]
the variances and covariances in \eqref{eq:vhat-didisc} at $(v,o,h)=(0,p,h^{\mathrm{rot}})$, scaled by $nh^{\mathrm{rot}}$, converge in probability to their limits in Proposition~\ref{prop:common-limits}. Then $\widehat h_3/h^{\star,\mathrm S}(t_{\mathrm{RD}};\bar{\boldsymbol{\kappa}}^{\mathrm{MV}})\to_p1$.
\end{lemma}
\begin{proof}
Write $\bar V^{\mathrm S}(t_{\mathrm{RD}};\boldsymbol{\kappa})$ and $\bar{\mathtt B}(t_{\mathrm{RD}};\boldsymbol{\kappa})$ for the constants of Section~\ref{sec:est-bw} at weights $\boldsymbol{\kappa}$, and $\widehat{\bar V}{}^{\mathrm S}(t_{\mathrm{RD}};\boldsymbol{\kappa})$ and $\widehat{\bar{\mathtt B}}(t_{\mathrm{RD}};\boldsymbol{\kappa})$ for $\widehat{\bar V}(t_{\mathrm{RD}};0,p,h^{\mathrm{rot}})$ and $\widehat{\bar{\mathtt B}}(t_{\mathrm{RD}};0,p,q,h^{\mathrm{rot}},b(t_{\mathrm{RD}}))$ of Appendix~\ref{app:est-bw} at weights $\boldsymbol{\kappa}$, with $b(t_{\mathrm{RD}})$ held at its value in $\widehat h_3$. By \eqref{eq:common_h_opt}, the infeasible bandwidth at the limit weights is
\[
  h^{\star,\mathrm S}(t_{\mathrm{RD}};\bar{\boldsymbol{\kappa}}^{\mathrm{MV}})=\Bigl(\frac{((p+1)!)^2}{2(p+1)}\,\frac{\bar V^{\mathrm S}(t_{\mathrm{RD}};\bar{\boldsymbol{\kappa}}^{\mathrm{MV}})}{\bar{\mathtt B}(t_{\mathrm{RD}};\bar{\boldsymbol{\kappa}}^{\mathrm{MV}})^2}\Bigr)^{1/(2p+3)}n^{-1/(2p+3)},
\]
and, by Appendix~\ref{app:est-bw}, the third-step bandwidth is the same formula at the second-step weights and the estimated constants,
\[
  \widehat h_3=\Bigl(\frac{((p+1)!)^2}{2(p+1)}\,\frac{\widehat{\bar V}{}^{\mathrm S}(t_{\mathrm{RD}};\widehat{\boldsymbol{\kappa}}^{\mathrm{MV}})}{\widehat{\bar{\mathtt B}}(t_{\mathrm{RD}};\widehat{\boldsymbol{\kappa}}^{\mathrm{MV}})^2}\Bigr)^{1/(2p+3)}n^{-1/(2p+3)}.
\]
The constant $((p+1)!)^2/(2(p+1))$ and the factor $n^{-1/(2p+3)}$ cancel in the ratio,
\[
  \frac{\widehat h_3}{h^{\star,\mathrm S}(t_{\mathrm{RD}};\bar{\boldsymbol{\kappa}}^{\mathrm{MV}})}
  =\Biggl(\frac{\widehat{\bar V}{}^{\mathrm S}(t_{\mathrm{RD}};\widehat{\boldsymbol{\kappa}}^{\mathrm{MV}})}{\bar V^{\mathrm S}(t_{\mathrm{RD}};\bar{\boldsymbol{\kappa}}^{\mathrm{MV}})}\cdot
  \frac{\bar{\mathtt B}(t_{\mathrm{RD}};\bar{\boldsymbol{\kappa}}^{\mathrm{MV}})^2}{\widehat{\bar{\mathtt B}}(t_{\mathrm{RD}};\widehat{\boldsymbol{\kappa}}^{\mathrm{MV}})^2}\Biggr)^{1/(2p+3)}.
\]
By Lemma~\ref{lem:w-hat-consistent}, $\widehat{\boldsymbol{\kappa}}^{\mathrm{MV}}\to_p\bar{\boldsymbol{\kappa}}^{\mathrm{MV}}$. The estimate $\widehat{\bar V}{}^{\mathrm S}(t_{\mathrm{RD}};\boldsymbol{\kappa})$ is a quadratic in $\boldsymbol{\kappa}$ and $\widehat{\bar{\mathtt B}}(t_{\mathrm{RD}};\boldsymbol{\kappa})$ is linear in $\boldsymbol{\kappa}$, with coefficients that converge in probability to those of $\bar V^{\mathrm S}(t_{\mathrm{RD}};\boldsymbol{\kappa})$ and $\bar{\mathtt B}(t_{\mathrm{RD}};\boldsymbol{\kappa})$ by assumption. Hence, by the continuous mapping theorem, $\widehat{\bar V}{}^{\mathrm S}(t_{\mathrm{RD}};\widehat{\boldsymbol{\kappa}}^{\mathrm{MV}})\to_p\bar V^{\mathrm S}(t_{\mathrm{RD}};\bar{\boldsymbol{\kappa}}^{\mathrm{MV}})$ and $\widehat{\bar{\mathtt B}}(t_{\mathrm{RD}};\widehat{\boldsymbol{\kappa}}^{\mathrm{MV}})\to_p\bar{\mathtt B}(t_{\mathrm{RD}};\bar{\boldsymbol{\kappa}}^{\mathrm{MV}})$. The limits satisfy $\bar V^{\mathrm S}(t_{\mathrm{RD}};\bar{\boldsymbol{\kappa}}^{\mathrm{MV}})>0$ by Lemma~\ref{lem:nondegenerate} and $\bar{\mathtt B}(t_{\mathrm{RD}};\bar{\boldsymbol{\kappa}}^{\mathrm{MV}})\neq0$ by assumption, so the ratio above is a continuous function of these quantities at their limits, where it equals one. The continuous mapping theorem gives the result.
\end{proof}

\noindent
A fourth step would recompute the weights at $\widehat h_3$, which, treating $\widehat h_3$ as non-random, again converge to $\bar{\boldsymbol{\kappa}}^{\mathrm{MV}}$ by Lemma~\ref{lem:w-hat-consistent}, and then the bandwidth, which by the result therefore changes only by a factor $1+o_p(1)$.

We finish by proving Proposition~\ref{prop:w-mv}
in the main text, that estimating the weights does not change inference to first order. 

\begin{proof}[Proof of Proposition~\ref{prop:w-mv}]
First, $nh\,V^{\mathrm S}(t_{\mathrm{RD}};\boldsymbol{\kappa}^{\mathrm{MV}})$ converges in probability to $\bar V^{\mathrm S}(t_{\mathrm{RD}})$ at $\bar{\boldsymbol{\kappa}}^{\mathrm{MV}}$, by Proposition~\ref{prop:common-limits}, Lemma~\ref{lem:w-mv-h} and the continuous mapping theorem. This limit is positive by Lemma~\ref{lem:nondegenerate}. Hence $V^{\mathrm S}(t_{\mathrm{RD}};\boldsymbol{\kappa}^{\mathrm{MV}})^{1/2}$ is of exact order $(nh)^{-1/2}$.

The first claim in Proposition~\ref{prop:w-mv} is given directly in  Lemma~\ref{lem:w-mv-plugin}. For the second claim, taking the difference between $V^{\mathrm S}(t_{\mathrm{RD}};\widehat{\boldsymbol{\kappa}}^{\mathrm{MV}})$ and $V^{\mathrm S}(t_{\mathrm{RD}};\boldsymbol{\kappa}^{\mathrm{MV}})$ after scaling by $nh$ and re-arranging can be simplified as
\[
  nh\bigl[V^{\mathrm S}(t_{\mathrm{RD}};\widehat{\boldsymbol{\kappa}}^{\mathrm{MV}})-V^{\mathrm S}(t_{\mathrm{RD}};\boldsymbol{\kappa}^{\mathrm{MV}})\bigr]
  =(\widehat{\boldsymbol{\kappa}}^{\mathrm{MV}}-\boldsymbol{\kappa}^{\mathrm{MV}})'\bigl[(nh\,\boldsymbol\Xi)(\widehat{\boldsymbol{\kappa}}^{\mathrm{MV}}+\boldsymbol{\kappa}^{\mathrm{MV}})-2\,nh\,\boldsymbol\xi\bigr]=o_p(1),
\]
since $nh\,\boldsymbol\Xi$ and $nh\,\boldsymbol\xi$ are $O_p(1)$ by Proposition~\ref{prop:common-limits}, $\widehat{\boldsymbol{\kappa}}^{\mathrm{MV}}$ and $\boldsymbol{\kappa}^{\mathrm{MV}}$ are $O_p(1)$ by Lemmas~\ref{lem:w-mv-h} and~\ref{lem:w-hat-consistent}, and $\widehat{\boldsymbol{\kappa}}^{\mathrm{MV}}-\boldsymbol{\kappa}^{\mathrm{MV}}=o_p(1)$. Dividing by $nh\,V^{\mathrm S}(t_{\mathrm{RD}};\boldsymbol{\kappa}^{\mathrm{MV}})$, whose limit is positive, gives the second claim.

Before turning to the third claim, we show that $T(\boldsymbol{\kappa}^{\mathrm{MV}})=O_p(1)$. As $\boldsymbol{\kappa}^{\mathrm{MV}}$ satisfies the constraints, by the assumption of Proposition~\ref{prop:w-mv} the numerator of $T(\boldsymbol{\kappa}^{\mathrm{MV}})$ is
\[
  \widehat{\att}(t_{\mathrm{RD}};\boldsymbol{\kappa}^{\mathrm{MV}})-\att(t_{\mathrm{RD}})
  =(\widehat D_{t_{\mathrm{RD}}}-D_{t_{\mathrm{RD}}})-\boldsymbol{\kappa}^{\mathrm{MV}\prime}(\widehat{\mathbf D}_0-\mathbf D_0).
\]
By Lemma~\ref{lem:dhat-order}, $\widehat D_t-D_t=O_p\bigl(h^{p+1}+(nh)^{-1/2}\bigr)$ for every period $t$, which is $O_p((nh)^{-1/2})$ as $h^{p+1}=O((nh)^{-1/2})$ by Assumption~\ref{ass:reg}(g). As $\boldsymbol{\kappa}^{\mathrm{MV}}=O_p(1)$ by Lemma~\ref{lem:w-mv-h}, the numerator is $O_p((nh)^{-1/2})$. The denominator $V^{\mathrm S}(t_{\mathrm{RD}};\boldsymbol{\kappa}^{\mathrm{MV}})^{1/2}$ is of exact order $(nh)^{-1/2}$, as shown at the start of the proof. Hence $T(\boldsymbol{\kappa}^{\mathrm{MV}})=O_p(1)$.
Next, consider the difference between $T(\widehat{\boldsymbol{\kappa}}^{\mathrm{MV}})$ and $T(\boldsymbol{\kappa}^{\mathrm{MV}})$. After adding and subtracting $\widehat{\att}(t_{\mathrm{RD}};\boldsymbol{\kappa}^{\mathrm{MV}})/V^{\mathrm{S}}(t_{\mathrm{RD}};\widehat{\boldsymbol{\kappa}}^{\mathrm{MV}})^{1/2}$, this can be written as
\[
\begin{aligned}
T(\widehat{\boldsymbol{\kappa}}^{\mathrm{MV}})-T(\boldsymbol{\kappa}^{\mathrm{MV}})
  &=\frac{\widehat{\att}(t_{\mathrm{RD}};\widehat{\boldsymbol{\kappa}}^{\mathrm{MV}})-\widehat{\att}(t_{\mathrm{RD}};\boldsymbol{\kappa}^{\mathrm{MV}})}{V^{\mathrm S}(t_{\mathrm{RD}};\boldsymbol{\kappa}^{\mathrm{MV}})^{1/2}}\biggl(\frac{V^{\mathrm S}(t_{\mathrm{RD}};\boldsymbol{\kappa}^{\mathrm{MV}})}{V^{\mathrm S}(t_{\mathrm{RD}};\widehat{\boldsymbol{\kappa}}^{\mathrm{MV}})}\biggr)^{1/2}
  \\&\quad+T(\boldsymbol{\kappa}^{\mathrm{MV}})\Biggl[\biggl(\frac{V^{\mathrm S}(t_{\mathrm{RD}};\boldsymbol{\kappa}^{\mathrm{MV}})}{V^{\mathrm S}(t_{\mathrm{RD}};\widehat{\boldsymbol{\kappa}}^{\mathrm{MV}})}\biggr)^{1/2}-1\Biggr].
\end{aligned}
\]
In the first term, the difference of the estimators divided by $V^{\mathrm S}(t_{\mathrm{RD}};\boldsymbol{\kappa}^{\mathrm{MV}})^{1/2}$ converges in probability to zero by the first claim, and the ratio of the variances converges in probability to one by the second claim, so the first term is $o_p(1)$. In the second term, the bracket is $o_p(1)$ by the second claim and we showed above that $T(\boldsymbol{\kappa}^{\mathrm{MV}})=O_p(1)$, so the second term is $o_p(1)$.
\end{proof}
\noindent
Lemmas~\ref{lem:w-hat-consistent} and~\ref{lem:w-mv-plugin} and Proposition~\ref{prop:w-mv} are stated for the conventional estimator. Analogous results can be produced for the bias-corrected estimator \eqref{eq:est_agg_bc} with its own MV weights $\widehat{\boldsymbol{\kappa}}^{\mathrm{MV,BC}}$, given by the closed form of Lemma~\ref{lem:w-mv} with the variances and covariances of the bias-corrected estimators (Appendix~\ref{app:est-var}) in place of the conventional ones. The proofs follow the same steps as above, since the bias-corrected discontinuity estimators also converge at the rate $(nh)^{-1/2}$ and their variances are also of order $1/(nh)$ \citep{calonico2014robust}. To save space, we do not spell them out here.

\bibliographystyle{plainnat}
\bibliography{bib}

\end{document}